\documentclass[aps,prd,a4paper]{revtex4-2}
\usepackage[T1,T2A]{fontenc}
\usepackage[utf8]{inputenc}
\usepackage{amsmath}
\usepackage{amssymb}
\usepackage{amsthm}
\usepackage{mathtools}
\usepackage{hyperref}
\usepackage{graphicx}
\usepackage{xparse}
\usepackage{tikz-feynman}

\newcommand{\avrg}[1]{\left\langle#1\right\rangle}

\newcommand{\bra}[1]{\left\langle#1\right\vert} 

\newcommand{\ket}[1]{\left\vert#1\right\rangle}

\DeclareDocumentCommand \bracket{m m m}{\left\langle#1\left\vert#2\right\vert#3\right\rangle}
\DeclareDocumentCommand \inprod{m m}{\left\langle#1\left\vert#2\right.\right\rangle}

\DeclareDocumentCommand \fundiff{m}{\mathcal{D}#1}

\DeclareDocumentCommand \funint{> { \SplitList { , } } m }{\int\ProcessList{#1}{\fundiff}}

\newcommand{\comm}[2]{\left[#1,#2\right]_{-}}
\newcommand{\acomm}[2]{\left[#1,#2\right]_{+}}
\newcommand{\anycomm}[2]{\left[#1,#2\right]_{\pm}}

\newcommand{\fundelta}[2]{\frac{\delta#1}{\delta#2}}

\newcommand{\abs}[1]{\left\vert#1\right\vert}
\newcommand{\Tr}{\textnormal{Tr}}

\newcommand{\Lgr}{\mathcal{L}}
\newcommand{\divp}{\text{divp}}

\newtheorem{theorem}{Theorem}[section]
\newtheorem{corollary}{Corollary}[section]

\numberwithin{equation}{section}

\begin{document}

\unitlength = 1mm

\title{Bosonization, effective action, and R-operation in a generalized Nambu-Jona-Lasinio model}
\author{Sergii Kutnii}

\begin{abstract}
Bosonization in a theory with four-fermion interaction of Nambu-Jona-Lasinio type with additional U(N) symmetry is studied. 
It is demonstrated that bosonization is not uniquely determined by the interaction terms due to Fierz identities.
Effective action including both fermions and composite fields is constructed.
R-operation renormalization scheme is developed for the effective action. 
Relation of bosonization to quantum BBGKY hierarchy is demonstrated and a generalization involving higher-order combinations of fields is proposed.
Nambu-Jona-Lasinio model with three-fermion composite fields is studied as an example.
Fierz identities for sixth-order combinations of fermions are derived.
\end{abstract}

\begin{titlepage}

\maketitle

\end{titlepage}

\section{Introduction} \label{intro}
In the middle of last century, Bogoliubov and de Gennes developed a method of treating fermion interaction in condensed matter physics.
Given an interacting Pauli fermion Hamiltonian such as

\begin{equation}
\widehat{H} = \sum\limits_\sigma\int\psi^+_\sigma(\bold{r})\frac{\widehat{\mathbf{p}}^2}{2m}\psi(\bold{r})_\sigma\,d^3\bold{r} 
+ g\int\psi(\bold{r})^+_\uparrow\psi(\bold{r})^+_\downarrow\psi(\bold{r})_\downarrow\psi(\bold{r})_\uparrow\,d^3\bold{r}, \label{intro/H4Ferm}
\end{equation}

one has to build a mean field Hamiltonian defined by
\begin{eqnarray}
\widehat{H}_{MF}&=&\sum\limits_\sigma\int\psi^+_\sigma(\bold{r})\frac{\widehat{\bold{p}}^2}{2m}\psi(\bold{r})_\sigma\,d^3\bold{r} 
+ \int{}\left[\psi(\bold{r})^+_\uparrow\psi(\bold{r})^+_\downarrow\Delta(\bold{r}) + \Delta^{+}(\bold{r})\psi(\bold{r})_\downarrow\psi(\bold{r})_\uparrow\right]d^3\bold{r} \label{intro/HBdG}\\
\Delta(\bold{r})&=&\avrg{\psi(\bold{r})_\downarrow\psi(\bold{r})_\uparrow},\label{intro/BdGgapeqn}
\end{eqnarray}
where $\avrg{}$ means the ground state average.

The mean field Hamiltonian (\ref{intro/HBdG}) is quadratic and can be diagonalized by intoducing Bogoliubov quasiparticles, 
and (\ref{intro/BdGgapeqn}) will give the mass gap of the theory in the homogeneous case. The inhomogeneous case of great interest and widely studied as well.

This seemingly ad hoc recipe allows a path integral reformulation. 
Here I follow \cite{Svidz}. The statistical sum of the theory (\ref{intro/H4Ferm}) can be represented as a path integral
\begin{eqnarray}
Z[\tau] &=& Sp\exp\left[-\beta\widehat{H}\right] \sim\nonumber\\ 
&\sim&\funint{\zeta^+,\zeta}Sp\exp\left\lbrace -\beta\sum\limits_\sigma\int\psi^+_\sigma(\bold{r})\frac{\widehat{\bold{p}}^2}{2m}\psi(\bold{r})_\sigma\,d^3\bold{r}
- \right.\nonumber\\
&-&\left. g\beta\int{}\left[\psi(\bold{r})^+_\uparrow\psi(\bold{r})^+_\downarrow\zeta(\bold{r}) + \zeta^{+}(\bold{r})\psi(\bold{r})_\downarrow\psi(\bold{r})_\uparrow\right]d^3\bold{r}
+ g\beta\int\zeta^+(\bold{r})\zeta(\bold{r})d^3\bold{r}
\right\rbrace
\end{eqnarray}
Here $\sim$ means equality up to a constant multiplier.
Thus
\begin{eqnarray}
Z[\tau] &=& \funint{\zeta^+,\zeta}\exp\left\lbrace-\beta\Omega\left[\zeta^+,\zeta\right]\right\rbrace\nonumber\\
\exp\left\lbrace-\beta\Omega\left[\zeta^+,\zeta\right]\right\rbrace &\equiv&
Sp\exp\left\lbrace -\beta\sum\limits_\sigma\int\psi^+_\sigma(\bold{r})\frac{\widehat{\bold{p}}^2}{2m}\psi(\bold{r})_\sigma\,d^3\bold{r}
- \right.\nonumber\\
&-&\left. g\beta\int{}\left[\psi(\bold{r})^+_\uparrow\psi(\bold{r})^+_\downarrow\zeta(\bold{r}) + \zeta^{+}(\bold{r})\psi(\bold{r})_\downarrow\psi(\bold{r})_\uparrow\right]d^3\bold{r}
+ g\beta\int\zeta^+(\bold{r})\zeta(\bold{r})d^3\bold{r}
\right\rbrace \label{intro/BogoliubovEffAction}
\end{eqnarray}
It is easy to see that the approximation (\ref{intro/HBdG},\ref{intro/BdGgapeqn}) is just the semiclassical contribution to (\ref{intro/BogoliubovEffAction}) 
defined by 
\begin{eqnarray}
\frac{\delta\Omega}{\delta\zeta} = 0\nonumber\\
\frac{\delta\Omega}{\delta\zeta^+} = 0
\end{eqnarray}
The path integral formulation above elucidates the nature of Bogoliubov-de Gennes method and its place within quantum field theory somewhat.
It thus becomes tempting to do the same in a theory with interacting relativistic fermions. 
However, it turns out that the answer to the question ``what is the analog of Bogoliubov-de Gennes method in a relativistic quantum field theory?''  
is surprisingly complex and requires rethinking of some fundamental concepts. In particular, some of the lessons learned by looking for an answer to it are essentially the same as those of QFT in curved spacetime: that the concepts of ``vacuum'' and ``particles'' are essentially relative and that Wightman axioms are thus too restrictive to provide a good foundation for QFT \cite{HolWald}.
This paper was started as an attempt to answer the aforementioned question with at least field theoretical if not mathematical level of rigour.
It is treated as a merely formal question though. The Nambu-Jona-Lasinio type model of interacting Dirac fermions will be the primary testbed.

The original Nambu-Jona-Lasinio model proposed in \cite{NJL} is a direct relativistic analog of the BCS Hamiltonian of superconductivity:
\begin{equation}
S_{NJL} = \int\,d^4x\left\lbrace\bar{\psi}(x)i\widehat{\partial}\psi(x) 
+ \lambda\left[\bar{\psi}(x)\psi(x)\bar{\psi}(x)\psi(x) - \bar{\psi}(x)\gamma^5\psi(x)\bar{\psi}(x)\gamma^5\psi(x)\right]\right\rbrace.
\label{intro/NJL}
\end{equation}
In the same seminal paper, Nambu and Jona-Lasinio demonstrated that this model had a mass gap. 
Numerous generalizations were proposed since then. 

Let
\begin{equation}
\begin{split}
&\Gamma_1 = 1\\
&\Gamma^\mu_2 = \gamma^\mu\\
&\Gamma^{\mu\nu}_3 = \sigma^{\mu\nu}\\
&\Gamma^{\mu}_4 = \gamma^5\gamma^\mu\\
&\Gamma_5 = \gamma^5\\
&\tau_{i0}^{\bar{\alpha}} = \Gamma_i^{\bar{\alpha}}\\
&\tau_{i1}^{\bar{\alpha}} = T^a\Gamma_i^{\bar{\beta}},
\end{split} \label{intro/gammabasis}
\end{equation}
where $T^a$ are the generators of $su(N)$ algebra.
Shorthand notation $\tau^{\bar{\alpha}}_{\bar{i}}$ will also be used, 
where Latin multi-index $\bar{i}$ labels transformation laws and Greek $\bar{\alpha}$ stands for combined Lorentz and $su(N)$ generator indices.
$su(N)$ color indexes will be denoted by subscript $i,j,k,l\ldots$. Summation over repeated indices is implied.
The fermion Lagrangians having an extra global $U(N)$ symmetry 
with interactions of the form 
\begin{equation}
\mathcal{L}_\text{int} = \sum\limits_{a,\alpha}G_{\bar{i}}\bar{\psi}\tau_{\bar{i}\bar{\alpha}}\psi\bar{\psi}\tau_{\bar{i}}^{\bar{\alpha}}\psi \label{intro/interaction-terms}
\end{equation}
where the choice of constants respects Lorentz symmetry, will be studied here.

Bosonization, an obvious analog of the transformation (\ref{intro/BogoliubovEffAction}) for the NJL-like model of interacting fermions, was proposed by Kikkawa \cite{Kikkawa1974}, Ebert, and Reinhardt \cite{Ebert1978}, \cite{Reinhardt1978}. The bosonized NJL model has found a wide range of applications from meson phenomenology \cite{Volkov1984} to a tentative unified theory of electroweak and strong interactions \cite{Terazawa1977} . The idea is to utilize the relation
\begin{equation}
	\exp\left[\lambda{}\int\,A^2(x)d^dx\right] \sim \funint{B}\exp\left\{\int\left[\frac{1}{4\lambda}B(x)^2 + A(x)B(x)\right]d^dx\right\},
\end{equation}
a functional analog of Fourier or Laplace transform, to represent a generic four-fermion interaction term as
\begin{equation}
	\exp\left[-i\int{}G_{\bar{i}}\bar{\psi}\tau_{\bar{i}\bar{\alpha}}\psi\bar{\psi}\tau_{\bar{i}}^{\bar{\alpha}}\psi{}d^dx\right] \sim 
	\funint{\Sigma}\exp\left\{-i\int\left[-\frac{1}{4G_{\bar{i}}}\Sigma_{\bar{i}\bar{\alpha}}\Sigma_{\bar{i}}^{\bar{\alpha}} 
	+\Sigma_{\bar{i}\bar{\alpha}}\bar{\psi}\tau_{\bar{i}}^{\bar{\alpha}}\psi \right]d^dx\right\}.
\end{equation}
The NJL generating functional can thus be expressed as
\begin{equation}
\begin{split}
Z\left[\bar{\eta}\eta\right] &\equiv \funint{\bar{\psi}, \psi}\exp\left\{-i\int\left[\bar{\psi}\widehat{\partial}\psi 
	+ G_{\bar{i}}\bar{\psi}\tau_{\bar{i}\bar{\alpha}}\psi\bar{\psi}\tau_{\bar{i}}^{\bar{\alpha}}\psi 
	+ \bar{\eta}\psi + \bar{\psi}\eta\right]d^4x\right\} \sim \\
	&\sim \funint{\bar{\psi}, \psi, \Sigma}\exp\left\{-i\int\left[\bar{\psi}\left(\widehat{\partial} + \widehat{\Sigma}\right)\psi 
	- \frac{1}{4G_{\bar{i}}}\Sigma_{\bar{i}\bar{\alpha}}\Sigma_{\bar{i}}^{\bar{\alpha}}  + \bar{\eta}\psi + \bar{\psi}\eta\right]d^4x\right\} 
\end{split}\label{intro/Bosonization}
\end{equation}
	where $\widehat{\Sigma} = \tau_{\bar{i}\bar{\alpha}}\Sigma_{\bar{i}}^{\bar{\alpha}}$.
	
By taking the functional derivative with respect to $\Sigma$ one can find that
\begin{equation}
0 = \frac{\delta{\ln{Z}}}{\delta{\Sigma_{\bar{i}\bar{\alpha}}}(x)} =
-i \avrg{\bar{\psi}(x)\tau_{\bar{i}}^{\bar{\alpha}}\psi (x)} 
+ \frac{i}{2G_{\bar{i}}}\avrg{\Sigma_{\bar{i}}^{\bar{\alpha}}(x)},
\end{equation}
where angular brackets denote ground state average. The relation
\begin{equation}
	\avrg{\Sigma_{\bar{i}}^{\bar{\alpha}}(x)} =  2G_{\bar{i}}\avrg{\bar{\psi}(x)\tau_{\bar{i}}^{\bar{\alpha}}\psi (x)} \label{intro/gapeqn}
\end{equation}
is the gap equation. 
The functional integration over $\bar{\psi}, \psi$ in (\ref{intro/Bosonization})  can be carried out explicitly to yield
\begin{equation}
\begin{split}
	Z\left[\bar{\eta}\eta\right]  &\sim \funint{\Sigma}\exp\left\{i\Omega\left[\Sigma, \bar{\eta}, \eta\right] \right\}.\\
	\Omega\left[\Sigma, \bar{\eta}, \eta\right] &= -i\Tr\ln\left[\widehat{\partial} + \widehat{\Sigma} \right] 
	+ \iint\bar{\eta}(x)\left[\widehat{\partial} + \widehat{\Sigma} \right]^{-1}(x,y)\eta(y)d^4xd^4y + \frac{1}{4G_{\bar{i}}}\int\Sigma_{\bar{i}\bar{\alpha}}\Sigma_{\bar{i}}^{\bar{\alpha}}d^4x
	\label{intro/integrated-action}
\end{split}
\end{equation}
It is worth noting that fermions are still present here through the sources $\bar{\eta}, \eta$. Now, the saddle point approximation may be taken:
\begin{eqnarray}
	&&Z\left[\bar{\eta},\eta\right]  = \exp(i\Omega[\Sigma_{c}]]) \nonumber\\
	&&\left.\frac{\delta\Omega}{\delta\Sigma}\right\vert_{\Sigma = \Sigma_c} = 0
\end{eqnarray}
In this approximation the gap equation in the absence of sources becomes
\begin{equation}
	\Sigma_{c\,\bar{i}}^{\bar{\alpha}} = 2G_{\bar{i}}Sp\tau_{\bar{i}}^{\bar{\alpha}}\left[\widehat{\partial} + \widehat{\Sigma}_c\right]^{-1}(x,x).
\end{equation}
In full analogy with superconductivity, it can be shown to possess nontrivial solutions provided that there is a renormalization scheme to deal with the divergencies. It is also obvious that putting $\Sigma = \Sigma_c$ in the bosonized path integral (\ref{intro/Bosonization}) gives the exact equivalent of Bogoliubov-de Gennes Hamiltonian (\ref{intro/HBdG}).

However, bosonization still looks no less an ad hoc recipe than Bogoliubov-de Gennes method. Its place within the framework of quantum field theory is not quite clear.
It is the expression \eqref{intro/integrated-action} that gives a hint, bearing a remarkable similarity to Jackiw's result for quantum effective action computed up to one loop \cite{Jackiw1974}.
It thus becomes tempting to interpret bosonization as a version of effective action formalism. To do so, one needs source terms for the bosonic fields $\Sigma_{\bar{i}\bar{\alpha}}$ as well. 
They can be added to the fermionic generating functional coupled to quadratic combinations of fields. The actual generating functional that should be considered is
\begin{equation}
\begin{split}
Z\left[\left\{J\right\}\right] &= 
\funint{\bar{\psi},\psi}\exp\left\{i\int\left[\bar{\psi}i\widehat{\partial}\psi + 
G_i\bar{\psi}\Gamma_{i\,\bar{\alpha}}\psi\bar{\psi}\Gamma^{\bar{\alpha}}_i\psi
+ \bar{J}_1\psi +\bar{\psi}J_1 
+ 2\lambda_{2\bar{i}}J_{2\bar{i}\bar{\alpha}}\bar{\psi}\tau_{\bar{i}}^{\bar{\alpha}}\psi
+ \lambda_{2\bar{i}}J_{2\bar{i}\bar{\alpha}}J^{\bar{\alpha}}_{2\bar{i}}
\right]d^4x\right\}.
\end{split} \label{intro/genfunc}
\end{equation}
After a trivially modified bosonization, $J^{\bar{\alpha}}_{2\bar{i}}(x)$ become sources for bosonic fields.

The generating functional \eqref{intro/genfunc} resembles the composite fields generating functional proposed by Cornwall, Jackiw, and Tomboulis \cite{Cornwall1974}, save for the fact that same-point products are introduced instead of two-point ones. The Cornwall-Jackiw-Tomboulis generating functional would produce an analog of \eqref{intro/genfunc} under a constraint
\begin{equation}
J(x,y) = \tilde{J}(x)\delta(x-y) \label{intro/same-point-restriction}
\end{equation}
on the sources. Cornwall-Jackiw-Tomboulis effective action was considered for NJL-like models by Gusynin, Miransky, and Kushnir \cite{Gusynin1989a}, \cite{Gusynin1989b}, \cite{Gusynin1991}, \cite{Gusynin1992}.
It is worth noting that although that \eqref{intro/same-point-restriction} defines a special case, that special case results in a local field theory after bosonization, allowing thus to compute the resulting effective action to any order by Jackiw's shift in integration variables.

The link to Cornwall-Jackiw-Tomboulis formalism puts bosonization on even firmer ground but there still remains a question of interpretation. It is clear that the formalism can be extended further by introduction of many-point source terms. Correlators of the theory become dynamical variables in this formalism. Is their dynamics physical? If so, what is its precise physical meaning? A 1988 paper by Calzetta and Hu \cite{Calzetta1988} sheds light on this question. In this work, Schwinger-Keldysh or in-in path integral of the form 
\begin{equation}
Z\left[J_a,K_{ab},K_{abc},\ldots\right] = \funint{\phi^a}\exp\left\{i\left[S[\phi^a] + J_a\phi^a + \frac{1}{2}K_{ab}\phi^a\phi^b+\frac{1}{6}K_{abc}\phi^a\phi^b\phi^c+\ldots\right]\right\}
\label{intro/Schwinger-Keldysh}
\end{equation}
is proposed, where Latin indices span two values $+,-$ meaning forward- and backward-running temporal evolution, and source terms arise naturally from the logarithm of the initial density matrix. 
As pointed out by the authors, the resulting infinite hierarchy of equations on correlators is analogous to the BBGKY hierarchy of classical statistical physics.

Thus, quantum state information can enter the generating functional through nonzero multi-point sources or, after transformations, through their multilocal field counterparts. This allows to reformulate the notion of a quantum state without any reference to mathematical constructs inherited from nonrelativistic quantum mechanics such as Hilbert space vectors or density matrices by postulating that a state is defined by averages over it -- no more, no less. Anastopoulos, Hu, and Savvidou add a model of relativistic quantum measurement to the formalism \cite{Anastopoulos2023}, 
arriving at a generating functional having essentially the form \eqref{intro/Schwinger-Keldysh}, where both measurement and state preparation are represented by source terms under the exponent. This, in principle, allows to formulate quantum field theory in an entirely functional language.

Bosonization fits well into this picture. However, there remain two major issues to be addressed here.

Bosonization procedure is in fact ambiguous due to Fierz identities \cite{Fierz1937}, \cite[161--162]{itzykson2012}. As quartic combinations of fermionic fields are not linearly independent,
any quartic fermionic interaction can be rewritten as a sum of all possible terms of the form \eqref{intro/interaction-terms} with certain freedom to choose the coefficients. Thus, it is possible to introduce not just the $\sigma, \pi$ fields as it is usually done but all the fields corresponding to $\bar{\psi}\tau_{\bar{i}\bar{\alpha}}\psi$. The problem is, these fields will have essentially arbitrary masses with some weak constraints only. As a consequence, bosonization cannot give any reliable predictions if divergent diagrams of the model are merely regularized, as physically meaningful quantities such as meson masses will still depend on arbitrarily chosen parameters - a proper renormalization is needed. If renormalization is considered, such ambiguity may, in fact, be beneficial, allowing to cancel some divergences.

Renormalization is the second and main issue. As discussed above, a renormalization scheme is required to extract any meaningful predictions from the model. However, NJL model is not renormalizable in the ordinary sense, as it is easy to see that the interaction constants $G_{\bar{i}}$ are dimensional. Thus, it is not clear why usual renormalization schemes appropriate for renormalizable models should work. In fact, it is easy to see that the power series expansion of $Tr\ln\left[\widehat{\partial} + \widehat{\Sigma}\right]$ contains divergent terms of 3-rd and 4-th order in $\Sigma$. Compensating those divergencies would require counterterms altering not only the constants but form of the lagrangian. In general, it is unclear whether standard renormalization schemes remain valid when extra variables describing a nontrivial quantum state are added. Therefore, it is necessary to apply a more mathematically rigorous renormalization scheme, i.e. Bogoliubov-Parasiuk-Hepp-Zimmermann (BPHZ) theory of R-operation as it is suggested by Jackiw in \cite{Jackiw1974}. However, this is not a trivial task since R-operation was formulated within the framework of operator quantum field theory. Therefore, to extend BPHZ scheme to the functional picture outlined above, it is necessary to discuss the link between path integral and operator formulations of quantum field theory.

The aim of this work is twofold. The first is to do NJL bosonization right - that is by constructing an effective action including both bosons and fermions and applying BPHZ renormalization to it, up to one loop. The second, and more utilitarian, goal is to compute certain quantities arising in this scheme to provide future researchers with a reference on the math. The rest of the paper is organized as follows: the section \ref{formalism} is dedicated to revisiting the concepts of effective action and R-operation and finding a way to combine them; ambiguity of bosonization is studied and effective action for the bosonized NJL model is constructed in the section \ref{NJL-action}; a generalization of the bosonization framework is proposed in the section \ref{compfields}; and in the conclusions implications of the entire inquiry are discussed and directions of further investigation outlined.

While Schwinger-Keldysh formalism would be needed for most of the physics, this paper only deals with the relations that have similar forms in the in-in and in-out formalisms. Therefore, a simpler in-out effective action is studied. Extending the results to the in-in formalism is a trivial exercise and is left to the reader.

\section{Effective action formalism and R-operation renormalization} \label{formalism}

The definition of the effective action for a quantum field theory starts usually with a path integral (see e.g. \cite{Weinberg1996}). Jackiw in \cite{Jackiw1974} gave a recipe for computing the effective action and suggested R-operation approach to its renormalization. However, there is a major conceptual obstacle to implementing this program: R-operation's classic formulation (Zavyalov's \cite{Zavyalov1990} book can be used as a comprehensive treatise on the matter) was done within axiomatic S-matrix theory and relies on Wightman's axioms \cite[96--106]{StreaterWightman1964}. 
R-operation formalism makes use of Wightman's assumption of translational invariance of vacuum to deduce the form of counterterms (more on this below).
But effective action depends on the average field $\avrg{\phi(x)}$ that needs not to be constant. The quantum-corrected field equations arising from the effective action may have nontrivial solutions breaking translational invariance in the same manner as nontrivial spacetime metric does. To make things worse, particle pictures corresponding to different vacuum solutions will in general be inequivalent. As an example, consider an asymmetric Higgs Lagrangian with potential having minima at $\varphi_1, -\varphi_2$:
\begin{equation}
	\mathcal{L} = \frac{1}{2}\partial_\mu\phi\partial^\mu\phi  + g^2\phi^2\left[\frac{\varphi_1\varphi_2}{2} + \frac{\varphi_1 - \varphi_2}{3}\phi - \frac{\phi^2}{4}\right].
\end{equation}
The zero order effective action is just the classical action. The corresponding field equation will have two nontrivial vacuum solutions $\phi = \varphi_1, -\varphi_2$. If particles are to be understood as excitations over a vacuum, one is left with two possible particle decompositions
\begin{eqnarray}
	\phi &=&  \varphi_1 + \chi_1 = \varphi_1 + \sum\limits_\mathbf{k}\left[f_{1\mathbf{k}}(x)\alpha_{1\mathbf{k}} + f^*_{1\mathbf{k}}(x)\alpha^+_{1\mathbf{k}}\right] \nonumber\\
	\phi &=& -\varphi_2 + \chi_2 = -\varphi_2 + \sum\limits_\mathbf{k}\left[f_{2\mathbf{k}}(x)\alpha_{2\mathbf{k}} + f^*_{2\mathbf{k}}(x)\alpha^+_{2\mathbf{k}}\right] 
\end{eqnarray}
corresponding to two possible free-field approximations for the Hamiltonian
\begin{eqnarray}
	H \approx C_1 + \frac{1}{2}\int\left[\pi_1^2 - \mathbf{\nabla}\chi_1^2 + g^2\varphi_1\left(\varphi_1 + \varphi_2\right)\chi_1^2\right]dV \nonumber\\
	H \approx C_2 + \frac{1}{2}\int\left[\pi_2^2 - \mathbf{\nabla}\chi_2^2 + g^2\varphi_2\left(\varphi_1 + \varphi_2\right)\chi_2^2\right]dV 
\end{eqnarray}
As it was pointed out by Haag \cite{Haag1955}, the two particle pictures are inequivalent. Furthermore, in the same work Haag claimed to demonstrate the non-existence of the interaction picture of QFT. Strictly speaking, however, it was incompatibility of the interaction formalism with Wightman axioms that Haag really proved.
Thus, Wightman axioms cannot provide a good basis for QFT. The difficulties stem from formulating the theory in the Hilbert space of N-particle states from the onset. The notions of ``vacuum'' and ``particles'' must be relativized if not totally abandoned. Therefore, to define R-operation for the effective action, it is necessary to start from the foundations of QFT.

Only flat Minkowski spacetime will be considered here. The metric signature is $(+, -, -, -)$. The dynamic variables of the theory are operator-valued distributions $\pi, \phi$,  i.e. integrals 
\begin{eqnarray}
	\pi\left[f\right] = \int\pi(x)f(x)d^4x\nonumber\\
	\phi\left[f\right] = \int\phi(x)f(x)d^4x \label{formalism/deffields}
\end{eqnarray}
must be well-defined operators on some Hilbert space $\mathcal{H}$ for a suitable class of test functions, such as those with compact support, even though $\pi(x), \phi(x)$ may be singular at some points. Coordinate derivatives of the variables are defined in the sense of integration by parts. Multiplication can be defined recursively. If $\mathcal{O}_1(x), \mathcal{O}_2(x)$ are two local monomials in $\pi(x), \phi(x)$ and their derivatives (such expressions will be called local field monomials from now on), such that integrals
\begin{eqnarray}
	\mathcal{O}_1\left[f\right] = \int\mathcal{O}_1(x)f(x)d^4x\nonumber\\
	\mathcal{O}_2\left[f\right] = \int\mathcal{O}_2(x)f(x)d^4x
\end{eqnarray}
exist, then the integrals
\begin{eqnarray}
	\mathcal{O}_{12}\left[F\right] = \iint\mathcal{O}_1(x_1)\mathcal{O}_2(x_2)F(x_1,x_2)d^4x_1d^4x_2\nonumber\\
	\mathcal{O}_{21}\left[F\right] = \iint\mathcal{O}_2(x_1)\mathcal{O}_1(x_2)F(x_1,x_2)d^4x_1d^4x_2
\end{eqnarray}
are naturally defined for functions $F$ supported outside the hyperplane $x_1 = x_2$. The definition can be extended to functions with any compact support, but such an extension is ambiguous. If two multiplication operations $*, *^\prime$ are defined, then
\begin{equation}
	\mathcal{O}_1(x)*^\prime\mathcal{O}_2(y) -  \mathcal{O}_1(x)*\mathcal{O}_2(y) = \sum\limits_k\mathcal{Q}_k(x)P_k(x, \partial_{x^\mu})\delta(x-y) \label{formalism/defambiguity}
\end{equation}
where $\mathcal{Q}_k(x)$ are local field monomials, and 
\begin{equation}
	P_k(x, \partial_{x^\mu}) = \sum\limits_{n = 0}^{N_k}p_k^{\mu_1...\mu_n}(x)\partial_{x^{\mu_1}}...\partial_{x^{\mu_n}}.
\end{equation}
If translational symmetry is required, the coefficients $p_k^{\mu_1...\mu_n}(x)$ must be constant. The idea of R-operation renormalization is to redefine the multiplication operation in such a way that would eliminate ultraviolet divergencies in Feynman graphs. The restrictions on operator multiplication  are the requirements of associativity, Lorentz covariance, and canonical (anti)commutation relations
\begin{eqnarray}
	\anycomm{\phi(x)}{\phi(y)} &=& 0, \left\Vert{x - y}\right\Vert \leq 0\nonumber\\
	\anycomm{\pi(x)}{\pi(y)} &=& 0,  \left\Vert{x - y}\right\Vert \leq 0\nonumber\\
	\anycomm{\phi(y)}{\pi(x)} &=& i\delta(x - y). \label{formalism/commutators}
\end{eqnarray}
Time evolution of fields is defined by a Hamiltonian
\begin{equation}
	\begin{aligned}
		\mathcal{O}_t &= i\comm{H}{\mathcal{O}}\\
		H(t) &= \int\,d^3\mathbf{x}\sum\limits_kh_k(t,\mathbf{x})\mathcal{O}_k(t, \mathbf{x})\\
		H^+(t) &= H(t)
	\end{aligned}
\end{equation}
where $\mathcal{O}_k$ are local field monomials, and $h_k$ can, generally speaking, depend on spacetime point. Having the Hamiltonian, it is possible to define the evolution operator as
\begin{equation}
	\begin{aligned}
		U(t_1, t_2) &= T\exp\int\left[i\int\limits_{t_1}^{t_2}\,H(t)dt\right] \equiv \lim\limits_{\max\left\{
		\Delta{t}_k\right\} \rightarrow 0} T\prod\limits_{t_1}^{t_2} \left[1 + iH(t_k)\Delta{t_k}\right]\\
		\mathcal{O}(t_2) &= U(t_1, t_2)\mathcal{O}(t_1)U^+(t_1, t_2),
	\end{aligned} \label{formalism/evolution}
\end{equation}
where $T$ means temporal ordering of operator products. 
Let  the relations $\prec, \succ, \sim$ between spacetime points be defined as
\begin{equation}
	\begin{aligned}
		x&\succ {y} \Leftrightarrow \text{$x$ is in the future light cone of $y$};\\
		x& \prec {y} \Leftrightarrow \text{$x$ is in the past light cone of $y$};\\
		x& \sim y \Leftrightarrow \text{$x$ and $y$ are separated by a spacelike interval}.
	\end{aligned}
\end{equation}
Then for two operators
\begin{equation}
	T\mathcal{O}_1(x)\mathcal{O}_2(y) = 
	\begin{cases}
		\mathcal{O}_1(x)\mathcal{O}_2(y) & x = y\\
		\mathcal{O}_1(x)\mathcal{O}_2(y) & x \succsim y \\
		\mathcal{O}_2(y)\mathcal{O}_1(x) & \text{otherwise}.
	\end{cases}
\end{equation}
and for any operator product
\begin{equation}
	T\prod_k\mathcal{O}_k(x_k) = \left[T\prod_l\mathcal{O}_l(x_l)\right]\left[T\prod_{m}\mathcal{O}_m(x_m)\right],
\end{equation}
provided that the index set $\{k\}$ can be split into two non-intersecting subsets $\{l\}, \{m\}$ such that $\forall l \in \{l\}, m \in \{m\} : x_l \succ x_m$. i.e. simply speaking all future operators must stand to the left of past ones.
Now, it is assumed that the Hilbert space $\mathcal{H}$ on which field operators are defined, allows $\phi$-representation i.e. the eigenvectors $\ket{\varphi, t}$ such that
\begin{equation}
	\phi(\mathbf{x}, t)\ket{\varphi,t} = \varphi(\mathbf{x})\ket{\varphi,t}
\end{equation} 
form a complete basis. This, of course, implies that the space $\mathcal{H}$ is non-separable. The space can be constructed e.g. in a straightforward generalization of  quantum mechanics. Let $\varphi(\mathbf{x})$ be a distribution defined on 3-dimensional Euclidean space. Then $\pi, \phi$ can be naturally represented in the algebra of functional-differential operators
\begin{equation}
	\mathcal{O} = \sum\limits_{mn}\int\ldots\int\mathcal{O}_{mn}(x_1,\ldots,x_m,y_1,\ldots,y_n)\varphi(x_1)...\varphi(x_m)\frac{\delta}{\delta\varphi(y_1)}\ldots\frac{\delta}{\delta\varphi(y_n)}
\end{equation}  
acting on arbitrary functionals of $\varphi$. Let at $t = t_0$ 
\begin{equation}
	\begin{aligned}
		\phi(t_0, \mathbf{x}) &= \varphi(\mathbf{x})\\
		\pi(t_0, \mathbf{x}) &= -i\frac{\delta}{\delta\varphi(\mathbf{x})}.
	\end{aligned}
	\label{formalism/simpleoperators}
\end{equation}
It is easy to see that the canonical (anti)commutation relations (\ref{formalism/commutators}) will then hold automatically for $t = t_0$, and at arbitrary time the algebra can be constructed using the evolution operator (\ref{formalism/evolution}). The state $\ket{\varphi, t_0}$ is then the ``delta-functional''
\begin{equation}
	\ket{\varphi, t_0} = \Delta\left(\phi - \varphi\right) = \prod_{\mathbf{x}}\delta\left[\phi(\mathbf{x}) - \varphi(\mathbf{x})\right],
\end{equation}
and the momentum eigenstate is
\begin{equation}
	\begin{aligned}
		&\pi(\mathbf{x}, t_0)\ket{\varkappa, t_0} = \varkappa(\mathbf{x})\ket{\varkappa, t_0}\\
		&\ket{\varkappa, t_0} = \int\fundiff{\varphi^\prime}\exp\left[i\int\,d^3\mathbf{x}\varphi^\prime(\mathbf{x})\varkappa(\mathbf{x})\right]\ket{\varphi^\prime,t_0}
	\end{aligned}
\label{formalism/eigenmomentum}
\end{equation}
Thus
\begin{equation}
	\left\langle\varphi,t_0\vert\varkappa,t_0\right\rangle = \exp\left[i\int\,d^3\mathbf{x}\varphi(\mathbf{x})\varkappa(\mathbf{x})\right].
\end{equation}
Now it is possible to define the generating functional of time-ordered products as
\begin{equation}
	Z\left[t_\text{in}, t_\text{out};\varphi_\text{in},\varphi_\text{out};\eta_\phi\right] = 
	\bracket{\varphi_\text{out}, t_\text{out}}{T\exp\left[i\int\,\eta_\phi(x)\phi(x)d^4x\right]}{\varphi_\text{in}, t_\text{in}}. \label{formalizm/defZ}
\end{equation}
This functional encodes all the information about any time-ordered product of $\phi(x)$:
\begin{equation}
	\bracket{\varphi_\text{out}, t_\text{out}}{
		T\phi(x_1)\ldots\phi(x_n)
	}{\varphi_\text{in}, t_\text{in}} = 
	\left.\frac{(-i)^n}{Z\left[t_\text{in}, t_\text{out};\varphi_\text{in},\varphi_\text{out};\eta_\phi\right]}\frac{\delta}{\delta\eta_\phi(x_1)}\ldots\frac{\delta}{\delta\eta_\phi(x_n)}Z\left[t_\text{in}, t_\text{out};\varphi_\text{in},\varphi_\text{out};\eta_\phi\right]\right\vert_{\eta_\phi = 0}.
\end{equation}
$\pi$-dependent time-ordered products can be computed either from the equations of motion or by adding the source $\eta_\pi$ for $\pi$ which will not alter the formulae much. The connected generating functional
is defined as
\begin{equation}
	W\left[t_\text{in}, t_\text{out};\varphi_\text{in},\varphi_\text{out};\eta_\phi\right] = -i\ln{Z\left[t_\text{in}, t_\text{out};\varphi_\text{in},\varphi_\text{out};\eta_\phi\right]}, \label{formalism/defW}
\end{equation}
the effective field as
\begin{equation}
	\varphi(x) = \frac{\delta{W}}{\delta\eta_\phi(x)}, \label{formalism/deffield}
\end{equation}
and the effective action as
\begin{equation}
	\Gamma[\varphi] = W[\eta_\varphi] - \int\eta_\varphi(x)\varphi(x)d^4x, \label{formalism/defgamma}
\end{equation}
where $\eta_\varphi$ is understood as the result of inverting the effective field definition (\ref{formalism/deffield}). Then
\begin{equation}
	\fundelta{\Gamma}{\varphi(x)} = \int\fundelta{\eta(y)}{\varphi(x)}\fundelta{W}{\eta(y)}d^4y - \eta(x) 
	- \int\fundelta{\eta(y)}{\varphi(x)}\varphi(y)d^4y = -\eta(x) \label{formalism/fieldequation}
\end{equation}
Since it must be, by definition of $\varphi$, that
\begin{equation}
	\varphi(t_{\text{in},\text{out}}, \mathbf{x}) = \varphi_{\text{in},\text{out}}(\mathbf{x}),
\end{equation}
the initial and final states are encoded in $\Gamma\left[\varphi\right]$ through its dependence on boundary conditions. 

All this is quite standard, however, it was important to demonstrate that the effective action could be introduced without an explicit reference to the notions of vacuum and particles. 

The path integral expression for $Z$ can be obtained as follows:
\begin{equation}
	\begin{split}
		Z\left[t_\text{in}, t_\text{out};\varphi_\text{in},\varphi_\text{out};\eta_\phi\right] &= 
		\sum\limits_{\{\varphi_k\}}\prod_{t_k = t_\text{in}}^{t_\text{out}}\bracket{\varphi_{k+1},t_{k+1}}{T\exp\left[i\int\limits_{t_k}^{t_{k+1}}\eta_\phi(x)\phi(x)d^4x\right]}{\varphi_{k},t_k}.
	\end{split}
\end{equation} 
Then, one can notice that 
\begin{equation}
	\left.
	\begin{aligned}
		\mathcal{O}(t^\prime) &= U(t,t^\prime)\mathcal{O}(t)U^+(t,t^\prime)\\
		\mathcal{O}(t)\ket{\lambda,t} &= \lambda\ket{\lambda,t}
	\end{aligned}
	\right\rbrace \Rightarrow \mathcal{O}(t^\prime)U(t,t^\prime)\ket{\lambda,t} = \lambda U(t,t^\prime)\ket{\lambda,t}.
	\label{formalism/isospectral}
\end{equation}
Therefore
\begin{equation}
	\bra{\varphi_{k + 1}, t_{k+1}} = \bra{\varphi_{k + 1}, t_{k}}U^+(t_{k},t_{k+1}) 
\end{equation}
and, if all $\Delta{t}_k \equiv t_{k+1} - t_{k}$ are small, then
\begin{equation}
	\begin{split}
		&Z\left[t_\text{in}, t_\text{out};\varphi_\text{in},\varphi_\text{out};\eta_\phi\right] =\\ 
			&=\sum\limits_{\{\varphi_k\}}\prod_{t_k = t_\text{in}}^{t_\text{out}}\bracket{\varphi_{k+1}}{\exp\left[-iH(t_k)\Delta{t}_k + i\int\eta_\phi(x)\phi(x)d^3\mathbf{x}\Delta{t}_k\right]}{\varphi_{k}} + O\left(\max\{\abs{\Delta{t}_k}\}\right).
	\end{split}
\end{equation}
Now one can use the isospectral equation \eqref{formalism/isospectral} and time evolution once again to translate $t_k \rightarrow t_0$, where the  operators have simple form \eqref{formalism/simpleoperators}. Then the matrix element can be computed with ease.
The canonical commutation relations \eqref{formalism/commutators} can be used to move all $\pi$ to the left of $\phi$ in the Hamiltonian
\begin{equation}
	\begin{split}
	&\bracket{\varphi_{k+1}}{\exp\left[-iH(t_k)\Delta{t}_k + i\int\eta_\phi(x)\phi(x)d^3\mathbf{x}\Delta{t}_k\right]}{\varphi_{k}} = \\
	&=\funint{\varkappa_k}\inprod{\varphi_{k+1}}{\varkappa_k}\bracket{\varkappa_k}{\exp\left[-iH(t_k)\Delta{t}_k + i\int\eta_\phi(x)\phi(x)d^3\mathbf{x}\Delta{t}_k\right]}{\varphi_k} = \\
	&=\funint{\varkappa_k}\inprod{\varphi_{k+1}}{\varkappa_k}\inprod{\varkappa_k}{\phi_k}
	\exp\left[-iH[\varkappa_k,\varphi_k]\Delta{t}_k + i\int\eta_\phi(x)\phi(x)d^3\mathbf{x}\Delta{t}_k\right] + O(\Delta{t}_k^2) =\\
	&=\funint{\varkappa_k}\exp\left\{i\int\varkappa_k(\mathbf{x})\left[\phi_{k+1}(\mathbf{x}) - \phi_k(\mathbf{x})\right]
	-iH[\varkappa_k,\varphi_k]\Delta{t}_k + i\int\eta_\phi(x)\phi(x)d^3\mathbf{x}\Delta{t}_k\right\} + O(\Delta{t}_k^2) = \\
	&= \funint{\varkappa_k}\exp\left\{i\left[\int\varkappa_k\varphi_{k,t}d^3\mathbf{x} - H\left[\varkappa_k,\varphi_k\right] + \int\eta_\phi(x)\phi(x)d^3\mathbf{x}\right]\Delta{t}_k\right\} + O(\Delta{t}_k^2)
	\end{split}
\end{equation}
By taking continuous limit, one can obtain
\begin{equation}
	Z\left[t_\text{in}, t_\text{out};\varphi_\text{in},\varphi_\text{out};\eta_\phi\right]  = 
	\funint{\varphi,\pi}\exp\left\{iS_\text{cl}\left[\pi, \varphi\right] + i\int\eta_\phi(x)\varphi(x)d^4x\right\} \label{formalism/funint}.
\end{equation}
If the Hamiltonian is quadratic in $\pi$, like the scalar field one, the $\pi$ integral is Gaussian and yields
\begin{equation}
	Z\left[t_\text{in}, t_\text{out};\varphi_\text{in},\varphi_\text{out};\eta_\phi\right] = 
	C\funint{\varphi}\exp\left\{i\int\left[\mathcal{L}\left(\varphi, \varphi_\mu\right) + \eta\phi\right]d^4x\right\}.
\end{equation}
For Dirac and NJL fermions $\pi = i\psi^+$, therefore
\begin{equation}
	Z\left[t_\text{in}, t_\text{out};\psi_\text{in}, \bar{\psi}_\text{in}, \psi_\text{out}, \bar{\psi}_\text{out}; \eta, \bar{\eta}\right] = 
	\funint{\bar{\psi},\psi}\exp\left\{i\int\left[\mathcal{L}\left(\bar{\psi},\psi, \psi_\mu\right) + \bar{\eta}\psi + \bar{\psi}\eta\right]d^4x\right\}.
\end{equation}

Now, recalling the field equation \eqref{formalism/fieldequation} and the relations \eqref{formalism/defW}, \eqref{formalism/defgamma} between $Z, W$, and $\Gamma$, one can write
\begin{equation}
	\begin{split}
		\exp{i\Gamma[\varphi]} &= 
		\exp\left\{iW\left[t_\text{in}, t_\text{out};\varphi; -\fundelta{\Gamma}{\varphi} \right] + i\int\fundelta{\Gamma}{\varphi(x)}\varphi(x)d^4x\right\} =\\ 
		&=\funint{\phi}\exp\left\{iS_\text{cl}[\phi] - i\int\fundelta{\Gamma}{\varphi(x)}\left[\phi(x) - \varphi(x)\right]d^4x\right\} = \\
		&=\funint{\chi}\exp\left\{iS_\text{cl}[\varphi + \chi] - i\int\fundelta{\Gamma}{\varphi(x)}\chi(x)d^4x\right\}
	\end{split} \label{formalism/selfconsistent}
\end{equation}
This formula, where a simple shift of integration variable $\phi := \varphi + \chi$ was made, gives a self-consistent expression for $\Gamma$. 
Let
\begin{equation}
	\Gamma[\varphi] = S_\text{cl}[\varphi] + \Gamma_1[\varphi].
\end{equation}
Then
\begin{equation}
	\exp{i\Gamma_1[\varphi]} = \funint{\chi}\exp\left\{\frac{i}{2}\iint\chi(x)\frac{\delta^2S_\text{cl}[\varphi]}{\delta\varphi(x)\delta\varphi(y)}\chi(y)d^4xd^4y + V[\varphi, \chi] - i\int\fundelta{\Gamma_1}{\varphi(x)}\chi(x)d^4x\right\}
\end{equation}
Thus
\begin{equation}
	\Gamma[\varphi] = S_\text{cl}[\varphi] + i\theta{}Tr\ln\frac{\delta^2S_\text{cl}[\varphi]}{\delta\varphi(x)\delta\varphi(y)} + \ldots
\end{equation}
where $\theta$ depends on field statistics, being $\frac{1}{2}$ for bosons and $-1$ for fermions.

 It is possible to deduce any correlation function from the effective action. For the two-point bosonic correlator one can write
\begin{equation}
	\bracket{\text{out}}{T\phi(x)\phi(y)}{\text{in}} = -\exp\left[-iW\right]\fundelta{}{\eta(x)}\fundelta{}{\eta(y)}\exp\left[iW\right] =
	\fundelta{W}{\eta(x)}\fundelta{W}{\eta(y)} - i\frac{\delta^2W}{\delta\eta(x)\delta\eta(y)}
\end{equation} 
Now, by recalling \eqref{formalism/deffield}, \eqref{formalism/fieldequation}, one can obtain
\begin{equation}
		\bracket{\text{out}}{T\phi(x)\phi(y)}{\text{in}} = \varphi(x)\varphi(y) - i\fundelta{\varphi(y)}{\eta(x)} = 
		\varphi(x)\varphi(y) - i\left[\fundelta{\eta(x)}{\varphi(y)}\right]^{-1} = 
		\varphi(x)\varphi(y) + i\left[\frac{\delta^2\Gamma[\varphi]}{\delta\varphi(x)\varphi(y)}\right]^{-1}.
		\label{formalism/twopoint}
\end{equation}
It is also easy to see that
\begin{equation}
	\fundelta{}{\eta(x)} = \int\,d^4y\fundelta{\varphi(y)}{\eta(x)}\fundelta{}{\varphi(y)} = -\int\,d^4y\left[\frac{\delta^2\Gamma[\varphi]}{\delta\varphi(x)\varphi(y)}\right]^{-1}\fundelta{}{\varphi(y)}.
\end{equation}
For higher correlation functions
\begin{equation}
	\begin{split}
		&\bracket{\text{out}}{\prod_{k=1}^{n > 2}
		T\phi(x_k)
	}{\text{in}} =
	i^n\exp\left[-i\Gamma[\varphi] - i\int\eta(x)\varphi(x)d^4x\right]\times\\
	&\left.\times\left\{\prod_{k=1}^{n > 2}\int\,d^4y_k\left[\frac{\delta^2\Gamma[\varphi]}{\delta\varphi(x_k)\varphi(y_k)}\right]^{-1}\fundelta{}{\varphi(y_k)}\right\}
	\exp\left[i\Gamma[\varphi] + i\int\eta(x)\varphi(x)d^4x\right]
	\right\vert_{\eta = 0} = \\
	& = \prod_{k = 1}^{n - 2}\left\{\varphi(x_k) + i\int\,d^4y_k\left[\frac{\delta^2\Gamma[\varphi]}{\delta\varphi(x_k)\varphi(y_k)}\right]^{-1}\fundelta{}{\varphi(y_k)}\right\}
	\left\{\varphi(x_{n-1})\varphi(x_n) + i\left[\frac{\delta^2\Gamma[\varphi]}{\delta\varphi(x_{n-1})\varphi(x_n)}\right]^{-1}\right\}.\end{split}
\label{formalism/correlators}
\end{equation}
The effective action can be represented as a power series in $\varphi$:
\begin{equation}
	\Gamma\left[\varphi\right] = -\int\eta_\varphi(x)\varphi(x)d^4x + \frac{1}{2}\iint\varphi(x_1)G_0^{-1}(x_1, x_2)\varphi(x_2)d^4x_1d^4x_2 + V\left[\varphi\right].
\end{equation}
Then
\begin{equation}
	\begin{split}
	G_\varphi(x_1,x_2) &\equiv \left[\frac{\delta^2\Gamma[\varphi]}{\delta\varphi(x_1)\varphi(x_2)}\right]^{-1} = \\
	& = \int\,d^4y_0G_0(x_1,y_0)\sum_{n = 0}^\infty(-1)^n\prod_{k=1}^n\left\{\iint\,d^4y_kd^4z_k\frac{\delta^2V[\varphi]}{\delta\varphi(y_{k-1})\delta\varphi(z_k)}G_0(z_k, y_{k})\right\}\delta(y_{n} - x_2)
	\end{split}
\label{formalism/propagator}
\end{equation}
$G_\varphi$ is the effective propagator in the presence of background field $\varphi$.

The relations \eqref{formalism/correlators} and \eqref{formalism/propagator} have simple diagrammatic interpretation. $G_0$ can be represented as a line, and $\frac{\delta^NV[\varphi]}{\delta\varphi(x_1)\ldots\delta\varphi(x_N)}$ as an $N$-th degree vertex.
Then

\tikzfeynmanset{
	every blob/.style={draw=black!20!black, pattern color=black!20!black},
}

\begin{equation} 
	G_\varphi(x_1,x_2) \equiv 
		\feynmandiagram[inline = (a.base), horizontal=a to b] {
			b [particle=\(x_1\)] -- [plain] a [particle=\(x_2\)],
		};
 + 
	\feynmandiagram[inline = (a.base), horizontal = a to c] {
		a [particle=\(x_1\)] -- plain b [blob] -- plain c [particle=\(x_2\)]
	};
 +
	\feynmandiagram[inline = (a.base), horizontal = a to d] {
		a [particle=\(x_1\)] -- plain b [blob] -- plain c [blob] -- plain d [particle=\(x_2\)]
	};
 + \ldots \label{formalism/propagatorgraph}
\end{equation}
By virtue of \eqref{formalism/twopoint} this gives the two-point correlator.  The following theorem can be proven:
\begin{theorem}[Tree theorem] \label{formalism/treetheorem}
	Any correlator can be represented as a sum of tree diagrams where vertices correspond to $\frac{\delta^NV[\varphi]}{\delta\varphi(x_1)\ldots\delta\varphi(x_N)}$, and edges to $G_0(x,y)$. 
\end{theorem}

\begin{proof}[Proof]
	Observe that the theorem holds for the two-point correlation function. Assume that it is true for the $N$-point correlator $G_N(x_1,\ldots,x_N)$. By virtue of \eqref{formalism/correlators}, 
	\begin{equation}
		G_{N + 1}(x_1,\ldots,x_{N + 1}) = \left[\varphi(x_1) + i\int\,d^4yG_\varphi(x_1,y)\fundelta{}{\varphi(y)}\right]G_N(x_2,\ldots,x_{N + 1}).
	\end{equation}
	Since $G_0$ does not by definition depend on $\varphi$, $\fundelta{}{\varphi(y)}$ can couple either to $\varphi(x_k)$, which produces a dangling end, or to some vertex $\frac{\delta^mV[\varphi]}{\delta\varphi(y_1)\ldots\delta\varphi(y_m)}$ which adds a branch. Neither gives a loop. Thus, the theorem also holds for $G_{N+1}$.
\end{proof}
As a corollary, $\Gamma\left[\varphi\right]$ is the sum of all one-particle irreducible diagrams \cite{Weinberg1996}, and all the divergencies of a field theory are contained in its effective action which makes the latter a convenient object for developing renormalizaion schemes. The self-consistent equation \eqref{formalism/selfconsistent} will be the primary object of study.

The first step in defining an R-operation for effective action is to recall the derivation of path integral \eqref{formalism/funint} and notice that the integration variables $\phi(x), \pi(x)$ are, in fact, matrix elements of operator-valued distributions. Thus, they are distributions themselves, and their products have the same ambiguities. 
The only subtlety here stems from perturbation theory. Interacting field can be represented as a power series
\begin{equation}
	\phi_\text{int}(x) = \sum_{n = 1}^\infty\idotsint\,d^4y_1\ldots d^4y_n K_n(x;y_1,\ldots,y_n)\phi_\text{free}(y_1)\ldots\phi_\text{free}(y_n). \label{formalism/perturbfield}
\end{equation}  
Redefinition of free field multiplication $* \rightarrow *^\prime$ along the lines of \eqref{formalism/defambiguity} will obviously induce some renormalization of the interacting field $\phi_\text{int}(x) \rightarrow \phi_\text{int}^\prime(x)$ which in general is not the same as constructing the interacting fields first and then trying to redefine their products.  In addition, \eqref{formalism/perturbfield}  implies that perturbation theory requires a multiplication operation for the free fields since one cannot in general expect the kernels $K_n$ to be supported in ``safe'' areas only. Thus, it is natural to define R-operation as tweaking \textit{free} field multiplication to get rid of perturbative infinities. Since a path integral can only be computed as a perturbation of a Gaussian one, the integration variable can be viewed as a free field. Thus
\begin{equation}
	\begin{split}
		RZ[\eta] &= R\funint{\phi}\exp{iS[\phi, \eta]} =\\
		&= \sum_{N = 0}^\infty\frac{i^N}{N!}\funint{\phi}RS_\text{int}^N\left[\phi\right]\exp\left[\frac{i}{2}\iint\phi(x)D_0^{-1}(x,y)\phi(y)d^4xd^4y + i\int\phi(x)\eta(x)d^4x\right] \sim \\
		& \sim \sum_{N = 0}^\infty\frac{i^N}{N!}RS^N\left(-i\frac{\delta}{\delta\eta}\right)\exp\left[-\frac{i}{2}\iint\eta(x)D_0(x,y)\eta(y)d^4xd^4y\right], \label{formalism/renormz}
	\end{split}
\end{equation}
where $S_{int}\left[\phi, \eta\right]$ is the interaction part of the action. In what follows till the end of this section the fields will be considered free unless stated otherwise. Also, field variables as they appear in the path integral will be considered, rather than field operators.

 Let
\begin{equation}
	\phi_{\bar{\lambda}}(x) = \partial_{\lambda_1}\ldots\partial_{\lambda_n}\phi(x), n \in {0,\ldots,\infty},
\end{equation}
where the case $n = 0$ is the undifferentiated field, and
\begin{equation}
		\Delta\left[\phi_{\bar{\lambda}_1}(x_1),\ldots,\phi_{\bar{\lambda}_n}(x_n)\right] = 
		\begin{cases}\phi_{\bar{\lambda}}(x), n = 1\\
			P_{n,\bar{\lambda_1}\ldots\bar{\lambda_n}}\left[x_1, \phi(x_1), \frac{\partial}{\partial{x_1^\mu}}\right]\delta(x_1 - x_2)\ldots\delta(x_1 - x_n), n > 1,
		\end{cases}
\end{equation}
where $P_n\left[x_1, \phi(x_1), \frac{\partial}{\partial{x_1^\mu}}\right]$ is some polynomial in $\phi(x_1), \frac{\partial}{\partial{x_1^\mu}}$ operator with possibly point-dependent coefficients. Then two multiplication operations $*$ and $*^\prime$ are related by 
\begin{equation}
	\left[\phi_{\bar{\lambda}_1}(x_1)\ldots\phi_{\bar{\lambda}_n}(x_n)\right]^\prime = 
	\sum_{V_1\cup\ldots\cup{V_m} = \left\{1,\ldots,n\right\}}\prod_{k=1}^m\Delta\left[\left\{\phi_{\bar{\lambda}_l}(x_l)\vert l \in V_k\right\}\right],
	\label{formalism/renormprod}
\end{equation}
where the summation is performed over all possible breakdowns of the index set $\left\{1,\ldots,n\right\}$ into non-empty non-intersecting subsets $V_k$; $m$ runs from $1$ to $n$, and the term with $m = n$ is obviously equal to  $\phi_{\bar{\lambda}_1}(x_1)\ldots\phi_{\bar{\lambda}_n}(x_n)$.  

Some assumption about the raw propagator $G_0$ is required to proceed further. For all physically interesting cases, 
\begin{equation}
G_0(x,y) = 
\begin{cases}
	\int\frac{T(p)e^{ip(x - y)}d^4p}{p^2 - m^2}\\
	C\delta(x-y)
\end{cases} \label{formalism/propagator}
\end{equation}
where $T(p)$ is some at most linear expression in $p^\mu$, and the case $G_0 = C\delta(x-y)$ is the bosonized NJL model's bosonic field and its likes.  

The most general loop operator can be defined as
\begin{equation}
L_{p}[\{q\}, \{m\}]F(p;\{q\}) \equiv 
\int\frac{F(p;\{q\})d^4p}{\prod_{i=1}^l\left[(p - k_i)^2 - m_i^2 + i\epsilon\right]} \label{formalism/defloop},
\end{equation}
where $\{q\}$ is the set of all external momenta, and $k_i$ can be expressed through $\{q\}$. With the use of Schwinger parametrization this can be expressed as
\begin{equation}
\begin{split}
L_{p}[\{q\}, \{m\}]F(p;\{q\}) &= \int_0^\infty\,d\bar{\alpha}\int\,d^4p\,F(p;\{q\})\exp\left\{-\epsilon\sum_{k=1}^l\alpha_k 
+ i\sum_{k=1}^l\alpha_k\left[(p - k_k)^2 - m_k^2\right]\right\} = \\
&=\left|
	\begin{array}{l}
	\sigma = \sum_{i=1}^l\alpha_i\\
	\alpha_k = \sigma\tau_k\\
	\bar{k} = \sum_{i=1}^l\tau_ik_i\\
	p = r + \bar{k}\\
	d\bar{\alpha} = d\sigma{}d\bar{\tau}\sigma^{l - 1}\delta\left(1 - \sum_{i=1}^l\tau_i\right)
	\end{array}
\right| = 
\int_0^1d\bar{\tau}\delta\left(1 - \sum_{i=1}^l\tau_i\right)\times\\
&\times\int_0^\infty\,d\sigma\sigma^{l - 1}
\exp\left[-\epsilon\sigma + i\sigma\sum_{i=1}^l\tau_i\left(k_i^2 - m_i^2\right) - i\sigma\bar{k}^2\right]\int\,d^4rF(r + \bar{k}, \{q\})
\exp\left[i\sigma{r^2}\right]
\end{split} \label{formalism/Schwinger}
\end{equation}
If $F(p;\left\{q\right\})$ is polynomial in $p^\mu$, integration over $r$ can be performed easily using
\begin{equation}
\begin{split}
\int\exp\left[ip^2\alpha\right]d^4p = -\frac{i\pi^2}{\alpha^2}\\
\int\prod_{k=1}^n{p^{\mu_k}}\exp\left[ip^2\alpha\right]d^4p &= 
\begin{cases}
0, n = 2m + 1\\
\frac{i\pi^2C_m}{\alpha^{2 + m}}S_m^{\mu_1\ldots\mu_n}(\eta), n = 2m,
\end{cases}
\end{split}\label{formalism/reduction}
\end{equation}
where $C_m$ is a numeric constant and $S_m^{\mu_1\ldots\mu_n}(\eta)$ is $m$-th symmetric power of the metric. 
Now it is clear that
\begin{equation}
\begin{split}
&L_{p}[\{q\}, \{m\}]F(p;\{q\}) = 
\int_0^1d\bar{\tau}\delta\left(1 - \sum_{i=1}^l\tau_i\right)G(\left\{q\right\},\bar{\tau})\\
&G(\left\{q\right\},\bar{\tau}) = \int_0^\infty\,d\sigma\sum_{m = 0}^{N}\sigma^{l - m - 3}G_m(\left\{q\right\})\exp\left[-\sigma{}Q(\bar{\tau},\left\{q\right\})\right] = \sum_{m = 0}^{N}\Gamma(l - m - 2)G_m(\left\{q\right\})Q^{m - l + 2}(\bar{\tau}, \left\{q\right\})
\end{split} \label{formalism/loopsolved}
\end{equation}
The divergent terms of $G$ will have nonnegative powers of $Q$. Since $Q$ is polynomial in $k_i$, this implies that the divergent part of the loop is quasi-local which can be demonstrated by going back to $x$-representation:
\begin{equation}
\text{divp}{LF} = \idotsint\,d^4x_1{\ldots}d^4x_lF(x_1,\ldots,x_l)C_n(\partial^\mu_{x_i})\prod_{i=2}^l\delta(x_i - x_1),
\label{formalism/localdivergence}
\end{equation}
where $C_n$ is a differential operator with constant divergent coefficients.

The generating functional \eqref{formalism/renormz} can be expressed as a sum of diagrams:
\begin{equation}
\begin{split}
Z\left[\eta\right] &= \sum_D\idotsint\,G_D(\bar{q})\eta(q_1)\ldots\eta(q_{m_D})\,dq_1\ldots\,dq_{m_D};\\
G_D(\bar{q}) &= 
L_{\bar{\alpha}_1,p_1}[\{m\}_1;\{q\}_D]L_{\bar{\alpha}_2,p_2}[\{m\}_2;p_1,\{q\}_D]\ldots\,
L_{\bar{\alpha}_{l_D},p_{l_D}}[\{m\}_{l_D},p_1,\ldots,p_{l_D - 1},\{q\}_D]
F_D\left[\left\{\bar{\alpha}_k\right\},\left\{p\right\},\bar{q}_D\right].
\end{split} \label{formalism/loopexpr}
\end{equation}
Here $D$-th diagram has $m_D$ external momenta and $l_D$ loops.  
For all local theories with bare propagators of the form \eqref{formalism/propagator}, an arbitrary diagram can be factorized into its connected parts and
any connected (sub)diagram will have $\delta\left(\sum_{i}q_i\right)$ in the numerator, where $q_i$ are external momenta. A loop is connected by definition.
If field products are redefined in the spirit of \eqref{formalism/renormprod} then
\begin{equation}
Z^\prime = Z + \Delta{Z},
\end{equation}
$\Delta{Z}$ being the sum of counterterm diagrams. 
Now \eqref{formalism/localdivergence} implies that $\text{divp}Z$ and $\Delta{Z}$ have the same form with a suitable choice of counterterms.
Thus multiplication of field distributions can be altered to render $\text{divp}_rZ + \Delta_rZ$ finite even after the intermediate regularization is lifted.

A practical R-operation algorithm can be defined as follows \cite{Hepp1966}, \cite{Zimmermann1969}, \cite{Zavyalov1990}, \cite{itzykson2012}. Let \textit{forest} be a sequence of diagrams $\{G_i\}$ where each $G_{i+1}$ is either a subdiagram of $G_i$ or does not overlap the latter. For any diagram in a theory with propagators of the form 
\eqref{formalism/propagator}, its divergent part is given by first $n$ terms of its Taylor expansion in external momenta.
The crucial ingredients of an R-operation are two operators:
\begin{equation}
\begin{split}
&M_G: M_GG = \divp{G}\\{}
&\Lambda_G: \Lambda_GG = -M_G + \lambda_G(q),
\end{split}
\end{equation}
where $\lambda_G(q)$ is a finite polynomial in external momenta of order no more than $n$.

Then
\begin{equation}
RG_D = \left[1 + \sum_{\mathcal{F}}\prod_{G \in \mathcal{F}}\Lambda_G\right]D,
\end{equation}
where the summation is performed over all possible forests of subdiagrams in $G$, including the diagram itself.

A generic theory has infinitely many diagrams and thus will depend on an infinite number of constants.
If there is some additional symmetry such as gauge invariance, extra restrictions must be imposed on finite renormalization
which will reduce the independent parameter count. In some cases the latter will become finite: these are the renormalizable theories.

The construction above can be repeated for effective action instead of $Z$ with background field $\varphi$ in place of external sources.
In this case, it is the same-point powers of the deviation field $\chi$ that will be redefined in the process of renormalization.
This is the formalism that will be used to construct the effective action for bosonized NJL model. The crucial difference with the standard formulation of BPHZ is that the presence of background field $\varphi$ allows, in principle, counterterms that depend on it. As $\varphi$ may vary with spacetime point, the $R$-operation may become point-dependent as well. This will be elaborated in more detail in \ref{onelooprenorm}.

\section{The one-loop effective action for bosonized NJL model}  \label{NJL-action}

\subsection{Fierz identities and non-uniqueness of bosonization}

The form of the original NJL Lagrangian \eqref{intro/NJL} suggests introducing two bosonic fields $\xi(x),\eta(x)$,
coupling to $\bar{\psi}\psi$ and $\bar{\psi}\gamma^5\psi$ respectively. 
However, Fierz identities imply that the NJL interaction terms can be rewritten through other quartic scalar combinations of $\bar{\psi}, \psi$ and $\gamma$-matrices. In general, all the terms 
$\bar{\psi}\psi\bar{\psi}\psi$, $\bar{\psi}\gamma_\mu\psi\bar{\psi}\gamma^\mu\psi$, 
$\bar{\psi}\sigma_{\mu\nu}\psi\bar{\psi}\sigma^{\mu\nu}\psi$,
$\bar{\psi}\gamma^5\gamma_\mu\psi\bar{\psi}\gamma^5\gamma^\mu\psi$,
$\bar{\psi}\gamma^5\psi\bar{\psi}\gamma^5\psi$ may be involved. Thus, bosonization prescription for the model \eqref{intro/NJL} is in fact ambiguous:
one may introduce different inventories of bosonic fields. 
The principle of maximum degrees of freedom will be upheld here: introduce all possible fields, then Fierz identities will result in relations on coupling constants.
In this subsection. the following questions will be addressed:
\begin{enumerate}
\item What are the Fierz identities for fermions with additional $U(N)$ symmetry?
\item How many linearly independent quartic interaction terms are there?
\item How the bosonic coupling constants are related to the original model's ones?
\end{enumerate}

As interaction terms in the lagrangian must be invariant under symmetries of the model, the answer to the questions above requires a recourse to the theory of invariants.
Let 
\begin{equation}
G^a = 
\begin{cases}
0, a = 0;\\
T^a, \text{otherwise}
\end{cases} \label{fierz4/extbasis}
\end{equation}
be the complete basis in $gl(N, \mathbb{C})$.

A generic fermionic self-interaction term is 
\begin{equation}
\Lgr_\text{int} = g_{\bar{i}_1\ldots\bar{i}_n}t_{\bar{\alpha}_1\ldots\bar{\alpha}_n}
\bar{\psi}\tau_{\bar{i}_1}^{\bar{\alpha}_1}\psi\ldots
\bar{\psi}\tau_{\bar{i}_n}^{\bar{\alpha}_n}\psi
\end{equation}
Invariance under symmetries of the model requires $t_{\bar{\alpha}_1\ldots\bar{\alpha}_n}$ to be invariant.
Since the complete group of symmetries is the direct product of Lorentz and unitary groups,
\begin{equation}
t_{\bar{\alpha}_1\ldots\bar{\alpha}_n} = \sum_{k,l}t_{kl}\Theta_{k;\mu_1\ldots\mu_m}\Phi_{l;a_1\ldots{}a_n},
\end{equation}
where $\Theta$ and $\Phi$ are basis Lorentz- and $U(N)$-invariant forms respectively. 
The question then boils down to classification of such forms which, in turn, relies on two 
results from the theory of invariants. 

\begin{theorem}[Classification of Lorentz-invariant forms] \label{fierz4/lorentzforms}
All Lorentz-invariant (pseudo)tensors of arbitrary rank must be built from the metric and Levi-Civita symbol
$\epsilon^{\kappa\lambda\mu\nu}$.
\end{theorem}
This follows from Weyl's theorem on invariants of orthogonal groups \cite[53]{Weyl1946}.

As $gl(N,\mathbb{C})$ is a complexification of $u(N)$, Ad-invariant $n$-forms must be the same on both.
Their classification is facilitated by the following
\begin{theorem}[Trace decomposition theorem for $gl(N,\mathbb{C})$] \label{fierz4/tracedecomposition}
Let $\Phi(X_1,\ldots,X_n)$ be an Ad-invariant linear n-form on the defining representation of $gl(N, \mathbb{C})$:
\begin{equation}
\Phi(X_1,\ldots,X_n) = \Phi_{i_1\ldots{}i_n}^{j_1\ldots{}j_n}X_{1;j_1}^{i_1}\ldots{}X_{n;j_n}^{i_n}
\end{equation}
Then
\begin{equation}
\Phi_{i_1\ldots{}i_n}^{j_1\ldots{}j_n} = \sum_{P}\Phi_P\delta^{j_1}_{Pi_1}\ldots\delta^{j_n}_{Pi_n},
\end{equation}
where summation is carried over all permutations of lower indices. This can also be stated in an equivalent form:
let $I_1,\ldots,I_l$ are ordered index sets $I_k=\{i_{k1},\ldots,i_{km_k}\}$ such that
$I_1\cup\ldots\cup{I_n} = \{1,\ldots,n\}$. Then
\begin{equation}
\Phi(X_1,\ldots,X_n) = \sum_{I_1,\ldots,I_l}\prod_{I_k}\Tr\left[\prod_{i \in I_k}X_i\right]
\end{equation}
For example, for $n=3$:
\begin{equation}
\begin{split}
\Phi(X_1,X_2,X_3) &= \Phi_1\Tr{X_1X_2X_3} + \Phi_2\Tr{X_1X_3X_2} +\\ 
&+ \Phi_3\Tr{X_1}\Tr{X_2X_3} + \Phi_4\Tr{X_2}\Tr{X_1X_3} + \Phi_5\Tr{X_3}\Tr{X_1X_2}
+ \Phi_6\Tr{X_1}\Tr{X_2}\Tr{X_3}
\end{split}
\end{equation}
\end{theorem}
This is a corollary of the first fundamental theorem of invariant theory for matrices \cite[21--22]{Kraft1996}.
In turn, it has two important corollaries. 

\begin{corollary} \label{fierz4/tensorcorollary}
Expression of an $n$-form $\Phi$ through traces implies that any Ad-invariant $\Phi^{a_1\ldots{a_n}}$
can be expressed through $\delta^{ab}, f^{abc}, d^{abc}$.
\end{corollary}

\begin{corollary} \label{fierz4/deltacorollary}
Let $\Phi^{a_1\ldots{a_n}}$ be coefficients of some Ad-invariant n-form.
Then,
\begin{equation}
\Phi^{a_1\ldots{a_n}}G^{a_1\, i_1}_{j_1}\ldots{}G^{a_n\,i_n}_{j_n} = 
\sum_{P}C_P\delta^{i_1}_{Pj_1}\ldots\delta^{i_n}_{Pj_n}
\end{equation}
\end{corollary}

For $n=2$,
\begin{equation}
T^a_{ij}T^a_{kl} = \frac{1}{2}\delta_{il}\delta_{jk} - \frac{1}{2N}\delta_{ij}\delta_{kl} \label{fierz4/trace2},
\end{equation}
which can also be proven by other means (\cite{Haber2021}, which contains a reference on identities involving $u(N)$ generators; for $\gamma$-matrix identities see \cite{PBPal2007}).

Only parity-invariant interactions will be studied here. 
It is possible to consider parity-violating interactions too, for example $\bar{\psi}\gamma^5\psi\bar{\psi}\psi$,
but this will require more effort while not making any conceptual difference.
As Fierz transformations preserve parity, there must be a separate set of Fierz identities for odd-parity terms. 

Let
\begin{equation}
\begin{split}
I_{10} &= \bar{\psi}\psi\bar{\psi}\psi\\
I_{11} &= \bar{\psi}T^a\psi\bar{\psi}T^a\psi\\
I_{20} &= \bar{\psi}\gamma_\mu\psi\bar{\psi}\gamma^\mu\psi\\
I_{21} &= \bar{\psi}T^a\gamma_\mu\psi\bar{\psi}T^a\gamma^\mu\psi\\
I_{30} &= \bar{\psi}\sigma_{\mu\nu}\psi\bar{\psi}\sigma^{\mu\nu}\psi\\
I_{31} &= \bar{\psi}T^a\sigma_{\mu\nu}\psi\bar{\psi}T^a\sigma^{\mu\nu}\psi\\
I_{40} &= \bar{\psi}\gamma^5\gamma_{\mu}\psi\bar{\psi}\gamma^5\gamma^\mu\psi\\
I_{41} &= \bar{\psi}T^a\gamma^5\gamma_{\mu}\psi\bar{\psi}T^a\gamma^5\gamma^\mu\psi\\
I_{50} &= \bar{\psi}\gamma^5\psi\bar{\psi}\gamma^5\psi\\
I_{51} &= \bar{\psi}T^a\gamma^5\psi\bar{\psi}T^a\gamma^5\psi
\end{split} \label{fierz4/defterms}
\end{equation}
As a direct consequence of \ref{fierz4/lorentzforms} and \ref{fierz4/tracedecomposition}, all possible Lorentz-, P- and $U(N)$-invariant quartic fermion interaction terms can be written as linear combinations of the terms above.

Then, \eqref{fierz4/trace2} can be used to proceed further: 
\begin{equation}
\bar{\psi}_iT^a_{ij}\Gamma_{(\alpha)}\psi_j\bar{\psi}_kT^a_{kl}\Gamma_{(\beta)}\psi_l = 
\frac{1}{2}\bar{\psi}_i\Gamma_{(\alpha)}\psi_j\bar{\psi}_j\Gamma_{(\beta)}\psi_i - \frac{1}{2N}\bar{\psi}_i\Gamma_{(\alpha)}\psi_i\bar{\psi}_j\Gamma_{(\beta)}\psi_j, \label{fierz4/deltadecomposition}
\end{equation}

But any matrix in $Cl_{1,3}(\mathbb{R})\otimes{u(N)}$ can be written as
\begin{equation}
\begin{split}
M &= \frac{1}{4N}\Tr{M} + \frac{T^a}{2}\Tr{T^aM} 
+ \frac{\gamma_\mu}{4N}\Tr{\gamma^\mu{}M} + \frac{T^a\gamma_\mu}{2}\Tr{T^a\gamma^\mu{}M} 
+ \frac{\sigma_{\mu\nu}}{8N}\Tr{\sigma^{\mu\nu}M} + \frac{T^a\sigma_{\mu\nu}}{4}\Tr{T^a\sigma^{\mu\nu}M}\\
&- \frac{\gamma^5\gamma_\mu}{4N}\Tr{\gamma^5\gamma^\mu{}M} - \frac{T^a\gamma^5\gamma_\mu}{2}\Tr{T^a\gamma^5\gamma^\mu{M}}
+ \frac{\gamma^5}{4N}\Tr{\gamma^5{M}} + \frac{T^a\gamma^5}{2}\Tr{T^a\gamma^5M}.
\end{split} \label{fierz4/decomposition}
\end{equation}
where $\Gamma_{(\alpha)} \in Cl_{1,3}(\mathbf{R})$. $\psi_i\bar{\psi}_i \in Cl_{1,3}(\mathbf{R})$ and can be expanded as
\begin{equation}
\psi_i\bar{\psi}_i = - \frac{\bar{\psi}_i\psi_i}{4} - \frac{\gamma_\mu\bar{\psi}_i\gamma^\mu\psi_i}{4} 
- \frac{\sigma_{\mu\nu}\bar{\psi}_i\sigma^{\mu\nu}\psi_i}{8} 
+ \frac{\gamma^5\gamma_\mu\bar{\psi}_i\gamma^5\gamma^\mu\psi_i}{4} 
- \frac{\gamma^5\bar{\psi}_i\gamma^5\psi_i}{4}. \label{fierz4/cliffordexpansion}
\end{equation}
Here the minus sign appears because $\psi$ is a Grassmann variable.
Then various $\gamma$ matrix identities, of which 
\begin{equation}
\gamma_\lambda\gamma_\mu\gamma_\nu = \eta_{\lambda\mu}\gamma_\nu - \eta_{\lambda\nu}\gamma_\mu + \eta_{\mu\nu}\gamma_\lambda + i\epsilon_{\lambda\mu\nu\rho}\gamma^5\gamma^\rho
\end{equation}
is the most important, can be used to yield
\begin{equation}
\begin{split}
\bar{\psi}T^a\psi\bar{\psi}T^a\psi &= -\frac{(N + 4)\bar{\psi}\psi\bar{\psi}\psi}{8N} 
- \frac{\bar{\psi}\gamma_\mu\psi\bar{\psi}\gamma^\mu\psi}{2}
- \frac{\bar{\psi}\sigma_{\mu\nu}\psi\bar{\psi}\sigma^{\mu\nu}\psi}{16} 
+ \frac{\bar{\psi}\gamma^5\gamma_\mu\psi\bar{\psi}\gamma^5\gamma^\mu\psi}{8}
- \frac{\bar{\psi}\gamma^5\psi\bar{\psi}\gamma^5\psi}{8}\\
\bar{\psi}T^a\gamma_\mu\psi\bar{\psi}T^a\gamma^\mu\psi &=
-\frac{\bar{\psi}\psi\bar{\psi}\psi}{2}
+ \frac{(N - 2)\bar{\psi}\gamma_\mu\psi\bar{\psi}\gamma^\mu\psi}{4N}
+ \frac{\bar{\psi}\gamma^5\gamma_\mu\psi\bar{\psi}\gamma^5\gamma^\mu\psi}{4} 
+ \frac{\bar{\psi}\gamma^5\psi\bar{\psi}\gamma^5\psi}{2}\\
\bar{\psi}T^a\sigma_{\mu\nu}\psi\bar{\psi}T^a\sigma^{\mu\nu}\psi &= 
-\frac{3\bar{\psi}\psi\bar{\psi}\psi}{2} + \frac{N-2}{4N}\bar{\psi}\sigma_{\mu\nu}\psi\bar{\psi}\sigma^{\mu\nu}\psi - \frac{3\bar{\psi}\gamma^5\psi\bar{\psi}\gamma^5\psi}{2} \\
\bar{\psi}T^a\gamma^5\gamma_\mu\psi\bar{\psi}T^a\gamma^5\gamma^\mu\psi &=
\frac{\bar{\psi}\psi\bar{\psi}\psi}{2} + \frac{\bar{\psi}\gamma_\lambda\psi\bar{\psi}\gamma^\lambda\psi}{4}
+ \frac{(N - 2)\bar{\psi}\gamma^5\gamma_\lambda\psi\bar{\psi}\gamma^5\gamma^\lambda\psi}{4N}
- \frac{\bar{\psi}\gamma^5\psi\bar{\psi}\gamma^5\psi}{2}\\
\bar{\psi}T^a\gamma^5\psi\bar{\psi}T^a\gamma^5\psi &=
-\frac{\bar{\psi}\psi\bar{\psi}\psi}{8} 
+ \frac{\bar{\psi}\gamma_\mu\psi\bar{\psi}\gamma^\mu\psi}{8}
- \frac{\bar{\psi}\sigma_{\mu\nu}\psi\bar{\psi}\sigma^{\mu\nu}\psi}{16} 
- \frac{\bar{\psi}\gamma^5\gamma_\mu\psi\bar{\psi}\gamma^5\gamma^\mu\psi}{8}
- \frac{(N + 4)\bar{\psi}\gamma^5\psi\bar{\psi}\gamma^5\psi}{8N}
\end{split} \label{fierz4/identities}
\end{equation}

More identities could be obtained by using the expansion 
\begin{equation}
\begin{split}
\psi\bar{\psi} &= - \frac{\bar{\psi}\psi}{4N} 
- \frac{\gamma_\mu\bar{\psi}\gamma^\mu\psi}{4N}
- \frac{\sigma_{\mu\nu}\bar{\psi}\sigma^{\mu\nu}\psi}{8N} 
+ \frac{\gamma^5\gamma_\mu\bar{\psi}\gamma^5\gamma^\mu\psi}{4N} 
- \frac{\gamma^5\bar{\psi}\gamma^5\psi}{4N} - \\
&- \frac{T^a\bar{\psi}T^a\psi}{2} 
- \frac{T^a\gamma_\mu\bar{\psi}T^a\gamma^\mu\psi}{2}
- \frac{T^a\sigma_{\mu\nu}\bar{\psi}T^a\sigma^{\mu\nu}\psi}{4} 
+ \frac{T^a\gamma^5\gamma_\mu\bar{\psi}T^a\gamma^5\gamma^\mu\psi}{2} 
- \frac{T^a\gamma^5\bar{\psi}T^a\gamma^5\psi}{2}
\end{split}
\end{equation}
which allows to express $\bar{\psi}\Gamma_{(\alpha)}\psi\bar{\psi}\Gamma_{(\alpha)}\psi$ through
$\bar{\psi}\Gamma_{(\beta)}\psi\bar{\psi}\Gamma_{(\beta)}\psi$ and $\bar{\psi}T^a\Gamma_{(\beta)}\psi\bar{\psi}T^a\Gamma_{(\beta)}\psi$.
This can be written in matrix form using the notation \eqref{fierz4/defterms}:
\begin{equation}
	\begin{split}
		\begin{bmatrix}
			I_{10}\\I_{20}\\I_{30}\\I_{40}\\I_{50}
		\end{bmatrix} &= 
		\begin{bmatrix}
			-\frac{1}{4N}&-\frac{1}{4N}&-\frac{1}{8N}&\frac{1}{4N}&-\frac{1}{4N}\\
			-\frac{1}{N}&\frac{1}{2N}&0&\frac{1}{2N}&\frac{1}{N}\\
			-\frac{3}{N}&0&\frac{1}{2N}&0&-\frac{3}{N}\\
			\frac{1}{N}&\frac{1}{2N}&0&\frac{1}{2N}&-\frac{1}{N}\\
			-\frac{1}{4N}&\frac{1}{4N}&-\frac{1}{8N}&-\frac{1}{4N}&-\frac{1}{4N}
		\end{bmatrix}
		\begin{bmatrix}
			I_{10}\\I_{20}\\I_{30}\\I_{40}\\I_{50}
		\end{bmatrix} + 
		\begin{bmatrix}
			-\frac{1}{2}&-\frac{1}{2}&-\frac{1}{4}&\frac{1}{2}&-\frac{1}{2}\\
			-2&1&0&1&2\\
			-6&0&1&0&-6\\
			2&1&0&1&-2\\
			-\frac{1}{2}&\frac{1}{2}&-\frac{1}{4}&-\frac{1}{2}&-\frac{1}{2}
		\end{bmatrix}
		\begin{bmatrix}
			I_{11}\\I_{21}\\I_{31}\\I_{41}\\I_{51}
		\end{bmatrix} =\\
		&=		
		\begin{bmatrix}
			-\frac{1}{4N}&-\frac{1}{4N}&-\frac{1}{8N}&\frac{1}{4N}&-\frac{1}{4N}\\
			-\frac{1}{N}&\frac{1}{2N}&0&\frac{1}{2N}&\frac{1}{N}\\
			-\frac{3}{N}&0&\frac{1}{2N}&0&-\frac{3}{N}\\
			\frac{1}{N}&\frac{1}{2N}&0&\frac{1}{2N}&-\frac{1}{N}\\
			-\frac{1}{4N}&\frac{1}{4N}&-\frac{1}{8N}&-\frac{1}{4N}&-\frac{1}{4N}
		\end{bmatrix} 
		\begin{bmatrix}
			I_{10}\\I_{20}\\I_{30}\\I_{40}\\I_{50}
		\end{bmatrix}+\\
		&+
		\begin{bmatrix}
			-\frac{1}{2}&-\frac{1}{2}&-\frac{1}{4}&\frac{1}{2}&-\frac{1}{2}\\
			-2&1&0&1&2\\
			-6&0&1&0&-6\\
			2&1&0&1&-2\\
			-\frac{1}{2}&\frac{1}{2}&-\frac{1}{4}&-\frac{1}{2}&-\frac{1}{2}
		\end{bmatrix}
		\begin{bmatrix}
			-\frac{N + 4}{8N}&-\frac{1}{8}&-\frac{1}{16}&\frac{1}{8}&-\frac{1}{8}\\
			-\frac{1}{2}&\frac{N - 2}{4N}&0&\frac{1}{4}&\frac{1}{2}\\
			-\frac{3}{2}&0&\frac{N - 2}{4N}&0&-\frac{3}{2}\\
			\frac{1}{2}&\frac{1}{4}&0&\frac{N - 2}{4N}&-\frac{1}{2}\\
			-\frac{1}{8}&\frac{1}{8}&-\frac{1}{16}&-\frac{1}{8}&-\frac{N + 4}{8N}
		\end{bmatrix}
		\begin{bmatrix}
			I_{10}\\I_{20}\\I_{30}\\I_{40}\\I_{50}
		\end{bmatrix}
	\end{split}
\end{equation}
It is easy to check that this results in a trivial identity and does not yield anything new. Thus
\eqref{fierz4/identities} is the complete set of Fierz identities for fermions with additional $U(N)$ symmetry.

The most general massless parity-preserving NJL-like lagrangian with global $U(N)$ symmetry is therefore
\begin{equation}
\Lgr(\bar{\psi},\psi) = \bar{\psi}i\widehat\partial\psi + G_1\bar{\psi}\psi\bar{\psi}\psi + G_2\bar{\psi}\gamma_\mu\psi\bar{\psi}\gamma^\mu\psi
+ G_3\bar{\psi}\sigma_{\mu\nu}\psi\bar{\psi}\sigma^{\mu\nu}\psi + G_4\bar{\psi}\gamma^5\gamma_\mu\psi\bar{\psi}\gamma^5\gamma^\mu\psi 
+ G_5\bar{\psi}\gamma^5\psi\bar{\psi}\gamma^5\psi
\end{equation}
Its maximal bosonization is
\begin{equation}
\begin{split}
\Lgr(\bar{\psi},\psi, \widehat{\Sigma}) &= 
\bar{\psi}\left(i\widehat\partial + \widehat{\Sigma}\right)\psi 
- \frac{\Sigma_{10}^2}{4\lambda_{10}} - \frac{\Sigma_{11}^a\Sigma_{01}^a}{4\lambda_{11}}
- \frac{\Sigma_{20\,\mu}\Sigma^\mu_{20}}{4\lambda_{20}} - \frac{\Sigma^a_{21\,\mu}\Sigma^{a\mu}_{21}}{4\lambda_{21}} -\\
&- \frac{\Sigma_{30\,\mu\nu}\Sigma^{\mu\nu}_{30}}{4\lambda_{30}} - \frac{\Sigma^a_{31\,\mu\nu}\Sigma^{a\mu\nu}_{31}}{4\lambda_{31}}
- \frac{\Sigma_{40\,\mu}\Sigma^\mu_{40}}{4\lambda_{40}} - \frac{\Sigma^a_{41\,\mu}\Sigma^{a\mu}_{41}}{4\lambda_{41}}
- \frac{\Sigma_{50}^2}{4\lambda_{50}} - \frac{\Sigma^a_{51}\Sigma^a_{51}}{4\lambda_{51}}\\
\widehat{\Sigma} &=
\Sigma_{10} + T^a\Sigma^a_{11} + \gamma_\mu\Sigma^\mu_{20} + T^a\gamma_\mu\Sigma^{a\mu}_{21} +\\
&+ \sigma_{\mu\nu}\Sigma^{\mu\nu}_{30} + T^a\sigma_{\mu\nu}\Sigma^{a\mu\nu}_{31}
+ \gamma^5\gamma_\mu\Sigma^\mu_{40} + T^a\gamma^5\gamma_\mu\Sigma^{a\mu}_{41} 
+ \gamma^5\Sigma_{50} + T^a\gamma^5\Sigma^a_{51} 
\end{split}
\end{equation}
The Fierz identities \eqref{fierz4/identities} will result in
\begin{equation}
	\begin{bmatrix}
		G_1\\G_2\\G_3\\G_4\\G_5
	\end{bmatrix} = 
	\begin{bmatrix}
		\lambda_{10}\\\lambda_{20}\\\lambda_{30}\\\lambda_{40}\\\lambda_{50}
	\end{bmatrix} + 
	\begin{bmatrix}
		-\frac{N + 4}{8N}&-\frac{1}{2}&-\frac{3}{2}&\frac{1}{2}&-\frac{1}{8}\\
		-\frac{1}{2}&\frac{N - 2}{4N}&0&\frac{1}{4}&\frac{1}{8}\\
		-\frac{1}{16}&0&\frac{N - 2}{4N}&0&-\frac{1}{16}\\
		\frac{1}{8}&\frac{1}{4}&0&\frac{N - 2}{4N}&-\frac{1}{8}\\
		-\frac{1}{8}&\frac{1}{2}&-\frac{3}{2}&-\frac{1}{2}&-\frac{N + 4}{8N}
	\end{bmatrix}
	\begin{bmatrix}
		\lambda_{11}\\\lambda_{21}\\\lambda_{31}\\\lambda_{41}\\\lambda_{51}		
	\end{bmatrix}
	\label{fierz4/conditions}
\end{equation}

\subsection{The one-loop terms} \label{oneloopterms}
To construct effective action for bosonized NJL model, one needs source terms for the bosonic fields.
They can be introduced by slightly modifying the fermionic generating functional.
Let
\begin{equation}
\begin{split}
Z\left[\left\{J\right\}\right] &= 
\funint{\bar{\psi},\psi}\exp\left\{i\int\left[\bar{\psi}i\widehat{\partial}\psi + 
G_i\bar{\psi}\Gamma_{i\,\bar{\alpha}}\psi\bar{\psi}\Gamma^{\bar{\alpha}}_i\psi
+ \bar{J}_1\psi +\bar{\psi}J_1 
+ 2\lambda_{2\bar{i}}J_{2\bar{i}\bar{\alpha}}\bar{\psi}\tau_{\bar{i}}^{\bar{\alpha}}\psi
+ \lambda_{2\bar{i}}J_{2\bar{i}\bar{\alpha}}J^{\bar{\alpha}}_{2\bar{i}}
\right]d^4x\right\}.
\end{split} \label{oneloopterms/genfunc}
\end{equation}
Inserting
\begin{equation}
\text{const} = \funint{\Sigma_{\bar{i}}}\exp\left\{-i\int\left[\frac{\lambda_{2\bar{i}}^{-\frac{1}{2}}\Sigma_{\bar{i}\bar{\alpha}}}{2} 
- \lambda_{2\bar{i}}^\frac{1}{2}\bar{\psi}\tau_{\bar{i}\bar{\alpha}}\psi
 - \lambda_{2\bar{i}}^\frac{1}{2}J_{2\bar{i}\bar{\alpha}}\right]
 \left[\frac{\lambda_{2\bar{i}}^{-\frac{1}{2}}\Sigma_{\bar{i}}^{\bar{\alpha}}}{2} 
- \lambda_{2\bar{i}}^\frac{1}{2}\bar{\psi}\tau_{\bar{i}}^{\bar{\alpha}}\psi
 - \lambda_{2\bar{i}}^\frac{1}{2}J_{2\bar{i}}^{\bar{\alpha}}\right]d^4x\right\}
 \label{oneloopterms/transform}
\end{equation}
will result in
\begin{equation}
\begin{split}
\tilde{Z}\left[\left\{J\right\}\right] &= 
\funint{\bar{\psi},\psi,\Sigma}\exp\left\{i\int\left[\bar{\psi}\left(i\widehat{\partial} + \widehat{\Sigma}\right)\psi
+ \bar{J}_1\psi +\bar{\psi}J_1
- \frac{\Sigma_{\bar{i}\bar{\alpha}}\Sigma_{\bar{i}}^{\bar{\alpha}}}{4\lambda_{2\bar{i}}} + J_{2\bar{i}\bar{\alpha}}\Sigma_{\bar{i}}^{\bar{\alpha}}
\right]d^4x\right\}\\
\widehat{\Sigma} &\equiv \tau_{\bar{i}\bar{\alpha}}\Sigma_{\bar{i}}^{\bar{\alpha}},
\end{split} \label{oneloopterms/bosonized}
\end{equation}
provided that Fierz conditions \eqref{fierz4/conditions} hold for $\lambda_{2\bar{i}}$.

Due to the presence of terms with different chirality, 
\eqref{oneloopterms/bosonized} has no continuous, even global, chiral symmetry. 
However, the \textit{discrete} chiral symmetry ($\chi$-symmetry) 
\begin{equation}
\begin{split}
&\psi \rightarrow \gamma^5\psi\\
&\bar{\psi} \rightarrow -\bar{\psi}\gamma^5\\
&\Sigma_{10},\Sigma_{11}^a,\Sigma_{30}^{\mu\nu},\Sigma_{31}^{a\,\mu\nu},\Sigma_{50},\Sigma_{51}^a \rightarrow
-\Sigma_{10},-\Sigma_{11}^a,-\Sigma_{30}^{\mu\nu},-\Sigma_{31}^{a\,\mu\nu},-\Sigma_{50},-\Sigma_{51}^a\\
&\Sigma_{20}^\mu,\Sigma_{21}^{a\,\mu},\Sigma_{40}^\mu,\Sigma_{41}^{a\,\mu} \rightarrow
\Sigma_{20}^\mu,\Sigma_{21}^{a\,\mu},\Sigma_{40}^\mu,\Sigma_{41}^{a\,\mu}
\end{split} \label{oneloopterms/chisymmetry}
\end{equation} 
still holds.

Now, to apply the definition \eqref{formalism/selfconsistent}, the fields must be split into background and deviation parts (the words `classical' and `quantum' look somewhat misleading in this context):
\begin{equation}
\begin{split}
\psi &\rightarrow \Psi + \psi\\
\Sigma_{\bar{i}}^{\bar{\alpha}} &\rightarrow \Sigma_{\bar{i}}^{\bar{\alpha}} + \varsigma_{\bar{i}}^{\bar{\alpha}},
\end{split}
\end{equation}
and the effective action can be written as
\begin{equation}
\begin{split}
\Gamma\left[\bar{\Psi},\Psi,\Sigma\right] &= 
\int\left[\bar{\Psi}\left(i\widehat{\partial} + \widehat{\Sigma}\right)\Psi
- \frac{\Sigma_{\bar{i}\bar{\alpha}}\Sigma_{\bar{i}}^{\bar{\alpha}}}{4\lambda_{2\bar{i}}}
\right]d^4x + W_1\left[\bar{\Psi},\Psi,\Sigma\right] \\
W_1\left[\bar{\Psi},\Psi,\Sigma\right] &= - iP_\text{1PI}\ln\funint{\bar\psi,\psi,\varsigma}\exp\left\{i\int\left[
\bar{\psi}\left(i\widehat{\partial} + \widehat{\Sigma}\right)\psi +
\bar{\Psi}\widehat{\varsigma}\psi + \bar{\psi}\widehat{\varsigma}\Psi + \bar{\psi}\widehat{\varsigma}\psi
- \frac{\varsigma_{\bar{i}\bar{\alpha}}\varsigma_{\bar{i}}^{\bar{\alpha}}}{4\lambda_{2\bar{i}}} - \right.\right.\\
&\left.\left.-\fundelta{W_1}{\Psi}\psi - \bar{\psi}\fundelta{W_1}{\bar{\Psi}}
- \fundelta{W_1}{\Sigma_{\bar{i}}^{\bar{\alpha}}}\varsigma_{\bar{i}}^{\bar{\alpha}}
\right]d^4x
\right\},
\end{split} \label{oneloopterms/defW}
\end{equation}
where $P_\text{1PI}$ is the operator that extracts one-particle irreducible diagrams from its operand.

Let
\begin{equation}
\mathcal{G}(x,y) = \frac{1}{i\widehat{\partial} + \widehat{\Sigma}} 
\end{equation}
and the vertex operator be defined as
\begin{equation}
\begin{split}
V &= i\frac{\delta}{\delta{J_{\bar{\Psi}}}}\tau_{i}^{\bar{\alpha}}
\frac{\delta}{\delta{J_\Psi}}\frac{\delta}{\delta{J_\Sigma^{\bar{i}\bar{\alpha}}}}
\end{split} \label{oneloopterms/vertex_op}
\end{equation}
where large Latin letters stand for combined spinor and $U(N)$ indices. Then perturbative expansion can be written, up to a constant, as
\begin{equation}
\begin{split}
&W_1\left[\bar{\Psi},\Psi,\Sigma\right] =  
-iP_\text{1PI}\ln\left\{\exp[iV]Z_\text{1-loop}\left[\bar{\Psi},\Psi,\Sigma, J_\Psi, J_{\bar{\Psi}}, J_\Sigma\right]\right\}
\left\vert{\begin{array}{l}
J_\Psi = -\fundelta{W_1}{\Psi}\\
J_{\bar{\Psi}} = -\fundelta{W_1}{\bar{\Psi}}\\
J_\Sigma^{\bar{i}\bar{\alpha}} = -\fundelta{W_1}{\Sigma_{{\bar{i}\bar{\alpha}}}}
\end{array}}\right.\\
&Z_\text{1-loop}\left[\bar{\Psi},\Psi,\Sigma, J_\Psi, J_{\bar{\Psi}}, J_\Sigma\right] = 
\funint{\bar\psi,\psi,\varsigma}\exp\left\{i\int\left[
\bar{\psi}\left(i\widehat{\partial} + \widehat{\Sigma}\right)\psi +
\bar{\Psi}\widehat{\varsigma}\psi + \bar{\psi}\widehat{\varsigma}\Psi 
- \frac{\varsigma_{\bar{i}\bar{\alpha}}\varsigma_{\bar{i}}^{\bar{\alpha}}}{4\lambda_{2\bar{i}}} + \right.\right.\\
&\left.\left.+J_\Psi\psi + \bar{\psi}J_{\bar{\Psi}} + J_{\bar{i}\bar{\alpha}}\varsigma_{\bar{i}}^{\bar{\alpha}}
\right]d^4x
\right\}
\end{split}
\end{equation}

The terms $\bar{\Psi}\widehat{\varsigma}\psi + \bar{\psi}\widehat{\varsigma}\Psi$
can be eliminated by shifting the integration variables
\begin{equation}
\begin{split}
\psi(x) \rightarrow \psi(x) - \int{\mathcal{G}(x,y)\widehat{\varsigma}(y)\Psi(y)}d^4y\\
\bar{\psi}(x) \rightarrow  
\bar{\psi}(x) - \int\bar{\Psi}(y)\widehat{\varsigma}(y)\mathcal{G}(y,x)d^4y.
\end{split}
\end{equation}
which results in
\begin{equation}
\begin{split}
&Z_\text{1-loop}\left[\bar{\Psi},\Psi,\Sigma, J_\Psi, J_{\bar{\Psi}}, J_\Sigma\right] = 
\funint{\bar{\psi},\psi, \varsigma}\exp\left\{i\int\left[
\bar{\psi}\left(i\widehat{\partial} + \widehat{\Sigma}\right)\psi
  + J_\Psi\psi + \bar{\psi}J_{\bar{\Psi}}\right]d^4x - \right.\\
&\left. - \iint\left\{
\varsigma_{\bar{i}}^{\bar{\alpha}}(x)\left[
\frac{\delta_{\bar{i}\bar{j}}\delta_{\bar{\alpha}}^{\bar{\beta}}}{4\lambda_{2\bar{i}}}\delta(x-y)
+ \bar{\Psi}(x)\tau_{\bar{i}\bar{\alpha}}\mathcal{G}(x,y)\tau_{\bar{j}}^{\bar{\beta}}\Psi(y)
\right]\varsigma_{\bar{j}\bar{\beta}}(y) + \right.\right.\\
&\left.\left. - J_{\Sigma\bar{i}\bar{\alpha}}(x)\delta(x-y)\varsigma_{\bar{i}}^{\bar{\alpha}}(y) 
+ J_{\Psi}(x)\mathcal{G}(x,y)\widehat{\varsigma}(y)\Psi(y)
+ \bar{\Psi}(x)\widehat{\varsigma}(y)\mathcal{G}(x,y)J_{\bar{\Psi}}(y)
\right\}d^4xd^4y\right\} 
\end{split}
\end{equation}
Let 
\begin{equation}
\begin{split}
&L_{\bar{i}\bar{\alpha}}^{\bar{j}\bar{\beta}}(x,y) = 
\delta_{\bar{i}}^{\bar{j}}\delta_{\bar{\alpha}}^{\bar{\beta}}{4\lambda_{2\bar{i}}}\delta(x-y)\\
&T_{\bar{i}\bar{\alpha}}^{\bar{j}\bar{\beta}}(x,y) = \bar{\Psi}(x)\tau_{\bar{i}\bar{\alpha}}\mathcal{G}(x,y)\tau_{\bar{j}}^{\bar{\beta}}\Psi(y)
\end{split} \label{oneloopterms/defLT}
\end{equation}
The deviations can now be integrated out easily, which yields
\begin{equation}
\begin{split}
&Z_\text{1-loop}\left[\bar{\Psi},\Psi,\Sigma, \bar{J}, J\right] = 
\text{Det}\left(i\widehat{\partial} + \widehat{\Sigma}\right)
\left[\text{Det}\left(L + T\right)\right]^{-\frac{1}{2}}
\exp\left[
-iJ_{\Psi}\frac{1}{i\widehat{\partial} + \widehat{\Sigma}}J_{\bar{\Psi}}
+ i\tilde{J}_\Sigma\frac{1}{L + T}\tilde{J}_\Sigma
\right]\\
&\tilde{J}_{\Sigma\bar{i}\bar{\alpha}} = J_{\Sigma\bar{i}\bar{\alpha}} - J_\Psi\mathcal{G}\tau_{\bar{i}\bar{\alpha}}\Psi 
- \bar{\Psi}\tau_{\bar{i}\bar{\alpha}}\mathcal{G}J_{\bar{\Psi}}
\end{split}
\end{equation}
Thus
\begin{equation}
\begin{split}
W_1\left[\bar{\Psi},\Psi,\Sigma\right] &= 
-i\Tr\ln\left(i\widehat{\partial} + \widehat{\Sigma}\right) + \frac{i}{2}\Tr\ln\left(L + T\right) -\\
&-iP_\text{1PI}\ln\left\{\exp(iV)\exp\left[
-iJ_{\Psi}\frac{1}{i\widehat{\partial} + \widehat{\Sigma}}J_{\bar{\Psi}}
+ i\tilde{J}_\Sigma\frac{1}{L + T}\tilde{J}_\Sigma
\right]\right\}\left|\begin{array}{l}
J_\Psi = -\fundelta{W_1}{\Psi}\\
J_{\bar{\Psi}} = -\fundelta{W_1}{\bar{\Psi}}\\
J_\Sigma^{\bar{i}\bar{\alpha}} = -\fundelta{W_1}{\Sigma_{{\bar{i}\bar{\alpha}}}}
\end{array}\right.
\end{split} \label{oneloopterms/loopexpr}
\end{equation}

\subsection{Diagrammatics} \label{oneloopdiags}
The expression \eqref{oneloopterms/loopexpr} implies that 
\begin{equation*}
\exp\left[
-iJ_{\Psi}\frac{1}{i\widehat{\partial} + \widehat{\Sigma}}J_{\bar{\Psi}}
+ i\tilde{J}_\Sigma\frac{1}{L + T}\tilde{J}_\Sigma
\right]
\end{equation*}
can be dropped from the final result after computing perturbative corrections:
as the effective action is the sum of 1-particle irreducible diagrams, so is $W_1$, 
and it is easy to see that substituting the expressions of $J$ as functional derivatives of $W_1$ into the exponent
will yield only reducible diagrams.

Thus, one can write
\begin{equation}
\Gamma_\text{1-loop}\left[\bar{\Psi}, \Psi, \Sigma\right] = 
\int\left[\bar{\Psi}\left(i\widehat{\partial} + \widehat{\Sigma}\right)\Psi
- \frac{\Sigma_{\bar{i}\bar{\alpha}}\Sigma_{\bar{i}}^{\bar{\alpha}}}{4\lambda_{2\bar{i}}}
\right]d^4x
-i\Tr\ln\left(i\widehat{\partial} + \widehat{\Sigma}\right) + \frac{i}{2}\Tr\ln\left(L + T\right),
\label{oneloopdiags/logarithms}
\end{equation}
and one-loop computation boils down to calculating 
\begin{equation}
\Tr\ln\left(A + B\right) = \Tr\ln\left[A\left(1 + A^{-1}B\right)\right] = \Tr\ln{A} 
+ \sum_{i=1}^\infty\frac{(-1)^{n - 1}}{n}\Tr\left[A^{-1}B\right]^n.
\end{equation}
Furthermore, cyclicity of trace implies that
\begin{equation}
\begin{split}
\Tr\ln\left(1 + L^{-1}T\right) &= \sum_{i=0}^\infty\frac{(-1)^{n - 1}}{n}\Tr\frac{1}{i\widehat{\partial} + \widehat{\Sigma}}\Phi =
\Tr\ln\left[\frac{1}{i\widehat{\partial} + \widehat{\Sigma}}\left(i\widehat{\partial} + \widehat{\Sigma} + \Phi\right)\right] =\\
&= - \Tr\ln\left(i\widehat{\partial} + \widehat{\Sigma}\right) + \Tr\ln\left[i\widehat{\partial} + \widehat{\Sigma} + \Phi\right]\\
\Phi(x) &= 4\lambda_{2\bar{i}}\tau_{\bar{i}\bar{\alpha}}\Psi(x)\bar{\Psi}(x)\tau_{\bar{i}}^{\bar{\alpha}}.
\end{split}
\end{equation}
Thus, the one-loop effective action can be rewritten as
\begin{equation}
\Gamma_\text{1-loop}\left[\bar{\Psi}, \Psi, \Sigma\right] = 
\int\left[\bar{\Psi}\left(i\widehat{\partial} + \widehat{\Sigma}\right)\Psi
- \frac{\Sigma_{\bar{i}\bar{\alpha}}\Sigma_{\bar{i}}^{\bar{\alpha}}}{4\lambda_{2\bar{i}}}
\right]d^4x
-\frac{3i}{2}\Tr\ln\left(i\widehat{\partial} + \widehat{\Sigma}\right) 
+ \frac{i}{2}\Tr\ln\left(i\widehat{\partial} + \widehat{\Sigma} + \Phi\right).
\end{equation}
This has the advantage of both traces being of the form $\Tr\ln\left[i\widehat{\partial} + \Omega\right]$.
The Fourier transform with the following conventions
\begin{equation}
\begin{split}
f(x) &= \frac{1}{(2\pi)^2}\int{f(p)e^{-ipx}d^4p}\\
S(x - y) &= \frac{1}{(2\pi)^4}\int\frac{\widehat{p}e^{-ip(x-y)}d^4p}{p^2}
\end{split}
\end{equation}
yields, in the diagrammatic language
\begin{equation}
\begin{split} 
	\feynmandiagram[inline = (a.base), horizontal=a to b] {
			a --[fermion] b,
	}; &= \frac{\widehat{p}}{p^2}\\
	\feynmandiagram[inline = (a.base), horizontal=a to b] {
			a --[plain, very thick] b,
	}; &= \widehat{\Omega}(q)\\
	\begin{tikzpicture}[baseline=(b)]
		\begin{feynman}[small]
			\vertex (a);
			\vertex [right=0.5cm of a] (b);
			\vertex [above right=0.7cm of b] (c);
			\vertex [below right=0.7cm of b] (d);
			\diagram* {
				(a) -- [plain, very thick] (b),
				(d) -- [fermion] (b) -- [fermion] (c)
			};
		\end{feynman}
	\end{tikzpicture} 
\end{split} \label{oneloopdiags/feynmanrules}
\end{equation}
The effective propagator $\mathcal{G}$ can be expanded as
\begin{equation}
\mathcal{G} = S\sum_{n = 0}^\infty\left[-{\widehat{\Sigma}}S\right]^n
\end{equation}

It is also clear that divergent diagrams in \eqref{oneloopdiags/logarithms} will be of fourth order in $S$ at most.
Thus, up to a constant
\begin{equation}
\begin{split}
&\Tr\ln\left[i\widehat{\partial} + \widehat{\Sigma}\right] = 
	\begin{tikzpicture}[baseline=(b)]
		\begin{feynman}
			\vertex (a);
			\vertex [right=0.5cm of a] (b);
			\vertex [above right=0.7cm of b] (c);
			\vertex [below right=0.7cm of b] (d);
			\diagram* {
				(a) -- [plain, very thick] (b),
				(b) -- [plain, quarter left] (c),
				(b) -- [plain, quarter right] (d),
				(d) -- [fermion, half right, looseness=1.8] (c),
			};
		\end{feynman} 
	\end{tikzpicture} 
	- \frac{1}{2}
	\begin{tikzpicture}[baseline=(b)]
		\begin{feynman}
			\vertex(a);
			\vertex [right=0.5cm of a] (b);
			\vertex [right=1cm of b] (c);
			\vertex [right=0.5cm of c] (d);
			\diagram* {
				(a) --[plain, very thick] (b),
				(b) --[fermion, half right] (c),
				(c) --[fermion, half right] (b),
				(c) --[plain, very thick] (d)
			};
		\end{feynman}
	\end{tikzpicture}
	+ \frac{1}{3}
	\begin{tikzpicture}[baseline=(b)] 
		\begin{feynman}
			\vertex(a);
			\vertex [right=0.5cm of a] (b);
			\vertex [above right=1cm of b] (c);
			\vertex [below right=1cmof b] (d);
			\vertex [above right=0.5cm of c] (e);
			\vertex [below right=0.5cm of d] (f);
			\diagram* {
				(a) --[plain, very thick] (b),
				(b) --[fermion] (d) --[fermion] (c) --[fermion] (b),
				(c) --[plain, very thick] (e),
				(d) --[plain, very thick] (f)
			};
		\end{feynman}
	\end{tikzpicture}
	- \frac{1}{4}
	\begin{tikzpicture}[baseline=(i0)]
		\begin{feynman}
			\vertex (i0);
			\vertex [above=1cm of i0] (i1);
			\vertex [below=2cm of i1] (i2);
			\vertex [right=2cm of i2] (i3);
			\vertex [above=2cm of i3] (i4);
			\vertex [below right=0.7cm of i1] (a);
			\vertex [above right=0.7cm of i2] (b);
			\vertex [above left=0.7cm of i3] (c);
			\vertex [below left=0.7cm of i4] (d);
			\diagram* {
				(i1)--[plain, very thick] (a),
				(i2)--[plain, very thick] (b),
				(i3)--[plain, very thick] (c),
				(i4)--[plain, very thick] (d),
				(a)--[fermion] (b) --[fermion] (c) --[fermion] (d) --[fermion] (a), 
			};
		\end{feynman}
	\end{tikzpicture} + O(1)\\
\end{split}
\end{equation}

\subsection{Renormalization} \label{onelooprenorm}
The BPHZ renormalization scheme relies on the fact that if the propagator has the form
\begin{equation}
D = \frac{1}{K_n(\partial_\mu)},
\end{equation}
where $K_n$ is some polynomial in derivative operators, then, in momentum space
\begin{equation}
\frac{1}{K(p + q)} = \frac{1}{K(p) + \Delta{K}(p,q)} = \frac{1}{K(p)} 
- \frac{1}{K(p)}\Delta(K)(p,q)\frac{1}{K(p)} + \frac{1}{K(p)}\Delta(K)(p,q)\frac{1}{K(p)}\Delta(K)(p,q)\frac{1}{K(p)} - \ldots.
\end{equation}
The crucial observation is that if all propagators in a given diagram are expanded in this manner, then only a finite number of initial terms in the resulting series will diverge. Thus it is possible to render the diagram finite by subtracting these terms.

In the effective action formalism, the effective propagator for the deviation field has the form
\begin{equation}
\mathcal{G} = \frac{1}{T(\partial_\mu) + \mu(\varphi)}.
\end{equation}
It now becomes clear that expanding in $\mu(\varphi)$ has the same properties as the BPHZ momentum expansion.
Thus it becomes possible to define field-dependent subtraction operator $M[\varphi]$
and BPHZ operator $\Lambda[\varphi] = -M[\varphi] + \lambda_\text{eff}[\varphi]$.

To be more precise, the equations \eqref{oneloopdiags/logarithms} suggest the following algorithm:
\begin{enumerate}
	\item Expand $\ln\left[i\widehat{\partial} + \Omega\right], \frac{1}{i\widehat{\partial} + \Omega}$ in $\Omega$.
	\item Go to momentum space.
	\item Expand all $\frac{1}{\widehat{p} - \widehat{q}}$ in $\widehat{q}$.
	\item Collect all divergent terms in the resulting expansion.
\end{enumerate}
The only catch here is masslessness of the original fermions which may result in infrared divergences. 
This problem can be tackled by Lowenstein prescription \cite{Lowenstein1975}, \cite{Lowenstein1976}, \cite{Blaschke2013}: 
adding the mass term $(1 - s)M$ to the fermionic propagator, 
where $s$ is the `softness' parameter. Then 
\begin{equation}
\frac{1}{i\widehat{\partial} - (1-s)M + \Omega} = \frac{1}{i\widehat{\partial} - M} 
- \frac{1}{i\widehat{\partial} - M}\left[sM + \Omega\right]\frac{1}{i\widehat{\partial} - M} + \ldots
\end{equation}
After subtracting the divergent terms, the resulting expression must be evaluated at $s = 1$ to get rid of the fictitious mass.

The subtraction operator can be expressed more easily in diagrammatic form. Let now
\begin{equation}
\begin{split}
	\feynmandiagram[inline = (a.base), horizontal=a to b] {
			a --[ghost] b,
	}; &= \frac{1}{\widehat{p} - M}\\
	\feynmandiagram[inline = (a.base), horizontal=a to b] {
			a --[ghost, insertion=0.5] b,
	}; &= \frac{1}{\widehat{p} - M}\widehat{q}\frac{1}{\widehat{p} - M}\\
	\feynmandiagram[inline = (a.base), horizontal=a to b] {
			a --[ghost, insertion=0.33, insertion=0.66] b,
	}; &= \frac{1}{\widehat{p} - M}\widehat{q}\frac{1}{\widehat{p} - M}\widehat{q}\frac{1}{\widehat{p} - M}
\end{split}
\end{equation}
and so on. Here $p$ denotes the running loop momentum and $q$ - external momentum carried by the line.
Then
\begin{equation}
\begin{split}
&\mathcal{M}\left[\bar{\Psi},\Psi,\Sigma\right]\Tr\ln\left[i\widehat{\partial} + \Omega\right] = 
	\begin{tikzpicture}[baseline=(b)]
		\begin{feynman}
			\vertex (a);
			\vertex [right=0.5cm of a] (b);
			\vertex [above right=0.7cm of b] (c);
			\vertex [below right=0.7cm of b] (d);
			\diagram* {
				(a) -- [plain, very thick] (b),
				(b) -- [ghost, quarter left] (c),
				(b) -- [ghost, quarter right] (d),
				(d) -- [ghost, half right, looseness=1.8] (c),
			};
		\end{feynman} 
	\end{tikzpicture} 
	- \frac{1}{2}
	\begin{tikzpicture}[baseline=(b)]
		\begin{feynman}
			\vertex(a);
			\vertex [right=0.5cm of a] (b);
			\vertex [right=1cm of b] (c);
			\vertex [right=0.5cm of c] (d);
			\diagram* {
				(a) --[plain, very thick] (b),
				(b) --[ghost, half right] (c),
				(c) --[ghost, half right] (b),
				(c) --[plain, very thick] (d)
			};
		\end{feynman}
	\end{tikzpicture}
	- \frac{1}{2}
	\begin{tikzpicture}[baseline=(b)]
		\begin{feynman}
			\vertex(a);
			\vertex [right=0.5cm of a] (b);
			\vertex [right=1cm of b] (c);
			\vertex [right=0.5cm of c] (d);
			\diagram* {
				(a) --[plain, very thick] (b),
				(b) --[ghost, half right] (c),
				(c) --[ghost, half right, insertion=0.5] (b),
				(c) --[plain, very thick] (d)
			};
		\end{feynman}
	\end{tikzpicture}
	- \frac{1}{2}
	\begin{tikzpicture}[baseline=(b)]
		\begin{feynman}
			\vertex(a);
			\vertex [right=0.5cm of a] (b);
			\vertex [right=1cm of b] (c);
			\vertex [right=0.5cm of c] (d);
			\diagram* {
				(a) --[plain, very thick] (b),
				(b) --[ghost, half right] (c),
				(c) --[ghost, half right, insertion=0.33, insertion=0.66] (b),
				(c) --[plain, very thick] (d)
			};
		\end{feynman}
	\end{tikzpicture} +\\
	&+ \frac{1}{3}
	\begin{tikzpicture}[baseline=(b)] 
		\begin{feynman}
			\vertex(a);
			\vertex [right=0.5cm of a] (b);
			\vertex [above right=1cm of b] (c);
			\vertex [below right=1cmof b] (d);
			\vertex [above right=0.5cm of c] (e);
			\vertex [below right=0.5cm of d] (f);
			\diagram* {
				(a) --[plain, very thick] (b),
				(b) --[ghost] (d) --[ghost] (c) --[ghost] (b),
				(c) --[plain, very thick] (e),
				(d) --[plain, very thick] (f)
			};
		\end{feynman}
	\end{tikzpicture}	
	+ \frac{1}{3}
	\begin{tikzpicture}[baseline=(b)] 
		\begin{feynman}
			\vertex(a);
			\vertex [right=0.5cm of a] (b);
			\vertex [above right=1cm of b] (c);
			\vertex [below right=1cmof b] (d);
			\vertex [above right=0.5cm of c] (e);
			\vertex [below right=0.5cm of d] (f);
			\diagram* {
				(a) --[plain, very thick] (b),
				(b) --[ghost] (d) --[ghost, insertion=0.5] (c) --[ghost] (b),
				(c) --[plain, very thick] (e),
				(d) --[plain, very thick] (f)
			};
		\end{feynman}
	\end{tikzpicture} 
	+ \frac{1}{3}
	\begin{tikzpicture}[baseline=(b)] 
		\begin{feynman}
			\vertex(a);
			\vertex [right=0.5cm of a] (b);
			\vertex [above right=1cm of b] (c);
			\vertex [below right=1cmof b] (d);
			\vertex [above right=0.5cm of c] (e);
			\vertex [below right=0.5cm of d] (f);
			\diagram* {
				(a) --[plain, very thick] (b),
				(b) --[ghost] (d) --[ghost] (c) --[ghost, insertion=0.5] (b),
				(c) --[plain, very thick] (e),
				(d) --[plain, very thick] (f)
			};
		\end{feynman}
	\end{tikzpicture}	
	- \frac{1}{4}
	\begin{tikzpicture}[baseline=(i0)]
		\begin{feynman}
			\vertex (i0);
			\vertex [above=1cm of i0] (i1);
			\vertex [below=2cm of i1] (i2);
			\vertex [right=2cm of i2] (i3);
			\vertex [above=2cm of i3] (i4);
			\vertex [below right=0.7cm of i1] (a);
			\vertex [above right=0.7cm of i2] (b);
			\vertex [above left=0.7cm of i3] (c);
			\vertex [below left=0.7cm of i4] (d);
			\diagram* {
				(i1)--[plain, very thick] (a),
				(i2)--[plain, very thick] (b),
				(i3)--[plain, very thick] (c),
				(i4)--[plain, very thick] (d),
				(a)--[ghost] (b) --[ghost] (c) --[ghost] (d) --[ghost] (a), 
			};
		\end{feynman}
	\end{tikzpicture} 
	\end{split} \label{onelooprenorm/subtraction}
\end{equation}
and one can write
\begin{equation}
R\Gamma_\text{1-loop}\left[\bar{\Psi}, \Psi, \Sigma\right] = 
\int\left[\bar{\Psi}\left(i\widehat{\partial} + \widehat{\Sigma}\right)\Psi
- \frac{\Sigma_{\bar{i}\bar{\alpha}}\Sigma_{\bar{i}}^{\bar{\alpha}}}{4\lambda_{2\bar{i}}}
+ \Lgr_R\left[\bar{\Psi},\Psi,\Sigma\right]\right]d^4x + \text{finite terms}.
\end{equation}

The only special remark that should be made here concerns the lack of tadpole diagrams with derivatives in \eqref{onelooprenorm/subtraction}. This is just the effect of subtracting in momentum space. It is possible, for example, to regularize
\begin{equation*}
\Tr\frac{1}{i\widehat{\partial}}\widehat{\Omega} = 
\int\Tr\frac{1}{i\widehat{\partial}}(x,x)\widehat{\Omega}(x)d^4x
\end{equation*} 
by separation of points first, then expand $\widehat{\Omega}(x + \xi)$ into Taylor series, of which only the first four terms will produce divergences.
This gives first order in $\widehat{\Omega}$ terms up to third derivatives. 
In general, it is possible (but much more tedious) to define subtraction in the coordinate space this way. 

$\Lgr_{R}\left(\bar{\Psi},\Psi,\Sigma\right)$ defines finite renormalization. 
In a point-independent BPHZ scheme this would be some purely differential operator which might 
alter numeric parameters of the theory at most. But here $\Lgr_{R}\left(\bar{\Psi},\Psi,\Sigma\right)$ 
must have the same overall structure as $\mathcal{M}\Tr\ln\Lambda$. Thus it may be of up to fourth order in $\Sigma$, up to eighths in $\Psi$, and up to third order in derivatives. This is lots of interesting physics. Thus it is important to study the general structure it may have.

It turns out that aside from the structure of subtractions, another requirement that can be imposed on $\Lgr_{R}$ is gauge symmetry.
Let
\begin{equation}
U = \exp\left\{i\left[f_0(x) + T^af_a(x)\right]\right\}
\end{equation}
The lagrangian of deviations in \eqref{oneloopterms/defW} 
\begin{equation}
\Lgr_\text{dev} = \bar{\psi}\left(i\widehat{\partial} + \widehat{\Sigma}\right)\psi +
\bar{\Psi}\widehat{\varsigma}\psi + \bar{\psi}\widehat{\varsigma}\Psi + \bar{\psi}\widehat{\varsigma}\psi
- \frac{\varsigma_{\bar{i}\bar{\alpha}}\varsigma_{\bar{i}}^{\bar{\alpha}}}{4\lambda_{2\bar{i}}}
\end{equation}
is invariant under gauge transformations
\begin{equation}
\begin{split}
&\psi \rightarrow U\psi\\
&\Psi \rightarrow U\Psi\\
&\Sigma \rightarrow U\Sigma{U}^+ - iU^+\widehat{\partial}U\\
&\varsigma \rightarrow U\varsigma{U}^+.
\end{split} \label{onelooprenorm/deftransform}
\end{equation}
Thus, gauge symmetry can be imposed on $W_1$ order by order. However, this suggests a very strong negative answer to the question whether gauge symmetry can emerge in the bosonized NJL model due to renormalization since the zero order of the effective action contains the explicit symmetry-breaking term 
\begin{equation*}
\frac{\Sigma_{\bar{i}\bar{\alpha}}\Sigma_{\bar{i}}^{\bar{\alpha}}}{4\lambda_{2\bar{i}}}.
\end{equation*}
Thus, if symmetry is imposed on $W_1$, it will be broken by the quadratic term, and if $W_1$ is not symmetric by construction, accidental symmetry would require lots of cancellations of symmetry-breaking terms at all orders of perturbation theory.

Let now
\begin{equation}
\begin{split}
&v_\mu = \Sigma_{20\mu}\\
&V_\mu = \Sigma_{21\mu}^aT^a\\
&D_\mu = \partial_\mu - iv_\mu - iV_\mu\\
&F_{0\mu\nu} = \partial_\mu{v_\nu} - \partial_\nu{v_\mu}\\
&F_{1\mu\nu} = \partial_\mu{V_\nu} - \partial_\nu{V_\mu} -i\comm{V_\mu}{V_\nu} 
\end{split}
\end{equation}
Gauge symmetry requirement implies that derivatives and $v,V$ may be present in $\Lgr_R$
as invariant combinations $D_\mu, F_{i\mu\nu}$ only. Furthermore, $D_\mu$ must act on bosons by commutation,
and on $\Psi$ by left multiplication.

Let $W$ be $\Sigma$ minus the vector part
\begin{equation}
\widehat{W} = \widehat{\Sigma} - \gamma^{\mu}\left(v_\mu + V_\mu\right).
\end{equation} 
Then, abstractly,
\begin{equation}
\Lgr_R = \Lgr_{R1}\left(D_\mu,\widehat{W},\Psi, \bar{\Psi}\right) + \Lgr_{R2}\left(D_\mu,\widehat{W},\Psi, \bar{\Psi}\right) 
+ \Lgr_{R3}\left(D_\mu,\widehat{W},\Psi, \bar{\Psi}\right) + \Lgr_{R4}\left(D_\mu,\widehat{W},\Psi, \bar{\Psi}\right),
\end{equation}
where $\Lgr_{Rn}\left(D_\mu,\widehat{W},\Psi,\bar{\Psi}\right)$ is some Lorentz-, $U(N)$-, and gauge-invariant $n$-form
This can be further restricted by discrete symmetries. The $\chi$-symmetry \eqref{oneloopterms/chisymmetry} precludes first- and third-order terms. This is because each $\chi$-invariant object has one Lorentz index, therefore they must be always coupled for contraction. Thus, monomials of odd combined power will have odd number of sign-changing multipliers and will change sign under $\chi$-transform.

Imposing discrete and gauge symmetries, and dropping full derivatives, one can write
\begin{equation}
\begin{split}
&\Lgr_R = \frac{g_{1\bar{i}}W_{\bar{i}\bar{\alpha}}W_{\bar{i}}^{\bar{\alpha}}}{2}
+ g_2\bar{\Psi}i\widehat{D}\Psi + g_{3\bar{i}}W_{\bar{i}\bar{\alpha}}\bar{\Psi}\tau_{\bar{i}}^{\bar{\alpha}}\Psi
+ g_{4i}\bar{\Psi}\Gamma_{i\bar{\alpha}}\Psi\bar{\Psi}\Gamma_{i}^{\bar{\alpha}}\Psi
+ g_{5\mathbf{A}\bar{i}\bar{j}\bar{k}\bar{l}}t_{\mathbf{A}\bar{\alpha}\bar{\beta}\bar{\gamma}\bar{\delta}}
W_{\bar{i}}^{\bar{\alpha}}W_{\bar{j}}^{\bar{\beta}}W_{\bar{k}}^{\bar{\gamma}}W_{\bar{l}}^{\bar{\delta}} +\\
&+ g_{6\mathbf{A}\bar{i}\bar{j}\bar{k}\bar{l}}t_{\mathbf{A}\bar{\alpha}\bar{\beta}\bar{\gamma}\bar{\delta}}
W_{\bar{i}}^{\bar{\alpha}}W_{\bar{j}}^{\bar{\beta}}W_{\bar{k}}^{\bar{\gamma}}\bar{\Psi}\tau_{\bar{l}}^{\bar{\delta}}\Psi
+ g_{7\mathbf{A}\bar{i}\bar{j}\bar{k}\bar{l}}t_{\mathbf{A}\bar{\alpha}\bar{\beta}\bar{\gamma}\bar{\delta}}
W_{\bar{i}}^{\bar{\alpha}}W_{\bar{j}}^{\bar{\beta}}
\bar{\Psi}\tau_{\bar{k}}^{\bar{\gamma}}\Psi\bar{\Psi}\tau_{\bar{l}}^{\bar{\delta}}\Psi +\\
&+ g_{8\mathbf{A}\bar{i}\bar{j}\bar{k}\bar{l}}t_{\mathbf{A}\bar{\alpha}\bar{\beta}\bar{\gamma}\bar{\delta}}
W_{\bar{i}}^{\bar{\alpha}}\bar{\Psi}\tau_{\bar{j}}^{\bar{\beta}}\Psi
\bar{\Psi}\tau_{\bar{k}}^{\bar{\gamma}}\Psi\bar{\Psi}\tau_{\bar{l}}^{\bar{\delta}}\Psi
+ g_{9\mathbf{A}\bar{i}\bar{j}\bar{k}\bar{l}}t_{\mathbf{A}\bar{\alpha}\bar{\beta}\bar{\gamma}\bar{\delta}}
\bar{\Psi}\tau_{\bar{i}}^{\bar{\alpha}}\Psi\bar{\Psi}\tau_{\bar{j}}^{\bar{\beta}}\Psi
\bar{\Psi}\tau_{\bar{k}}^{\bar{\gamma}}\Psi\bar{\Psi}\tau_{\bar{l}}^{\bar{\delta}}\Psi +\\
&+ g_{10\mathbf{A}\mathbf{B}\bar{i}\bar{j}\bar{k}}\theta_{\mathbf{A}}^{\bar{\kappa}\bar{\lambda}\mu\bar{\nu}}\chi_{\mathbf{B}}^{abc}
W_{\bar{i}\bar{\alpha}}^aW_{\bar{j}\bar{\beta}}^b\left[D_\mu{}W_{\bar{k}\bar{\nu}}\right]^c
+ g_{11\mathbf{A}\bar{i}\bar{j}}\theta_{\mathbf{A}}^{\kappa\bar{\lambda}\mu\bar{\nu}}
\left[D_{\kappa}W_{\bar{i}\bar{\lambda}}\right]^a\left[D_\mu{}W_{\bar{j}\bar{\nu}}\right]^a +\\
&+ g_{12\mathbf{A}\bar{i}\bar{j}\bar{k}}t_{\mathbf{A}}^{\bar{\alpha}\bar{\beta}\bar{\gamma}\mu}
W_{\bar{i}\bar{\alpha}}W_{\bar{j}\bar{\beta}}\bar{\Psi}\tau_{\bar{k}\bar{\gamma}}D_\mu\Psi
+ g_{13\mathbf{A}\bar{i}\bar{j}\bar{k}}t_{\mathbf{A}}^{\bar{\alpha}\bar{\beta}\bar{\gamma}\mu}
W_{\bar{i}\bar{\alpha}}\bar{\Psi}\tau_{\bar{j}\bar{\beta}}\Psi\bar{\Psi}\tau_{\bar{k}\bar{\gamma}}D_\mu\Psi +\\
&+ g_{14\mathbf{A}\bar{i}\bar{j}\bar{k}}t_{\mathbf{A}}^{\bar{\alpha}\bar{\beta}\bar{\gamma}\mu}
\bar{\Psi}\tau_{\bar{i}\bar{\alpha}}\Psi\bar{\Psi}\tau_{\bar{j}\bar{\beta}}\Psi\bar{\Psi}\tau_{\bar{k}\bar{\gamma}}D_\mu\Psi
+ g_{15\mathbf{A}\bar{i}\bar{j}\bar{k}\bar{l}}t_{\mathbf{A}}^{\bar{\alpha}\lambda\bar{\beta}\bar{\gamma}\mu\bar{\delta}}
\left(\bar{\Psi}\tau_{\bar{i}\bar{\alpha}}D_\lambda\tau_{\bar{j}\bar{\beta}}\Psi\right)
\left(\bar{\Psi}\tau_{\bar{k}\bar{\gamma}}D_\mu\tau_{\bar{l}\bar{\delta}}\Psi\right) +\\
&+ g_{16\mathbf{A}\mathbf{B}ijk}\theta_{\mathbf{A}}^{\bar{\alpha}\lambda\bar{\beta}\bar{\gamma}\mu\bar{\delta}}\chi_{\mathbf{B}}^{abc}
\left[D_{\lambda}W_{\bar{i}\bar{\alpha}}\right]^a
\bar{\Psi}\Gamma_{j\bar{\beta}}G^bD_\mu\Gamma_{k\bar{\gamma}}G^c\Psi
+ g_{17\mathbf{A}\bar{i}\bar{j}\bar{k}\bar{l}}t_{\mathbf{A}}^{\bar{\alpha}\lambda\bar{\beta}\bar{\gamma}\mu\bar{\delta}}
\left(\bar{\Psi}\tau_{\bar{i}\bar{\alpha}}\Psi\right)
\left(\bar{\Psi}\tau_{\bar{j}\bar{\beta}}D_\lambda\tau_{\bar{k}\bar{\gamma}}D_\mu\tau_{\bar{l}\bar{\delta}}\Psi\right) +\\
&+ g_{18\mathbf{A}\bar{i}\bar{j}\bar{k}\bar{l}}t_{\mathbf{A}}^{\bar{\alpha}\lambda\bar{\beta}\bar{\gamma}\mu\bar{\delta}}
W_{\bar{i}\bar{\alpha}}
\bar{\Psi}\tau_{\bar{j}\bar{\beta}}D_\lambda\tau_{\bar{k}\bar{\gamma}}D_\mu\tau_{\bar{l}\bar{\delta}}\Psi
+ g_{19\mathbf{A}\bar{i}\bar{j}\bar{k}\bar{l}}t_{\mathbf{A}}^{\bar{\alpha}\kappa\bar{\beta}\lambda\bar{\gamma}\mu\bar{\delta}}
\bar{\Psi}\tau_{\bar{i}\bar{\alpha}}D_{\kappa}\tau_{\bar{j}\bar{\beta}}D_\lambda\tau_{\bar{k}\bar{\gamma}}D_\mu\tau_{\bar{l}\bar{\delta}}\Psi +\\ 
&+ \frac{g_{20}}{4}F_{0\mu\nu}F_0^{\mu\nu} + \frac{g_{21}}{2}\Tr{}F_{1\mu\nu}F_1^{\mu\nu}
\end{split} \label{onelooprenorm/renormterms}
\end{equation}
Here the bold indices $\mathbf{A},\mathbf{B}$ enumerate basis Lorentz- and U(N)-invariant forms.
Whenever possible, the forms were combined into $t_{\mathbf{A}\bar{\alpha}\bar{\beta\ldots}}$. 
In this case, $\bar{\alpha},\bar{\beta}\ldots$ stand for combined Lorentz and U(N) indices.
Where $U(N)$ indices are explicit, they span the extended basis \eqref{fierz4/extbasis}. 

A more detailed discussion of \eqref{onelooprenorm/renormterms} is deferred till conclusions.
It suffices to say for now that bosonic fields get nontrivial dynamic from finite renormalization.

\subsection{Some remarks on QCD}
The bosonisation scheme that was developed for the NJL model in \ref{oneloopterms} -- \ref{onelooprenorm} 
can be repeated for QCD lagrangian with only minor changes. One has to start from 
\begin{equation}
\begin{split}
Z\left[\left\{J\right\}\right] &= 
\funint{\bar{\psi},\psi, A, \omega}\exp\left\{i\int\left[-\frac{1}{2g^2}\Tr{F_{\mu\nu}F^{\mu\nu}} + \bar{\psi}i\widehat{D}\psi +
\Lgr_{GF}(A) + \Lgr_\text{FP}(\omega, A) + \right.\right.\\
&\left.\left.+ J_{\psi1}\psi +\bar{\psi}J_{\bar{\psi}1} + J_{A1\mu}^aA^{a\mu} 
+ 2\lambda_{2\bar{i}}J_{\psi2\bar{i}\bar{\alpha}}\bar{\psi}\tau_{\bar{i}}^{\bar{\alpha}}\psi
- \lambda_{2\bar{i}}J_{\psi2\bar{i}\bar{\alpha}}J^{\bar{\alpha}}_{\psi2\bar{i}} 
- 2\lambda_{221}J_{\psi221}^{a\mu}A^a_\mu
\right]d^4x\right\},
\end{split}
\label{qcd/initial}
\end{equation}  
where $\Lgr_{GF}(A) + \Lgr_\text{FP}(\omega, A)$ are gauge-fixing and Faddeev-Popov terms, and the convention for covariant derivative is
\begin{equation}
D_\mu = \partial_\mu - iA_\mu.
\end{equation}
Then, insertion of
\begin{equation}
\begin{split}
&\text{const} = \funint{\Sigma}\exp\left\{i\int\left[\sum_{\bar{i} \neq 21}
\left(\frac{\Sigma_{\bar{i}\bar{\alpha}}}{2\lambda_{\bar{i}}^\frac{1}{2}} 
- \lambda_{\bar{i}}^\frac{1}{2}\bar{\psi}\tau_{\bar{i}\bar{\alpha}}\psi 
+ \lambda_{\bar{i}}^\frac{1}{2}J_{\psi2\bar{i}\bar{\alpha}}\right)
\left(\frac{\Sigma_{\bar{i}}^{\bar{\alpha}}}{2\lambda_{\bar{i}}^\frac{1}{2}} 
- \lambda_{\bar{i}}^\frac{1}{2}\bar{\psi}\tau_{\bar{i}}^{\bar{\alpha}}\psi 
+ \lambda_{\bar{i}}^\frac{1}{2}J_{\psi2\bar{i}}^{\bar{\alpha}}\right) + \right.\right.\\
&\left.\left.+
\left(\frac{\Sigma_{21\mu}^a + A_\mu^a}{2\lambda_{21}^\frac{1}{2}} 
- \lambda_{21}^\frac{1}{2}\bar{\psi}\gamma_{\mu}T^a\psi 
+ \lambda_{21}^\frac{1}{2}J_{\psi221\mu}^a\right)
\left(\frac{\Sigma_{21}^{a\mu} + A^{a\mu}}{2\lambda_{21}^\frac{1}{2}} 
- \lambda_{21}^\frac{1}{2}\bar{\psi}\gamma^{\mu}T^a\psi 
+ \lambda_{21}^\frac{1}{2}J_{\psi221}^{a\mu}\right)
\right]d^4x\right\}
\end{split} \label{qcd/insertion}
\end{equation}
yields
\begin{equation}
\begin{split}
Z\left[\left\{J\right\}\right] &= 
\funint{\bar{\psi},\psi, A, \omega, \Sigma}\exp\left\{i\int\left[-\frac{1}{2g^2}\Tr{F_{\mu\nu}F^{\mu\nu}}
+ \Lgr_{GF}(A) + \Lgr_\text{FP}(\omega, A) + \right.\right.\\
&\left.\left.+ \bar{\psi}\left(i\widehat{\partial} - \widehat{\Sigma}\right)\psi 
+ \sum_{\bar{i} \neq 21}\frac{\Sigma_{\bar{i}\bar{\alpha}}\Sigma_{\bar{i}}^{\bar{\alpha}}}{4\lambda_{\psi2\bar{i}}}
+ \frac{\left(\Sigma_{21\mu}^a + A_\mu^a\right)\left(\Sigma_{21}^{a\mu} + A^{a\mu}\right)}{4\lambda_{\psi221}}
+ \lambda_{\psi2\bar{i}}\bar{\psi}\tau_{\bar{i}\bar{\alpha}}\psi\bar{\psi}\tau_{\bar{i}}^{\bar{\alpha}}\psi +\right.\right.\\
&\left.\left.
+ J_{\psi1}\psi +\bar{\psi}J_{\bar{\psi}1} + J_{A1\mu}^aA^{a\mu} + J_{\psi2\bar{i}\bar{\alpha}}\Sigma_{\bar{i}}^{\bar{\alpha}}
\right]d^4x\right\}.
\end{split}
\end{equation} 
Fierz identities make it possible to eliminate the four-fermion interaction terms by imposing
\begin{equation}
	\begin{bmatrix}
		\lambda_{210}\\\lambda_{220}\\\lambda_{230}\\\lambda_{240}\\\lambda_{250}
	\end{bmatrix} + 
	\begin{bmatrix}
		-\frac{N + 4}{8N}&-\frac{1}{2}&-\frac{3}{2}&\frac{1}{2}&-\frac{1}{8}\\
		-\frac{1}{2}&\frac{N - 2}{4N}&0&\frac{1}{4}&\frac{1}{8}\\
		-\frac{1}{16}&0&\frac{N - 2}{4N}&0&-\frac{1}{16}\\
		\frac{1}{8}&\frac{1}{4}&0&\frac{N - 2}{4N}&-\frac{1}{8}\\
		-\frac{1}{8}&\frac{1}{2}&-\frac{3}{2}&-\frac{1}{2}&-\frac{N + 4}{8N}
	\end{bmatrix}
	\begin{bmatrix}
		\lambda_{211}\\\lambda_{221}\\\lambda_{231}\\\lambda_{241}\\\lambda_{251}		
	\end{bmatrix} = 0.
	\label{qcd/fierz_conditions}
\end{equation}
The resulting generating functional is given by
\begin{equation}
\begin{split}
Z\left[\left\{J\right\}\right] &= 
\funint{\bar{\psi},\psi, A, \omega, \Sigma}\exp\left\{i\int\left[-\frac{1}{2g^2}\Tr{F_{\mu\nu}F^{\mu\nu}}
+ \Lgr_{GF}(A) + \Lgr_\text{FP}(\omega, A) + \right.\right.\\
&\left.\left.+ \bar{\psi}\left(i\widehat{\partial} - \widehat{\Sigma}\right)\psi 
+ \sum_{\bar{i} \neq 21}\frac{\Sigma_{\bar{i}\bar{\alpha}}\Sigma_{\bar{i}}^{\bar{\alpha}}}{4\lambda_{\psi2\bar{i}}}
+ \frac{\left(\Sigma_{21\mu}^a + A_\mu^a\right)\left(\Sigma_{21}^{a\mu} + A^{a\mu}\right)}{4\lambda_{\psi221}} +\right.\right.\\
&\left.\left.
+ J_{\psi1}\psi +\bar{\psi}J_{\bar{\psi}1} + J_{A1\mu}^aA^{a\mu} + J_{\psi2\bar{i}\bar{\alpha}}\Sigma_{\bar{i}}^{\bar{\alpha}}
\right]d^4x\right\},
\end{split}
\label{qcd/bosonized}
\end{equation} 
which is essentially the same as \eqref{oneloopterms/bosonized}, save for the gluons.
The one-loop finite renormalization produced by integration over the fermionic deviations
will not differ from \eqref{onelooprenorm/renormterms}. 
Both models' mesonic sectors will be equivalent at one loop level.
The difference between QCD and NJL will manifest itself at higher perturbative orders
where the conditions \eqref{qcd/fierz_conditions} are likely to result in diagram cancellations
as it should be in a renormalizable theory.

There are two possible approaches to gluon self-interactions here. 
The minimal one is to treat them perturbatively. But one can do a bosonization-like transform for gluons as well,
as it was pointed out by Kondo \cite{Kondo2003}.

Let $\sim$ denote the bosonization duality, i.e. in \eqref{qcd/bosonized}
\begin{equation}
\begin{split}
&\Sigma_{\bar{i}\bar{\alpha}} \sim \lambda_{2\bar{i}}\bar{\psi}\tau_{\bar{i}\bar{\alpha}}\psi, \bar{i} \neq 21\\
&\Sigma_{21\mu}^a \sim -A_\mu^a + \lambda_{221}\bar{\psi}\gamma_{\mu}T^a\psi
\end{split}
\end{equation}

Then it is possible to introduce the fields
\begin{equation}
\begin{split}
&B_{\mu\nu} \sim \frac{1}{\theta}\left[\partial_{\mu}A_\nu - \partial_{\nu}A_\mu\right] -i\theta\comm{A_\mu}{A_\nu}\\
&T_{\mu\nu} \sim \acomm{A_\mu}{A_\nu} - \frac{\eta_{\mu\nu}}{2}A_{\mu}A^\mu\\
&S \sim A_{\mu}A^\mu
\end{split}
\end{equation}
The identity 
\begin{equation}
\comm{A_\mu}{A_\nu}\comm{A^\mu}{A^\nu} = \acomm{A_\mu}{A_\nu}\acomm{A^\mu}{A^\nu} - 4\left(A_{\mu}A^\mu\right)^2,
\end{equation}
which can be considered the Fierz identity for gluons,
implies that it is possible to eliminate gluon self-interaction terms completely by inserting 
\begin{equation}
\begin{split}
&\text{const} = \funint{B, T, S}\exp\left\{i\int\left[
\Tr\left(\frac{B_{\mu\nu}}{\lambda_B^\frac{1}{2}} - \lambda_B^\frac{1}{2}(\sim{B_{\mu\nu}})\right)
\left(\frac{B^{\mu\nu}}{\lambda_B^\frac{1}{2}} - \lambda_B^\frac{1}{2}(\sim{B^{\mu\nu}})\right) +\right.\right.\\
&\left.\left.
+ \Tr\left(\frac{T_{\mu\nu}}{\lambda_T^\frac{1}{2}} - \lambda_T^\frac{1}{2}(\sim{T_{\mu\nu}})\right)
\left(\frac{T^{\mu\nu}}{\lambda_T^\frac{1}{2}} - \lambda_T^\frac{1}{2}(\sim{T^{\mu\nu}})\right)
+ \Tr\left(\frac{S}{\lambda_S^\frac{1}{2}} - \lambda_S^\frac{1}{2}(\sim{S})\right)^2
\right]\right\}
\end{split} \label{qcd/gluontransform}
\end{equation}
under a suitable choice of $\lambda$. It turns out that this does not fix all of $\lambda_B,\lambda_T,\lambda_S,\theta$ 
simultaneously and the resulting `bosonized' lagrangian
depending on the new fields $B,T,S$ will still have one extra free scalar parameter.

Going down this road further can be a topic of another research. The main lesson of this subsection is that bosonization, 
far from being an ad hoc trick for the NJL model, is surprisingly universal. Implications and generalizations of this are discussed in the next section.

\section{The composite fields formalism for NJL and other models} \label{compfields}

\subsection{Generalizations}

\label{comp_formalism}
As it has been demonstrated in the preceding sections, bosonization is not unique to the
NJL model. Furthermore, it is not exactly `bosonization', as, being a same-point case of Cornwall-Jackiw-Tomboulis' composite fields formalism,
it may actually affect bosonic fields too. 

As composite fields formalism can be generalized to include infinite series of nonlocal terms, as proposed by Calzetta and Hu, so can bosonization. Starting from the field multiplet $\phi$, the classical action $S\left[\phi\right]$,
and a set of $\mu_k(\phi)$ -- local monomials in fields and their derivatives, one can construct the generating functional
\begin{equation}
Z\left[\{J\}\right] = \funint{\phi}\exp\left\{iS\left[\phi\right] + i\int\phi{}J_{\phi}d^4x + i\sum_{k}\lambda_k\int\left[2\mu_k(\phi)J_{\mu_k} - J_{\mu_k}^2\right]d^4x\right\} \label{comp_formalism/Z}
\end{equation}
by adding a source for each $\mu_k$. 
The generating functional \eqref{comp_formalism/Z} differs from \eqref{intro/Schwinger-Keldysh} as it includes free numeric parameters $\lambda_k$ and terms quadratic in sources.
This modification is a simple reparameterization of sources -- a change of variables describing the quantum state. 
Then inserting
\begin{equation}
\begin{split}
&\text{const} = \funint{\bar{\Sigma}}\exp\left\{i\sum_k\int\left[\frac{\Sigma_k}{2\lambda_k^\frac{1}{2}} 
- \lambda_k^\frac{1}{2}\mu_k(\phi) - \lambda_k^\frac{1}{2}J_{\mu_k}\right]^2d^4x\right\}\\
&\bar{\Sigma} = \left\{\Sigma_k\right\}
\end{split}
\end{equation}
yields
\begin{equation}
Z\left[\{J\}\right] = \funint{\phi,\bar{\Sigma}}\exp\left\{iS\left[\phi\right] + i\int\phi{}J_{\phi}d^4x 
+ \sum_k\int\left[\lambda_k\mu_k^2(\phi) + \frac{\Sigma_k^2}{4\lambda_k} + J_{\mu_k}\Sigma_k\right]d^4x\right\}.
\end{equation}
The coefficients $\lambda_k$ can be tuned for maximum number of cancellations in the modified action.

In principle, infinite sequence of $\mu_k$ corresponding to all possible local combinations of 
field and its derivatives is possible. As \eqref{comp_formalism/Z} suggests, if there are two sets
\begin{equation}
\begin{split}
&M_1 = \{\mu_{k}(\phi)\}\\
&M_2 = \{\mu_{l}(\phi)\}\\
&M_1 \subset M_2
\end{split}
\end{equation}
then the generating functional for the smaller set can be obtained from the larger one's by putting some $\lambda_m = 0$.
Thus the smaller composite fields extension of the given model defines a sector in a larger one.
This is also true for the original field theory without composite fields:
it can be viewed as the minimal sector of the extended model family.
The particular choice of the composite fields inventory and of corresponding constants should then be the matter of research interest.
In the remainder of this section the generalized formalism is exemplified by a bosonized NJL model with addfitional composite fields corresponding to cubic combinations of fermions. 

\subsection{The 6-order Fierz identities} \label{fierz6}

Adding cubic composite fields to an NJL-like model will produce terms of the form

\begin{equation}
\Phi_{6\bar{i}\bar{j}\bar{k}}^{\bar{\alpha}\bar{\beta}\bar{\gamma}} = 
\bar{\psi}\tau_{\bar{i}}^{\bar{\alpha}}\psi\bar{\psi}\tau_{\bar{j}}^{\bar{\beta}}\psi\bar{\psi}\tau_{\bar{k}}^{\bar{\gamma}}\psi.
\end{equation} 

Only some linear combinations of these will possess necessary symmetry properties though.
According to \ref{fierz4/lorentzforms}, \ref{fierz4/tracedecomposition}, these will be

\begin{equation}
\Lgr_{P\mathbf{A}klm} = t_{\mathbf{A}}^{\bar{\alpha}\bar{\beta}\bar{\gamma}}
\bar{\psi}_{Pi_1}\Gamma_{k\bar{\alpha}}\psi_{i_1}
\bar{\psi}_{Pi_2}\Gamma_{l\bar{\beta}}\psi_{i_2}
\bar{\psi}_{Pi_3}\Gamma_{m\bar{\gamma}}\psi_{i_3} \label{fierz6/genterms}
\end{equation}

Here, $\mathbf{A}$ labes basis Lorentz-invariant forms built from $\eta^{\mu\nu}, \epsilon^{\kappa\lambda\mu\nu}$, 
$\bar{\alpha}\ldots$ are Lorentz multi-indices, and $P$ is some permutation of $i_1,i_2,i_3$.
Then transformations along the lines of \eqref{fierz4/deltadecomposition}-\eqref{fierz4/identities}
can be used to express all \eqref{fierz6/genterms} through

\begin{equation}
\Lgr_{\mathbf{A}klm} = t_{\mathbf{A}}^{\bar{\alpha}\bar{\beta}\bar{\gamma}}
\bar{\psi}_{i_1}\Gamma_{k\bar{\alpha}}\psi_{i_1}
\bar{\psi}_{i_2}\Gamma_{l\bar{\beta}}\psi_{i_2}
\bar{\psi}_{i_3}\Gamma_{m\bar{\gamma}}\psi_{i_3} \label{fierz6/gammaterms}
\end{equation}

It is quite clear that this reasoning is valid for any fermion self-interaction of arbitrary order:
by virtue of \ref{fierz4/tracedecomposition} any Lorentz- and $U(N)$-invariant scalar depending on $\psi, \bar{\psi}$ only can be reduced to terms of \eqref{fierz6/gammaterms} type, where only Dirac basis matrices are sandwiched between $\bar{\psi},\psi$.
There remains a question whether the latter are really independent though.

The answer depends on the number of colors $N$. The classical Fierz identities for colorless fermions are a consequence
of completeness of the basis $\left\{\Gamma^{\bar{\alpha}}\right\}$ in the space of $4\times{4}$ matrices. But in the presence of colors the complete basis is $\Gamma_iG^a$.
$\psi\bar{\psi}$ is expressible through $5N^2$ terms $\bar{\psi}\Gamma_iG^a\psi$, not through $\bar{\psi}_l\Gamma_k\psi_l$ of which there are only $5N$. The expression 
\begin{equation}
F_{i_1\ldots{}i_m;k_1\ldots{}k_m} = \bar{\psi}_{i_1}\Gamma_{k_1}\psi_{i_1}\ldots\bar{\psi}_{i_m}\Gamma_{k_m}\psi_{i_m}
\label{fierz6/genmonomial}
\end{equation}
without summation over repeated indices is in general a product of $m$ totally independent factors. However, if $m > N$, some of the indices $i_k$ must coincide. Then \eqref{fierz6/genmonomial} will contain some combinations that will be subject to single-color Fierz identities. For sixth-order interaction terms this will happen if color symmetry is $U(2)$. This special case is to be studied elsewhere. The results below apply to higher color counts for which the only Fierz identities will be reduction formulas
expressing \eqref{fierz6/genterms} through \eqref{fierz6/gammaterms}. Explicit derivation relies on expressions 
\begin{equation}
\chi_{i_1i_2i_3\,j_1j_2j_3} = \chi^{abc}G^a_{i_1j_1}G^b_{i_2j_2}G^c_{i_3j_3} =
\sum_{P}\chi_P\delta_{i_1Pj_1}\delta_{i_2Pj_2}\delta_{i_3Pj_3}
\end{equation}
for $U(N)$-invariant forms.
As a corollary of \ref{fierz4/tracedecomposition}, there are only six basis ones:
\begin{equation}
\begin{split}
&\chi_{0\,i_1i_2i_3\,j_1j_2j_3} = \delta_{i_1j_1}\delta_{i_2j_2}\delta_{i_3j_3}\\
&\chi_{1\,i_1i_2i_3\,j_1j_2j_3} = \delta_{i_1j_1}T^a_{i_2j_2}T^a_{i_3j_3} = 
\frac{1}{2}\delta_{i_1j_1}\delta_{i_2j_3}\delta_{i_3j_2} 
- \frac{1}{2N}\delta_{i_1j_1}\delta_{i_2j_2}\delta_{i_3j_3}\\
&\chi_{2\,i_1i_2i_3\,j_1j_2j_3} = \delta_{i_2j_2}T^a_{i_1j_1}T^a_{i_3j_3} = 
\frac{1}{2}\delta_{i_2j_2}\delta_{i_1j_3}\delta_{i_3j_1} 
- \frac{1}{2N}\delta_{i_1j_1}\delta_{i_2j_2}\delta_{i_3j_3}\\
&\chi_{3\,i_1i_2i_3\,j_1j_2j_3} = \delta_{i_3j_3}T^a_{i_1j_1}T^a_{i_2j_2} = 
\frac{1}{2}\delta_{i_3j_3}\delta_{i_1j_2}\delta_{i_2j_1} 
- \frac{1}{2N}\delta_{i_1j_1}\delta_{i_2j_2}\delta_{i_3j_3}\\
&\chi_{4\,i_1i_2i_3\,j_1j_2j_3} = f^{abc}T^a_{i_1j_1}T^b_{i_2j_2}T^c_{i_3j_3} = 
\frac{i}{4}\left[\delta_{i_1j_2}\delta_{i_2j_3}\delta_{i_3j_1} 
- \delta_{i_1j_3}\delta_{i_2j_1}\delta_{i_3j_2}\right]\\
&\chi_{5\,i_1i_2i_3\,j_1j_2j_3} = d^{abc}T^a_{i_1j_1}T^b_{i_2j_2}T^c_{i_3j_3} 
=\frac{1}{4}\left[\delta_{i_1j_2}\delta_{i_2j_3}\delta_{i_3j_1} 
+ \delta_{i_1j_3}\delta_{i_2j_1}\delta_{i_3j_2}\right] -\\
&- \frac{1}{2N}\delta_{i_1j_1}\delta_{i_2j_3}\delta_{i_3j_2}
- \frac{1}{2N}\delta_{i_2j_2}\delta_{i_1j_3}\delta_{i_3j_1}
- \frac{1}{2N}\delta_{i_3j_3}\delta_{i_1j_2}\delta_{i_2j_1}
+ \frac{1}{N^2}\delta_{i_1j_1}\delta_{i_2j_2}\delta_{i_3j_3}.
\end{split} \label{fierz6/chiforms}
\end{equation}
The expressions for $\chi_4,\chi_5$ can be derived by noticing that $\chi_{4,5}(X_1^aT^a,X_2^bT^b,X_3^cT^c)$
are expressible through trace and (anti)commutator and then requiring $\chi_5(\ldots,\mathbf{1},\ldots) = 0$
for any position of unit matrix (for $\chi_4$ this follows automatically from properties of trace and commutator).

Basis Lorentz-invariant even-parity tensor products of $\Gamma$-matrices are
\begin{equation}
\begin{split}
&\mathbf{\theta}_{111} = \mathbf{1}\otimes\mathbf{1}\otimes\mathbf{1}\\
&\mathbf{\theta}_{122} = \mathbf{1}\otimes\gamma_\mu\otimes\gamma^\mu\\
&\mathbf{\theta}_{133} = \mathbf{1}\otimes\sigma_{\mu\nu}\otimes\sigma^{\mu\nu}\\
&\mathbf{\theta}_{144} = \mathbf{1}\otimes\gamma^5\gamma_\mu\otimes\gamma^5\gamma^\mu\\
&\mathbf{\theta}_{155} = \mathbf{1}\otimes\gamma^5\otimes\gamma^5\\
&\mathbf{\theta}_{523} = \gamma^5\otimes\gamma_\mu\otimes\gamma^5\gamma^\mu\\
&\mathbf{\theta}_{223} = \gamma_\mu\otimes\gamma_\nu\otimes\sigma^{\mu\nu}\\
&\mathbf{\theta}_{443} = \gamma^5\gamma_\mu\otimes\gamma^5\gamma_\nu\otimes\sigma^{\mu\nu}\\
&\mathbf{\theta}_{243} = \epsilon_{\kappa\lambda\mu\nu}\gamma^\kappa\otimes\gamma^5\gamma^\lambda\otimes\sigma^{\mu\nu}\\
&\mathbf{\theta}_{533} = \epsilon_{\kappa\lambda\mu\nu}\gamma^5\otimes\sigma^{\kappa\lambda}\otimes\sigma^{\mu\nu}\\
&\mathbf{\theta}_{333} = {\sigma_\kappa}^\lambda\otimes{\sigma_\lambda}^\mu\otimes{\sigma_\mu}^\kappa
\end{split}
\end{equation}

The interaction term generated by $\mathbf{\theta}_{ijk}$ and $\chi_l$ will be denoted as $I_{ijkl}$. 
Some of these will be identical, others will vanish due to symmetries. 
As the remaining ones are still quite numerous, they are listed by $\mathbf{\theta}$ below. 
The metod of computation is described in \ref{fierz6comp}.

For $\mathbf{\theta}_{111}$, the nontrivial identities are
\begin{equation}
\begin{split}
&\bar{\psi}\psi\bar{\psi}T^a\psi\bar{\psi}T^a\psi = \\
&=-\frac{(N + 4)\bar{\psi}\psi\bar{\psi}\psi\bar{\psi}\psi}{8N} 
- \frac{\bar{\psi}\psi\bar{\psi}\gamma_\mu\psi\bar{\psi}\gamma^\mu\psi}{2}
- \frac{\bar{\psi}\psi\bar{\psi}\sigma_{\mu\nu}\psi\bar{\psi}\sigma^{\mu\nu}\psi}{16} 
+ \frac{\bar{\psi}\psi\bar{\psi}\gamma^5\gamma_\mu\psi\bar{\psi}\gamma^5\gamma^\mu\psi}{8}
- \frac{\bar{\psi}\psi\bar{\psi}\gamma^5\psi\bar{\psi}\gamma^5\psi}{8}
\end{split}
\end{equation}

\begin{equation}
\begin{split}
&d^{abc}\bar{\psi}T^a\psi\bar{\psi}T^b\psi\bar{\psi}T^c\psi = 
\frac{N^2 + 12N + 32}{32N^2}\bar{\psi}\psi\bar{\psi}\psi\bar{\psi}\psi
+ \frac{3N + 12}{32N}\bar{\psi}\psi\bar{\psi}\gamma^5\psi\bar{\psi}\gamma^5\psi
+ \frac{9N + 12}{32N}\bar{\psi}\psi\bar{\psi}\gamma_\mu\psi\bar{\psi}\gamma^\mu\psi -\\
&-\frac{3N + 6}{3N}\bar{\psi}\psi\bar{\psi}\gamma^5\gamma_\mu\psi\bar{\psi}\gamma^5\gamma^\mu\psi
+\frac{13N + 12}{64N}\bar{\psi}\psi\bar{\psi}\sigma_{\mu\nu}\psi\bar{\psi}\sigma^{\mu\nu}\psi
-\frac{i}{64}\epsilon_{\kappa\lambda\mu\nu}\bar{\psi}\gamma^5\psi\bar{\psi}\sigma^{\kappa\lambda}\psi\bar{\psi}\sigma^{\mu\nu}\psi +\\
&+\frac{3}{32}\bar{\psi}\gamma^5\psi\bar{\psi}\gamma_\mu\psi\bar{\psi}\gamma^5\gamma^\mu\psi
-\frac{3}{64}\epsilon_{\kappa\lambda\mu\nu}\bar{\psi}\gamma^\kappa\psi\bar{\psi}\gamma^5\gamma^\lambda\psi\bar{\psi}\sigma^{\mu\nu}\psi
\end{split}
\end{equation}

$\theta_{122}$ yields
\begin{equation}
\begin{split}
&\bar{\psi}\psi\bar{\psi}T^a\gamma_\mu\psi\bar{\psi}T^a\gamma^\mu\psi = 
-\frac{1}{2}\bar{\psi}_{i_1}\psi_{i_1}\bar{\psi}_{i_2}\psi_{i_2}\bar{\psi}_{i_3}\psi_{i_3} +
\frac{1}{2}\bar{\psi}_{i_1}\psi_{i_1}\bar{\psi}_{i_2}\gamma^5\psi_{i_2}\bar{\psi}_{i_3}\gamma^5\psi_{i_3} +\\
 &+\frac{N - 2}{4N}\bar{\psi}_{i_1}\psi_{i_1}\bar{\psi}_{i_2}\gamma_\mu\psi_{i_2}\bar{\psi}_{i_3}\gamma^\mu\psi_{i_3}
 + \frac{1}{4}\bar{\psi}_{i_1}\psi_{i_1}\bar{\psi}_{i_2}\gamma^5\gamma_\mu\psi_{i_2}\bar{\psi}_{i_3}\gamma^5\gamma^\mu\psi_{i_3} 
\end{split}
\end{equation}

\begin{equation}
\begin{split}
&\bar{\psi}T^a\psi\bar{\psi}T^a\gamma_\mu\psi\bar{\psi}\gamma^\mu\psi 
= \bar{\psi}T^a\psi\bar{\psi}\gamma_\mu\psi\bar{\psi}T^a\gamma^\mu\psi =\\
&-\frac{N + 2}{4N}\bar{\psi}_{i_1}\psi_{i_1}\bar{\psi}_{i_2}\gamma_\mu\psi_{i_2}\bar{\psi}_{i_3}\gamma^\mu\psi_{i_3} 
+\frac{1}{8}\epsilon_{\kappa\lambda\mu\nu}\bar{\psi}_{i_1}\gamma^\kappa\psi_{i_1}\bar{\psi}_{i_2}\gamma^5\gamma^\lambda\psi_{i_2}\bar{\psi}_{i_3}\sigma^{\mu\nu}\psi_{i_3}
\end{split}
\end{equation}

\begin{equation}
\begin{split}
&d^{abc}\bar{\psi}T^a\psi\bar{\psi}T^b\gamma_\mu\psi\bar{\psi}T^c\gamma^\mu\psi = 
\frac{N + 4}{8N}\bar{\psi}\psi\bar{\psi}\psi\bar{\psi}\psi - \frac{N + 4}{8N}\bar{\psi}\psi\bar{\psi}\gamma^5\psi\bar{\psi}\gamma^5\psi
-\frac{3N^2 - 2N - 8}{8N^2}\bar{\psi}\psi\bar{\psi}\gamma_\mu\psi\bar{\psi}\gamma^\mu\psi -\\
&-\frac{3}{16}\bar{\psi}\gamma^5\psi\bar{\psi}\gamma_\mu\psi\bar{\psi}\gamma^5\gamma^\mu\psi
+\frac{i}{32}\epsilon_{\kappa\lambda\mu\nu}\bar{\psi}\gamma^5\psi\bar{\psi}\sigma^{\kappa\lambda}\psi\bar{\psi}\sigma^{\mu\nu}\psi
-\frac{7N + 4}{16N}\bar{\psi}\psi\bar{\psi}\gamma^5\gamma_\mu\psi\bar{\psi}\gamma^5\gamma^\mu\psi +\\
&+\frac{1}{8}\bar{\psi}\psi\bar{\psi}\sigma_{\mu\nu}\psi\bar{\psi}\sigma^{\mu\nu}\psi
-\frac{1}{4N}\epsilon_{\kappa\lambda\mu\nu}\bar{\psi}\gamma^\kappa\psi\bar{\psi}\gamma^5\gamma^\lambda\psi\bar{\psi}\sigma^{\mu\nu}\psi
\end{split}
\end{equation}

$\theta_{133}$ results in

\begin{equation}
\bar{\psi}\psi\bar{\psi}\sigma_{\mu\nu}T^a\psi\bar{\psi}\sigma^{\mu\nu}T^a\psi = 
-\frac{3}{2}\bar{\psi}\psi\bar{\psi}\psi\bar{\psi}\psi 
-\frac{3}{2}\bar{\psi}\psi\bar{\psi}\gamma^5\psi\bar{\psi}\gamma^5\psi
 + \frac{N-2}{4N}\bar{\psi}\psi\bar{\psi}\sigma_{\lambda\mu}\psi\bar{\psi}\sigma^{\lambda\mu}\psi
\end{equation}

\begin{equation}
\begin{split}
&\bar{\psi}T^a\psi\bar{\psi}\sigma_{\mu\nu}T^a\psi\bar{\psi}\sigma^{\mu\nu}\psi = 
\frac{i}{8}\epsilon_{\kappa\lambda\mu\nu}\bar{\psi}\gamma^5\psi\bar{\psi}\sigma^{\kappa\lambda}\psi\bar{\psi}\sigma^{\mu\nu}\psi - \\
&-\frac{N + 2}{4N}\bar{\psi}\psi\bar{\psi}\sigma_{\lambda\mu}\psi\bar{\psi}\sigma^{\lambda\mu}\psi
+\frac{1}{4}\epsilon_{\kappa\lambda\mu\nu}\bar{\psi}\gamma^{\kappa}\psi\bar{\psi}\gamma^5\gamma^{\lambda}\psi\bar{\psi}\sigma^{\mu\nu}\psi 
 - \frac{i}{4}\bar{\psi}{\sigma_\kappa}^\lambda\psi\bar{\psi}{\sigma_\lambda}^\mu\psi\bar{\psi}{\sigma_\mu}^\kappa\psi
\end{split}
\end{equation}

\begin{equation}
\begin{split}
&d^{abc}\bar{\psi}T^a\psi\bar{\psi}T^b\sigma_{\mu\nu}\psi\bar{\psi}T^c\sigma^{\mu\nu}\psi = 
\frac{3N + 12}{8N}\bar{\psi}\psi\bar{\psi}\psi\bar{\psi}\psi + \frac{9N + 12}{8N}\bar{\psi}\psi\bar{\psi}\gamma^5\psi\bar{\psi}\gamma^5\psi
+\frac{3}{8}\bar{\psi}\psi\bar{\psi}\gamma_\mu\psi\bar{\psi}\gamma^\mu\psi - \\
&- i\frac{N + 4}{16N}\epsilon_{\kappa\lambda\mu\nu}\bar{\psi}\gamma^5\psi\bar{\psi}\sigma^{\kappa\lambda}\psi\bar{\psi}\sigma^{\mu\nu}\psi
-\frac{9N^2 - 4N - 16}{16N^2}\bar{\psi}\psi\bar{\psi}\sigma_{\mu\nu}\psi\bar{\psi}\sigma^{\mu\nu}\psi +\\
&+\frac{N - 8}{16N}\epsilon_{\kappa\lambda\mu\nu}\bar{\psi}\gamma^\kappa\psi\bar{\psi}\gamma^5\gamma^\lambda\psi\bar{\psi}\sigma^{\mu\nu}\psi
-\frac{3}{16}\bar{\psi}\psi\bar{\psi}\gamma^5\gamma_\mu\psi\bar{\psi}\gamma^5\gamma^\mu\psi
+\frac{i}{2N}\bar{\psi}{\sigma_\kappa}^\lambda\psi\bar{\psi}{\sigma_\lambda}^\mu\psi\bar{\psi}{\sigma_\mu}^\kappa\psi.
\end{split}
\end{equation}

$\theta_{144}$ produces

\begin{equation}
\begin{split}
&\bar{\psi}\psi\bar{\psi}T^a\gamma^5\gamma_\mu\psi\bar{\psi}T^a\gamma^5\gamma^\mu\psi = 
\frac{1}{2}\bar{\psi}\psi\bar{\psi}\psi\bar{\psi}\psi
 -\frac{1}{2}\bar{\psi}\psi\bar{\psi}\gamma^5\psi\bar{\psi}\gamma^5\psi + \\
& +\frac{1}{4}\bar{\psi}\psi\bar{\psi}\gamma_\mu\psi\bar{\psi}\gamma^\mu\psi
+ \frac{N - 2}{4N}\bar{\psi}\psi\bar{\psi}\gamma^5\gamma_\mu\psi\bar{\psi}\gamma^5\gamma^\mu\psi
\end{split}
\end{equation}

\begin{equation}
\bar{\psi}T^a\psi\bar{\psi}T^a\gamma^5\gamma_\mu\psi\bar{\psi}\gamma^5\gamma^\mu\psi = 
-\frac{N + 2}{4N}\bar{\psi}\psi\bar{\psi}\gamma^5\gamma_\mu\psi\bar{\psi}\gamma^5\gamma^\mu\psi 
-\frac{1}{8}\epsilon_{\kappa\lambda\mu\nu}\bar{\psi}\gamma^\kappa\psi\bar{\psi}\gamma^5\gamma^\lambda\psi\bar{\psi}\sigma^{\mu\nu}\psi
\end{equation}

\begin{equation}
\begin{split}
&d^{abc}\bar{\psi}T^a\psi\bar{\psi}T^b\gamma^5\gamma_\mu\psi\bar{\psi}T^c\gamma^5\gamma^\mu\psi = 
-\frac{N+4}{8N}\bar{\psi}\psi\bar{\psi}\psi\bar{\psi}\psi + \frac{1}{8}\bar{\psi}\psi\bar{\psi}\gamma^5\psi\bar{\psi}\gamma^5\psi
-\frac{3}{16}\bar{\psi}\gamma^5\psi\bar{\psi}\gamma_\mu\psi\bar{\psi}\gamma^5\gamma^\mu\psi -\\
&-\frac{5N + 2}{8N}\bar{\psi}\psi\bar{\psi}\gamma_\mu\psi\bar{\psi}\gamma^\mu\psi
-\frac{i}{32}\epsilon_{\kappa\lambda\mu\nu}\bar{\psi}\gamma^5\psi\bar{\psi}\sigma^{\kappa\lambda}\psi\bar{\psi}\sigma^{\mu\nu}\psi
-\frac{3N^2 - 4N - 16}{16N^2}\bar{\psi}\psi\bar{\psi}\gamma^5\gamma_\mu\psi\bar{\psi}\gamma^5\gamma^\mu\psi -\\
&-\frac{1}{16}\bar{\psi}\psi\bar{\psi}\sigma_{\mu\nu}\psi\bar{\psi}\sigma^{\mu\nu}\psi
+\frac{1}{2N}\bar{\psi}\psi\bar{\psi}\gamma^5\psi\bar{\psi}\gamma^5\psi
+\frac{1}{4N}\epsilon_{\kappa\lambda\mu\nu}\bar{\psi}\gamma^\kappa\psi\bar{\psi}\gamma^5\gamma^\lambda\psi\bar{\psi}\sigma^{\mu\nu}\psi
\end{split}
\end{equation}

$\theta_{155}$ gives

\begin{equation}
\begin{split}
&\bar{\psi}\psi\bar{\psi}T^a\gamma^5\psi\bar{\psi}T^a\gamma^5\psi = 
-\frac{1}{8}\bar{\psi}\psi\bar{\psi}\psi\bar{\psi}\psi
-\frac{N+4}{8N}\bar{\psi}\psi_{i_1}\bar{\psi}\gamma^5\psi\bar{\psi}\gamma^5\psi -\\
&-\frac{1}{8}\bar{\psi}\psi\bar{\psi}\gamma^5\gamma_\mu\psi\bar{\psi}\gamma^5\gamma^\mu\psi
+\frac{1}{8}\bar{\psi}\psi\bar{\psi}\gamma_\mu\psi\bar{\psi}\gamma^\mu\psi 
-\frac{1}{16}\bar{\psi}\psi\bar{\psi}\sigma_{\mu\nu}\psi\bar{\psi}\sigma^{\mu\nu}\psi
\end{split}
\end{equation}

\begin{equation}
\bar{\psi}T^a\psi\bar{\psi}T^a\gamma^5\psi\bar{\psi}\gamma^5\psi = 
-\frac{N + 2}{4N}\bar{\psi}\psi\bar{\psi}\gamma^5\psi\bar{\psi}\gamma^5\psi
+\frac{i}{32}\epsilon_{\kappa\lambda\mu\nu}\bar{\psi}\gamma^5\psi\bar{\psi}\sigma^{\kappa\lambda}\psi\bar{\psi}\sigma^{\mu\nu}\psi
\end{equation}

\begin{equation}
\begin{split}
&d^{abc}\bar{\psi}T^a\psi\bar{\psi}T^b\gamma^5\psi\bar{\psi}T^c\gamma^5\psi = 
\frac{N+4}{32N}\bar{\psi}\psi\bar{\psi}\psi\bar{\psi}\psi
+\frac{3N^2 + 20N + 32}{32N^2}\bar{\psi}\psi\bar{\psi}\gamma^5\psi\bar{\psi}\gamma^5\psi -\\
&-i\frac{N + 4}{64N}\epsilon_{\kappa\lambda\mu\nu}\bar{\psi}\gamma^5\psi\bar{\psi}\sigma^{\kappa\lambda}\psi\bar{\psi}\sigma^{\mu\nu}\psi
-\frac{7N + 4}{32N}\bar{\psi}\psi\bar{\psi}\gamma_\mu\psi\bar{\psi}\gamma^\mu\psi
-\frac{3}{32}\bar{\psi}\gamma^5\psi\bar{\psi}\gamma_\mu\psi\bar{\psi}\gamma^5\gamma^\mu\psi +\\
&+\frac{N+1}{8N}\bar{\psi}\psi\bar{\psi}\gamma^5\gamma_\mu\psi\bar{\psi}\gamma^5\gamma^\mu\psi
+\frac{1}{64}\epsilon_{\kappa\lambda\mu\nu}\bar{\psi}\gamma^\kappa\psi\bar{\psi}\gamma^5\gamma^\lambda\psi\bar{\psi}\sigma^{\mu\nu}\psi
+\frac{13N + 4}{64N}\bar{\psi}\psi\bar{\psi}\sigma_{\mu\nu}\psi\bar{\psi}\sigma^{\mu\nu}\psi
\end{split}
\end{equation}

It is easy to see that, because of symmetries,
\begin{equation}
f^{abc}\bar{\psi}T^a\psi\bar{\psi}T^b\Gamma_{i\bar{\alpha}}\psi\bar{\psi}T^c\Gamma_i^{\bar{\alpha}}\psi = 0,
\end{equation}
where summation over $i$ is not implied.

$\theta_{523}$ results in

\begin{equation}
\bar{\psi}\gamma^5\psi\bar{\psi}T^a\gamma_\mu\psi\bar{\psi}T^a\gamma^5\gamma^\mu\psi = 
\frac{N - 1}{2N}\bar{\psi}\gamma^5\psi\bar{\psi}\gamma_\mu\psi\bar{\psi}\gamma^5\gamma^\mu\psi
\end{equation}

\begin{equation}
\bar{\psi}T^a\gamma^5\psi\bar{\psi}\gamma_\mu\psi\bar{\psi}T^a\gamma^5\gamma^\mu\psi =
\bar{\psi}T^a\gamma^5\psi\bar{\psi}T^a\gamma_\mu\psi\bar{\psi}\gamma^5\gamma^\mu\psi = 
-\frac{N + 2}{4N}\bar{\psi}\gamma^5\psi\bar{\psi}\gamma_\mu\psi\bar{\psi}\gamma^5\gamma^\mu\psi
\end{equation}

\begin{equation}
\begin{split}
&f^{abc}\bar{\psi}T^a\gamma^5\psi\bar{\psi}T^b\gamma_\mu\psi\bar{\psi}T^c\gamma^5\gamma^\mu\psi = 
\frac{i}{8}\bar{\psi}\psi\bar{\psi}\psi\bar{\psi}\psi 
+ \frac{3i}{8}\bar{\psi}\gamma^5\psi\bar{\psi}\gamma_\mu\psi\bar{\psi}\gamma^5\gamma^\mu\psi
- \frac{i}{8}\bar{\psi}\psi\bar{\psi}\gamma^5\psi\bar{\psi}\gamma^5\psi +\\
&+ \frac{7i}{64}\bar{\psi}\psi\bar{\psi}\gamma_\mu\psi\bar{\psi}\gamma^\mu\psi
+ \frac{i}{64}\epsilon_{\kappa\lambda\mu\nu}\bar{\psi}\gamma^\kappa\psi\bar{\psi}\gamma^5\gamma^\lambda\psi\bar{\psi}\sigma^{\mu\nu}\psi
+ \frac{17i}{64}\bar{\psi}\psi\bar{\psi}\gamma^5\gamma_\mu\psi\bar{\psi}\gamma^5\gamma^\mu\psi +\\
&+ \frac{i}{16}\bar{\psi}\psi\bar{\psi}\sigma_{\mu\nu}\psi\bar{\psi}\sigma^{\mu\nu}\psi
- \frac{1}{32}\epsilon_{\kappa\lambda\mu\nu}\bar{\psi}\gamma^5\psi\bar{\psi}\sigma^{\kappa\lambda}\psi\bar{\psi}\sigma^{\mu\nu}\psi
\end{split}
\end{equation}

\begin{equation}
\begin{split}
&d^{abc}\bar{\psi}T^a\gamma^5\psi\bar{\psi}T^b\gamma_\mu\psi\bar{\psi}T^c\gamma^5\gamma^\mu\psi = 
\frac{4 - N^2}{4N^2}\bar{\psi}\gamma^5\psi\bar{\psi}\gamma_\mu\psi\bar{\psi}\gamma^5\gamma^\mu\psi +\\
&+\frac{3}{64}\bar{\psi}\psi\bar{\psi}\gamma_\mu\psi\bar{\psi}\gamma^\mu\psi
-\frac{3}{64}\bar{\psi}\psi\bar{\psi}\gamma^5\gamma_\mu\psi\bar{\psi}\gamma^5\gamma^\mu\psi
+\frac{1}{64}\epsilon_{\kappa\lambda\mu\nu}\bar{\psi}\gamma^\kappa\bar{\psi}\gamma^5\gamma^\lambda\psi\bar{\psi}\sigma^{\mu\nu}\psi.
\end{split}
\end{equation}

$\theta_{223}$ yields

\begin{equation}
\bar{\psi}T^a\gamma_\mu\psi\bar{\psi}T^a\gamma_\nu\psi\bar{\psi}\sigma^{\mu\nu}\psi = 0
\end{equation}

\begin{equation}
\bar{\psi}\gamma_\mu\psi\bar{\psi}T^a\gamma_\nu\psi\bar{\psi}T^a\sigma^{\mu\nu}\psi = 
- \bar{\psi}T^a\gamma_\mu\psi\bar{\psi}\gamma_\nu\psi\bar{\psi}T^a\sigma^{\mu\nu}\psi = 
-\frac{3i}{4}\bar{\psi}\gamma^5\psi\bar{\psi}\gamma_\mu\psi\bar{\psi}\gamma^5\gamma^\mu\psi
\end{equation}

\begin{equation}
d^{abc}\bar{\psi}T^a\gamma_\mu\psi\bar{\psi}T^b\gamma_\nu\psi\bar{\psi}T^c\sigma^{\mu\nu}\psi = 0
\end{equation}

\begin{equation}
\begin{split}
&f^{abc}\bar{\psi}T^a\gamma_\mu\psi\bar{\psi}T^b\gamma_\nu\psi\bar{\psi}T^c\sigma^{\mu\nu}\psi = 
-\frac{3}{8}\bar{\psi}\psi\bar{\psi}\psi\bar{\psi}\psi
+ \frac{15}{16}\bar{\psi}\gamma^5\psi\bar{\psi}\gamma_\mu\psi\bar{\psi}\gamma^5\gamma^\mu\psi
+ \frac{3}{8}\bar{\psi}\psi\bar{\psi}\gamma^5\psi\bar{\psi}\gamma^5\psi - \\
&- \frac{9}{32}\bar{\psi}\psi\bar{\psi}\gamma_\mu\psi\bar{\psi}\gamma^\mu\psi
- \frac{9}{32}\bar{\psi}\psi\bar{\psi}\gamma^5\gamma_\mu\psi\bar{\psi}\gamma^5\gamma^\mu\psi
+ \frac{1}{16}\bar{\psi}\psi\bar{\psi}\sigma_{\mu\nu}\psi\bar{\psi}\sigma^{\mu\nu}\psi - \\
&- \frac{1}{16}\epsilon_{\kappa\lambda\mu\nu}\bar{\psi}\gamma^\kappa\psi\bar{\psi}\gamma^5\gamma^\lambda\psi\bar{\psi}\sigma^{\mu\nu}\psi
+ \frac{i}{32}\epsilon_{\kappa\lambda\mu\nu}\bar{\psi}\gamma^5\psi\bar{\psi}\sigma^{\kappa\lambda}\psi\bar{\psi}\sigma^{\mu\nu}\psi
\end{split}
\end{equation}

$\theta_{443}$ gives
\begin{equation}
\bar{\psi}T^a\gamma^5\gamma_\mu\psi\bar{\psi}T^a\gamma^5\gamma_\nu\psi\bar{\psi}\sigma^{\mu\nu}\psi = 0
\end{equation}

\begin{equation}
\bar{\psi}\gamma^5\gamma_\mu\psi\bar{\psi}T^a\gamma^5\gamma_\nu\psi\bar{\psi}T^a\sigma^{\mu\nu}\psi = 
- \bar{\psi}T^a\gamma^5\gamma_\mu\psi\bar{\psi}\gamma^5\gamma_\nu\psi\bar{\psi}T^a\sigma^{\mu\nu}\psi = 
-\frac{3i}{4}\bar{\psi}\gamma^5\psi\bar{\psi}\gamma_\mu\psi\bar{\psi}\gamma^5\gamma^\mu\psi
\end{equation}

\begin{equation}
d^{abc}\bar{\psi}T^a\gamma^5\gamma_\mu\psi\bar{\psi}T^b\gamma^5\gamma_\nu\psi\bar{\psi}T^c\sigma^{\mu\nu}\psi = 0
\end{equation}

\begin{equation}
\begin{split}
&f^{abc}\bar{\psi}T^a\gamma_\mu\psi\bar{\psi}T^b\gamma_\nu\psi\bar{\psi}T^c\sigma^{\mu\nu}\psi = 
\frac{3}{8}\bar{\psi}\psi\bar{\psi}\psi\bar{\psi}\psi
- \frac{3}{8}\bar{\psi}\psi\bar{\psi}\gamma^5\psi\bar{\psi}\gamma^5\psi
-\frac{i}{32}\epsilon_{\kappa\lambda\mu\nu}\bar{\psi}\gamma^5\psi\bar{\psi}\sigma^{\kappa\lambda}\psi\bar{\psi}\sigma^{\mu\nu}\psi +\\
&+\frac{3}{32}\bar{\psi}\psi\bar{\psi}\gamma_\mu\psi\bar{\psi}\gamma^\mu\psi
+\frac{3}{32}\bar{\psi}\psi\bar{\psi}\gamma^5\gamma_\mu\psi\bar{\psi}\gamma^5\gamma^\mu\psi
+\frac{1}{32}\epsilon_{\kappa\lambda\mu\nu}\bar{\psi}\gamma^\kappa\psi\bar{\psi}\gamma^5\gamma^\lambda\psi\bar{\psi}\sigma^{\mu\nu}\psi +\\
&+\frac{9}{16}\bar{\psi}\gamma^5\psi\bar{\psi}\gamma_\mu\psi\bar{\psi}\gamma^5\gamma^\mu\psi
-\frac{1}{16}\bar{\psi}\psi\bar{\psi}\sigma_{\mu\nu}\psi\bar{\psi}\sigma^{\mu\nu}\psi
\end{split}
\end{equation}

$\theta_{243}$ generates

\begin{equation}
\epsilon_{\kappa\lambda\mu\nu}\bar{\psi}\gamma^\kappa\psi\bar{\psi}T^a\gamma^5\gamma^\lambda\psi\bar{\psi}T^a\sigma^{\mu\nu}\psi = 
\frac{3}{2}\bar{\psi}\psi\bar{\psi}\gamma_\mu\psi\bar{\psi}\gamma^\mu\psi 
- \frac{N + 2}{4N}\epsilon_{\kappa\lambda\mu\nu}\bar{\psi}\gamma^\kappa\psi\bar{\psi}\gamma^5\gamma^\lambda\psi\bar{\psi}\sigma^{\mu\nu}\psi
\end{equation}

\begin{equation}
\begin{split}
&\epsilon_{\kappa\lambda\mu\nu}\bar{\psi}T^a\gamma^\kappa\psi\bar{\psi}T^a\gamma^5\gamma^\lambda\psi\bar{\psi}\sigma^{\mu\nu}\psi = 
\frac{i}{4}\epsilon_{\kappa\lambda\mu\nu}\bar{\psi}\gamma^5\psi\bar{\psi}\sigma^{\kappa\lambda}\psi\bar{\psi}\sigma^{\mu\nu}\psi +\\
&+\frac{1}{2}\bar{\psi}\psi\bar{\psi}\sigma_{\mu\nu}\psi\bar{\psi}\sigma^{\mu\nu}\psi 
- \frac{1}{2N}\epsilon_{\kappa\lambda\mu\nu}\bar{\psi}\gamma^\kappa\psi\bar{\psi}\gamma^5\gamma^\lambda\psi\bar{\psi}\sigma^{\mu\nu}\psi
\end{split}
\end{equation}

\begin{equation}
\epsilon_{\kappa\lambda\mu\nu}\bar{\psi}T^a\gamma^\kappa\psi\bar{\psi}\gamma^5\gamma^\lambda\psi\bar{\psi}T^a\sigma^{\mu\nu}\psi = 
- \frac{3}{2}\bar{\psi}\psi\bar{\psi}\gamma^5\gamma_\mu\psi\bar{\psi}\gamma^5\gamma^\mu\psi 
+ \frac{N + 2}{4N}\epsilon_{\kappa\lambda\mu\nu}\bar{\psi}\gamma^\kappa\psi\bar{\psi}\gamma^5\gamma^\lambda\psi\bar{\psi}\sigma^{\mu\nu}\psi
\end{equation}

\begin{equation}
\begin{split}
&f^{abc}\epsilon_{\kappa\lambda\mu\nu}\bar{\psi}T^a\gamma^\kappa\psi\bar{\psi}T^b\gamma^5\gamma^\lambda\psi\bar{\psi}T^c\sigma^{\mu\nu}\psi = 
-\frac{3i}{8}\bar{\psi}\gamma^5\psi\bar{\psi}\gamma_\mu\psi\bar{\psi}\gamma^5\gamma^\mu\psi
+\frac{3i}{4}\bar{\psi}\psi\bar{\psi}\gamma_\mu\psi\bar{\psi}\gamma^\mu\psi +\\
&+\frac{3i}{4}\bar{\psi}\psi\bar{\psi}\gamma^5\gamma_\mu\psi\bar{\psi}\gamma^5\gamma^\mu\psi
-\frac{i}{32}\epsilon_{\kappa\lambda\mu\nu}\bar{\psi}\gamma^\kappa\psi\bar{\psi}\gamma^5\gamma^\lambda\psi\bar{\psi}\sigma^{\mu\nu}\psi
-\frac{1}{16}\bar{\psi}{\sigma_\kappa}^\lambda\psi\bar{\psi}{\sigma_\lambda}^\mu\psi\bar{\psi}{\sigma_\mu}^\kappa\psi
\end{split}
\end{equation}

\begin{equation}
\begin{split}
&d^{abc}\epsilon_{\kappa\lambda\mu\nu}\bar{\psi}T^a\gamma^\kappa\psi\bar{\psi}T^b\gamma^5\gamma^\lambda\psi\bar{\psi}T^c\sigma^{\mu\nu}\psi = 
-\frac{3}{4}\bar{\psi}\psi\bar{\psi}\psi\bar{\psi}\psi
+\frac{3}{4}\bar{\psi}\psi\bar{\psi}\gamma^5\psi\bar{\psi}\gamma^5\psi 
+\frac{3}{4}\bar{\psi}\gamma^5\psi\bar{\psi}\gamma_\mu\psi\bar{\psi}\gamma^5\gamma^\mu\psi -\\
&-\frac{3N + 24}{16N}\bar{\psi}\psi\bar{\psi}\gamma_\mu\psi\bar{\psi}\gamma^\mu\psi
+\frac{N - 4}{8N}\bar{\psi}\psi\bar{\psi}\sigma_{\mu\nu}\psi\bar{\psi}\sigma^{\mu\nu}\psi
-\frac{3N - 24}{16N}\bar{\psi}\psi\bar{\psi}\gamma^5\gamma_\mu\psi\bar{\psi}\gamma^5\gamma^\mu\psi +\\
&+\frac{3N^2 + 32}{32N^2}\epsilon_{\kappa\lambda\mu\nu}\bar{\psi}\gamma^\kappa\psi\bar{\psi}\gamma^5\gamma^\lambda\psi\bar{\psi}\sigma^{\mu\nu}\psi 
+ i\frac{N - 4}{16N}\epsilon_{\kappa\lambda\mu\nu}\bar{\psi}\gamma^5\psi\bar{\psi}\sigma^{\kappa\lambda}\psi\bar{\psi}\sigma^{\mu\nu}\psi
+\frac{i}{16}\bar{\psi}{\sigma_\kappa}^\lambda\psi\bar{\psi}{\sigma_\lambda}^\mu\psi\bar{\psi}{\sigma_\mu}^\kappa\psi
\end{split}
\end{equation}

$\theta_{533}$ yields

\begin{equation}
\epsilon_{\kappa\lambda\mu\nu}\bar{\psi}\gamma^5\psi\bar{\psi}T^a\sigma^{\kappa\lambda}\psi\bar{\psi}T^a\sigma^{\mu\nu}\psi = 
-6i\bar{\psi}\psi\bar{\psi}\gamma^5\psi\bar{\psi}\gamma^5\psi
+\frac{N - 2}{4N}\epsilon_{\kappa\lambda\mu\nu}\bar{\psi}\gamma^5\psi\bar{\psi}\sigma^{\kappa\lambda}\psi\bar{\psi}\sigma^{\mu\nu}\psi
\end{equation}

\begin{equation}
\begin{split}
&\epsilon_{\kappa\lambda\mu\nu}\bar{\psi}T^a\gamma^5\psi\bar{\psi}T^a\sigma^{\kappa\lambda}\psi\bar{\psi}\sigma^{\mu\nu}\psi =
-\frac{N+2}{4N}\epsilon_{\kappa\lambda\mu\nu}\bar{\psi}_{i_1}\gamma^5\psi_{i_1}\bar{\psi}_{i_2}\sigma^{\kappa\lambda}\psi_{i_2}\bar{\psi}_{i_3}\sigma^{\mu\nu}\psi_{i_3} + \\
&-\frac{i}{2}\bar{\psi}_{i_1}\psi_{i_1}\bar{\psi}_{i_2}\sigma_{\mu\nu}\psi_{i_2}\bar{\psi}_{i_3}\sigma^{\mu\nu}\psi_{i_3} 
 - \frac{i}{2}\epsilon_{\kappa\lambda\mu\nu}\bar{\psi}_{i_1}\gamma^\kappa\psi_{i_1}\bar{\psi}_{i_2}\gamma^5\gamma^\lambda\psi_{i_2}\bar{\psi}_{i_3}\sigma^{\mu\nu}\psi_{i_3} 
 + \frac{1}{2}\bar{\psi}_{i_1}{\sigma_\kappa}^\lambda\psi_{i_1}\bar{\psi}_{i_2}{\sigma_\lambda}^\mu\psi_{i_2}\bar{\psi}_{i_3}{\sigma_\mu}^\kappa\psi_{i_3}
\end{split}
\end{equation}

\begin{equation}
f^{abc}\epsilon_{\kappa\lambda\mu\nu}\bar{\psi}T^a\gamma^5\psi\bar{\psi}T^b\sigma^{\kappa\lambda}\psi\bar{\psi}T^c\sigma^{\mu\nu}\psi = 0
\end{equation}

\begin{equation}
\begin{split}
&d^{abc}\epsilon_{\kappa\lambda\mu\nu}\bar{\psi}T^a\gamma^5\psi\bar{\psi}T^b\sigma^{\kappa\lambda}\psi\bar{\psi}T^c\sigma^{\mu\nu}\psi = 
i\frac{9N + 12}{2N}\bar{\psi}\psi\bar{\psi}\gamma^5\psi\bar{\psi}\gamma^5\psi 
+ \frac{3i}{4}\bar{\psi}\psi\bar{\psi}\gamma^5\gamma_\mu\psi\bar{\psi}\gamma^5\gamma^\mu\psi
+\frac{3i}{4}\bar{\psi}\psi\bar{\psi}\psi\bar{\psi}\psi +\\
&+i\frac{N + 2}{2N}\bar{\psi}\psi\bar{\psi}\sigma_{\mu\nu}\psi\bar{\psi}\sigma^{\mu\nu}\psi
-\frac{N^2 - N - 4}{4N^2}\epsilon_{\kappa\lambda\mu\nu}\bar{\psi}\gamma^5\psi\bar{\psi}\sigma^{\kappa\lambda}\psi\bar{\psi}\sigma^{\mu\nu}\psi
-\frac{3i}{4}\bar{\psi}\psi\bar{\psi}\gamma_\mu\psi\bar{\psi}\gamma^\mu\psi +\\
&+\frac{i}{N}\epsilon_{\kappa\lambda\mu\nu}\bar{\psi}\gamma^\kappa\psi\bar{\psi}\gamma^5\gamma^\lambda\psi\bar{\psi}\sigma^{\mu\nu}\psi
-\frac{1}{N}\bar{\psi}{\sigma_\kappa}^\lambda\psi\bar{\psi}{\sigma_\lambda}^\mu\psi\bar{\psi}{\sigma_\mu}^\kappa\psi
\end{split}
\end{equation}

Finally, $\theta_{333}$ gives

\begin{equation}
\begin{split}
&\bar{\psi}T^a{\sigma_\kappa}^\lambda\psi\bar{\psi}T^a{\sigma_\lambda}^\mu\psi\bar{\psi}{\sigma_\mu}^\kappa\psi =
\bar{\psi}{\sigma_\kappa}^\lambda\psi\bar{\psi}T^a{\sigma_\lambda}^\mu\psi\bar{\psi}T^a{\sigma_\mu}^\kappa\psi = \\
&=\bar{\psi}T^a{\sigma_\kappa}^\lambda\psi\bar{\psi}{\sigma_\lambda}^\mu\psi\bar{\psi}T^a{\sigma_\mu}^\kappa\psi =
-\frac{1}{2N}\bar{\psi}{\sigma_\kappa}^\lambda\psi\bar{\psi}{\sigma_\lambda}^\mu\psi\bar{\psi}{\sigma_\mu}^\kappa\psi
\end{split}
\end{equation}

\begin{equation}
f^{abc}\bar{\psi}T^a{\sigma_\kappa}^\lambda\psi\bar{\psi}T^b{\sigma_\lambda}^\mu\psi\bar{\psi}T^c{\sigma_\mu}^\kappa\psi = 
\frac{1}{16}\epsilon_{\kappa\lambda\mu\nu}\bar{\psi}\gamma^5\psi\bar{\psi}\sigma^{\kappa\lambda}\psi\bar{\psi}\sigma^{\mu\nu}\psi
-\frac{9}{4}\bar{\psi}\psi\bar{\psi}\gamma^5\psi\bar{\psi}\gamma^5\psi 
-\frac{3}{4}\bar{\psi}\psi\bar{\psi}\psi\bar{\psi}\psi
\end{equation}

\begin{equation}
d^{abc}\bar{\psi}T^a{\sigma_\kappa}^\lambda\psi\bar{\psi}T^b{\sigma_\lambda}^\mu\psi\bar{\psi}T^c{\sigma_\mu}^\kappa\psi = 
\frac{i}{8}\bar{\psi}\psi\bar{\psi}\sigma_{\mu\nu}\psi\bar{\psi}\sigma^{\mu\nu}\psi 
+ \frac{1}{N^2}\bar{\psi}{\sigma_\kappa}^\lambda\psi\bar{\psi}{\sigma_\lambda}^\mu\psi\bar{\psi}{\sigma_\mu}^\kappa\psi
\end{equation}

\subsection{NJL model with three-fermion fields} \label{trifermion}

Completeness of the basis \eqref{intro/gammabasis} implies that third-order monomials 
\begin{equation}
\mu_{21\,\bar{i}}^{\bar{\alpha}} = \psi\bar{\psi}\tau_{\bar{i}}^{\bar{\alpha}}\psi
\end{equation}
form a complete set of 2-1 products of $\psi$ and $\bar{\psi}$ (the numbers 2 and 1 here denote the powers of $\psi$ and $\bar{\psi}$ respectively), and their conjugates
\begin{equation}
\mu_{12\,\bar{i}}^{\bar{\alpha}} = \bar{\psi}\tau_{\bar{i}}^{\bar{\alpha}}\psi\bar{\psi}
\end{equation}
define a complete set of 1-2 products.

It is possible to write, along the lines of \ref{oneloopterms},

\begin{equation}
\begin{split}
Z\left[\left\{J\right\}\right] &= 
\funint{\bar{\psi},\psi}\exp\left\{i\int\left[\bar{\psi}i\widehat{\partial}\psi + 
G_i\bar{\psi}\Gamma_{i\,\bar{\alpha}}\psi\bar{\psi}\Gamma^{\bar{\alpha}}_i\psi +\right.\right.\\
&\left.\left.+ \bar{J}_1\psi +\bar{\psi}J_1 
+ 2\lambda_{11\bar{i}}J_{11\bar{i}\bar{\alpha}}\bar{\psi}\tau_{\bar{i}}^{\bar{\alpha}}\psi
+ \lambda_{11\bar{i}}J_{11\bar{i}\bar{\alpha}}J^{\bar{\alpha}}_{2\bar{i}} +\right.\right.\\
&\left.\left.
+ \left(\varkappa_{21\bar{i}}\bar{J}_{21\bar{i}\bar{\alpha}}\tau_{\bar{i}}^{\bar{\alpha}}\psi + \lambda_{21\bar{i}}\psi\bar{\psi}\tau_{\bar{i}}^{\bar{\alpha}}\psi\right)
+\left(\varkappa_{21\bar{i}}\bar{\psi}\tau_{\bar{i}}^{\bar{\alpha}} + \lambda_{21\bar{i}}\bar{\psi}\tau_{\bar{i}}^{\bar{\alpha}}\psi\bar{\psi}\right)J_{21\bar{i}\bar{\alpha}}
+\varkappa_{21\bar{i}}\bar{J}_{21\bar{i}\bar{\alpha}}J_{21\bar{i}}^{\bar{\alpha}}
\right]d^4x\right\}.
\end{split} \label{trifermion/genfunc}
\end{equation}

After inserting

\begin{equation}
\begin{split}
\text{const} &= \funint{\Sigma}\exp\left\{
-i\int\left[\frac{\lambda_{11\bar{i}}^{-\frac{1}{2}}\Sigma_{11\bar{i}\bar{\alpha}}}{2} 
- \lambda_{11\bar{i}}^\frac{1}{2}\bar{\psi}\tau_{\bar{i}\bar{\alpha}}\psi
 - \lambda_{11\bar{i}}^\frac{1}{2}J_{11\bar{i}\bar{\alpha}}\right]
 \left[\frac{\lambda_{11\bar{i}}^{-\frac{1}{2}}\Sigma_{11\bar{i}}^{\bar{\alpha}}}{2} 
- \lambda_{11\bar{i}}^\frac{1}{2}\bar{\psi}\tau_{\bar{i}}^{\bar{\alpha}}\psi
 - \lambda_{11\bar{i}}^\frac{1}{2}J_{11\bar{i}}^{\bar{\alpha}}\right]d^4x -\right.\\
 &\left.-i\int\left[
\varkappa_{21\bar{i}}^{-\frac{1}{2}}\bar{\Sigma}_{21\bar{i}\bar{\alpha}} 
- \left(\varkappa_{21\bar{i}}^\frac{1}{2}\bar{\psi}\tau_{\bar{i}\bar{\alpha}} + \lambda_{21\bar{i}}^\frac{1}{2}\bar{\psi}\tau_{\bar{i}\bar{\alpha}}\psi\bar{\psi}\right)
 - \varkappa_{21\bar{i}}^\frac{1}{2}\bar{J}_{21\bar{i}\bar{\alpha}}
 \right]\times\right.\\
 &\left.\times\left[
\varkappa_{21\bar{i}}^{-\frac{1}{2}}\Sigma_{21\bar{i}}^{\bar{\alpha}} 
- \left(\varkappa_{21\bar{i}}^\frac{1}{2}\tau_{\bar{i}}^{\bar{\alpha}}\psi + \lambda_{21\bar{i}}^\frac{1}{2}\psi\bar{\psi}\tau_{\bar{i}}^{\bar{\alpha}}\psi\right)
 - \varkappa_{21\bar{i}}^\frac{1}{2}J_{21\bar{i}}^{\bar{\alpha}}
 \right]d^4x\right\}
\end{split}
\label{trifermion/transform}
\end{equation}

the generating functional becomes

\begin{equation}
\begin{split}
\tilde{Z}\left[\left\{J\right\}\right] &= 
\funint{\bar{\psi},\psi,\Sigma}\exp\left\{i\int\left[\bar{\psi}\left(i\widehat{\partial} + \widehat{\Sigma}_{11} 
- \lambda_{21\bar{i}}\tau_{\bar{i}\bar{\alpha}}\tau_{\bar{i}}^{\bar{\alpha}}\right)\psi
+ \bar{J}_1\psi +\bar{\psi}J_1
- \frac{\Sigma_{11\bar{i}\bar{\alpha}}\Sigma_{11\bar{i}}^{\bar{\alpha}}}{4\lambda_{11\bar{i}}}
+ J_{11\bar{i}\bar{\alpha}}\Sigma_{11\bar{i}}^{\bar{\alpha}} - \right.\right.\\
&\left.\left.
-\varkappa_{21\bar{i}}^{-1}\bar{\Sigma}_{21\bar{i}\bar{\alpha}}\Sigma_{21\bar{i}}^{\bar{\alpha}}
+ \bar{\Sigma}_{21\bar{i}\bar{\alpha}}\left(\tau_{\bar{i}}^{\bar{\alpha}}\psi + g_{21\bar{i}}\psi\bar{\psi}\tau_{\bar{i}}^{\bar{\alpha}}\psi\right)
+ \left(\bar{\psi}\tau_{\bar{i}}^{\bar{\alpha}} + g_{21\bar{i}}\bar{\psi}\tau_{\bar{i}}^{\bar{\alpha}}\psi\bar{\psi}\right)\Sigma_{21\bar{i}\bar{\alpha}} +\right.\right.\\
&\left.\left.+ \bar{\Sigma}_{21\bar{i}\bar{\alpha}}J_{21\bar{i}}^{\bar{\alpha}}
+ \bar{J}_{21\bar{i}\bar{\alpha}}\Sigma_{21\bar{i}}^{\bar{\alpha}}
\right]d^4x\right\}\\
g_{21\bar{i}} &= \sqrt{\frac{\lambda_{21\bar{i}}}{\varkappa_{21\bar{i}}}}
\end{split} \label{trifermion/compfields}
\end{equation}

The $\varkappa$ and $\lambda$ coefficients are assumed to be positive real numbers here.
The composite fields $\Sigma_{21\bar{i}}^{\bar{\alpha}}$ have various half-integer spins. 
These fields can couple to both linear and cubic combinations of $\psi$ and $\bar{\psi}$. 
Proper scaling of the fields and their sources reduces the number of related coupling constants 
to two families $\varkappa_{21\bar{i}}$ and $\lambda_{21\bar{i}}$.

Fourth- and sixth-order interaction terms must cancel independently. This imposes two sets of relations on the constants.
However, only a small subset of sixth-order Fierz identities is involved: the ones for the terms of type
\begin{equation}
\Lgr_{6\bar{i}} = \bar{\psi}\psi\bar{\psi}\tau_{\bar{i}\bar{\alpha}}\psi\bar{\psi}\tau_{\bar{i}}^{\bar{\alpha}}\psi
\end{equation}
(summation over $\bar{i}$ not implied) which are essentially the same as fourth-order Fierz identities.
Thus, the complete set of cancellation conditions is

\begin{equation}
\begin{split}
	&\begin{bmatrix}
		\lambda_{21\,10}\\\lambda_{21\,20}\\\lambda_{21\,30}\\\lambda_{21\,40}\\\lambda_{21\,50}
	\end{bmatrix} + 
	\begin{bmatrix}
		-\frac{N + 4}{8N}&-\frac{1}{2}&-\frac{3}{2}&\frac{1}{2}&-\frac{1}{8}\\
		-\frac{1}{2}&\frac{N - 2}{4N}&0&\frac{1}{4}&\frac{1}{8}\\
		-\frac{1}{16}&0&\frac{N - 2}{4N}&0&-\frac{1}{16}\\
		\frac{1}{8}&\frac{1}{4}&0&\frac{N - 2}{4N}&-\frac{1}{8}\\
		-\frac{1}{8}&\frac{1}{2}&-\frac{3}{2}&-\frac{1}{2}&-\frac{N + 4}{8N}
	\end{bmatrix}
	\begin{bmatrix}
		\lambda_{21\,11}\\\lambda_{21\,21}\\\lambda_{21\,31}\\\lambda_{21\,41}\\\lambda_{21\,51}		
	\end{bmatrix} = 0\\
	&\begin{bmatrix}
		\theta_{10}\\\theta_{20}\\\theta_{30}\\\theta_{40}\\\theta_{50}
	\end{bmatrix} + 
	\begin{bmatrix}
		-\frac{N + 4}{8N}&-\frac{1}{2}&-\frac{3}{2}&\frac{1}{2}&-\frac{1}{8}\\
		-\frac{1}{2}&\frac{N - 2}{4N}&0&\frac{1}{4}&\frac{1}{8}\\
		-\frac{1}{16}&0&\frac{N - 2}{4N}&0&-\frac{1}{16}\\
		\frac{1}{8}&\frac{1}{4}&0&\frac{N - 2}{4N}&-\frac{1}{8}\\
		-\frac{1}{8}&\frac{1}{2}&-\frac{3}{2}&-\frac{1}{2}&-\frac{N + 4}{8N}
	\end{bmatrix}
	\begin{bmatrix}
		\theta_{11}\\\theta_{21}\\\theta_{31}\\\theta_{41}\\\theta_{51}		
	\end{bmatrix} = 
	\begin{bmatrix}
		G_1\\G_2\\G_3\\G_4\\G_5		
	\end{bmatrix}\\
	&\theta_{\bar{i}} = \lambda_{11\,\bar{i}} + \sqrt{\varkappa_{21\,\bar{i}}\lambda_{21\,\bar{i}}}
\end{split} \label{trifermion/cancellation}
\end{equation}

It is easy to see that addition of three-fermion fields makes $\lambda_{11\bar{i}}$ essentially free parameters.

It is also possible to introduce 3-0 composite fields dual to 

\begin{equation}
\mu_{30\,ijk} = \psi_i\otimes\psi_j\otimes\psi_k
\end{equation}

Such tensor products form a reducible representation of the model's symmetry group and its decomposition
into irreducible components depends on color number $N$. The $N=3$ case is of special interest in the QCD context.

It is also worth noting that one-loop approximation of the effective action will not contain a proper kinetic term for three-fermion fields:
if all fields are split into background and deviation parts, the quadratic in deviations terms of the composite fields lagrangian will contain
the background fields $\bar{\Sigma}_{21\bar{i}\bar{\alpha}}$ and $\Sigma_{21\bar{i}\bar{\alpha}}$ coupled to $\Psi$, $\bar{\Psi}$ only.

Three-fermion fields will acquire kinetic term from renormalization of higher-loop diagrams, the first one being sunset

\begin{equation}
	\begin{tikzpicture}[baseline=(a)]
		\begin{feynman}
			\vertex(a);
			\vertex [right=1cm of a] (b);
			\vertex [left=0.5cm of a] (o) {\(\bar{\Sigma}_{21}\)};
			\vertex [right=0.5cm of b] (i) {\(\Sigma_{21}\)};
			\diagram* {
				(a) --[fermion, very thick] (o),
				(b) --[fermion, half right] (a),
				(a) --[fermion] (b),
				(a) --[fermion, half right] (b),
				(i) --[fermion, very thick] (b),
			};
		\end{feynman}
	\end{tikzpicture}
\end{equation}

Constructing effective action for QCD with two- and three-fermion fields up to two loops is a task for further research.

\section{Conclusions} \label{conclusions}

The conclusions that can be drawn from the study above fall into two broad groups. 

The first is relevant to NJL and related models. 
Non-uniqueness of bosonization due to Fierz identities implies that 
particularities of fermion interaction in such models do not manifest themselves
at one-loop level. This also sheds some light on the relation between NJL and QCD:
the two theories produce similar meson dynamics up to one loop.

Interestingly, if three-fermion composite fields are considered,
one-loop approximation is not enough to give them physical dynamics 
and it is necessary to build effective action up to two loops at least, where the interaction would matter. 

Non-uniqueness of bosonization is also relevant to the subquark model of fundamental interactions
proposed by Terazawa, Chikashige, and Akama \cite{Terazawa1977}.
The one-loop finite renormalization Lagrangian \eqref{onelooprenorm/renormterms} has vector boson and higgsoid sectors 
and, as discussed above, it does so regardless of the subtleties of underlying subquark interactions.

The biggest obstacle to such subquark theories is, of course, the need to produce emergent gauge symmetry. 
The formalism developed here naturally produces almost gauge invariant effective actions, where the only symmetry breaking terms are the vector boson mass terms.
Interestingly, if symmetry conditions are imposed, these terms cannot be cancelled exactly by renormalization.
Nevertheless, there is enough freedom to make observable vector boson masses arbitrarily small.

What \eqref{onelooprenorm/renormterms} also predicts is plenty of interesting nonlinear phenomena.
In particular, the effective field equations should have 't Hooft-Polyakov monopole-like solutions.
However, if vector field has mass, albeit small, the monopole energy becomes infinite. 
This would neatly explain why monopoles are not observed.

Another type of nonlinear field equation solutions are kinks. 
Their existence is determined by $\chi$-symmetry \eqref{oneloopterms/chisymmetry} under which scalar composite fields change signs.
This means that if 
\begin{equation}
\begin{split}
\Sigma_{10} = M = \text{const}\\
\Sigma_{11}^a = M^a = \text{const}
\end{split}
\end{equation}
is a solution of the effective field equations, there must also be an opposite-sign solution, and there must also be a solution that 
approaches $M^{(a)}$ and $-M^{(a)}$ asymptotically in the opposite spatial directions. 
This is another evidence that Wightman's vacuum homogeneity assumption is incompatible with effective action formalism
since the latter naturally produces inhomogeneous ground states.

The second group of conclusions comprises the considerations relevant to the general formalism of quantum field theory.
Application of R-operation to effective action formalism is obviously one such result.

Applicability of composite fields formalism to any field theory raises the question whether it preserves renormalizability.
In the traditional R-operation formalism, renormalizability arises from extra requirements, such as gauge invariance, being imposed on renormalization constants
and making the number of independent parameters finite. 
The question is whether this persists when composite fields are added. 
If not, this would mean blurring of the line between renormalizable and non-renormalizable theories.
In the latter case, renormalizablity would be the property of some sectors of the theory only.

But the main result here is the relation between bosonization and quantum BBGKY hierarchy. Bosonization, far from being an ad hoc recipe, is demonstrated to be a special case of the latter.
Physically, it means that nontrivial quantum states and processes involving them are under consideration. Successful application of bosonization to NJL model suggests that in some theories such nontrivial quantum states may comprise the actual low-energy sector. Interestingly, presence of various condensates suggests that QCD also belongs to this class (see e.g. \cite{Ioffe2002}), 
which makes studying composite fields dynamics in QCD a promising research program.

\appendix

\section{Computation of 6-order Fierz identities} \label{fierz6comp}

Reduction of a generic term \eqref{fierz6/genterms} to ta sum of canonical terms \eqref{fierz6/gammaterms} is a two-stage process.
In the first stage, the formulas \eqref{fierz6/chiforms} are used to express it in color indices. This results in
two types of nontrivial terms
\begin{equation}
\begin{split}
\Lgr_{1ijk}^{\bar{\alpha}\bar{\beta}\bar{\gamma}} = 
\bar{\psi}_{l_1}\Gamma^{\bar{\alpha}}_i\psi_{l_1}
\bar{\psi}_{l_2}\Gamma^{\bar{\beta}}_j\psi_{l_3}
\bar{\psi}_{l_3}\Gamma^{\bar{\gamma}}_k\psi_{l_2}\\
\Lgr_{2ijk}^{\bar{\alpha}\bar{\beta}\bar{\gamma}} = 
\bar{\psi}_{l_1}\Gamma^{\bar{\alpha}}_i\psi_{l_2}
\bar{\psi}_{l_2}\Gamma^{\bar{\beta}}_j\psi_{l_3}
\bar{\psi}_{l_3}\Gamma^{\bar{\gamma}}_k\psi_{l_1}\\
\end{split}
\end{equation}

Then \eqref{fierz4/cliffordexpansion} is used to convert $\Lgr_2$ to $\Lgr_1$ and $\Lgr_1$s to canonical terms.
All that remains is to simplify products of $\gamma$-matrices. This was done with the help of \textit{dirac} software
package \cite{Kutnii2023}. The resulting identities are listed below.

\begin{equation}
\begin{split}
&\bar{\psi}_{i_1}\psi_{i_1}\bar{\psi}_{i_2}\psi_{i_3}\bar{\psi}_{i_3}\psi_{i_2} = \bar{\psi}_{i_1}\psi_{i_2}\bar{\psi}_{i_2}\psi_{i_1}\bar{\psi}_{i_3}\psi_{i_3} =
-\frac{1}{4}\bar{\psi}_{i_1}\psi_{i_1}\bar{\psi}_{i_2}\psi_{i_2}\bar{\psi}_{i_3}\psi_{i_3} 
-\frac{1}{4}\bar{\psi}_{i_1}\psi_{i_1}\bar{\psi}_{i_2}\gamma^5\psi_{i_2}\bar{\psi}_{i_3}\gamma^5\psi_{i_3} +\\
&+ \frac{1}{4}\bar{\psi}_{i_1}\psi_{i_1}\bar{\psi}_{i_2}\gamma^5\gamma_\mu\psi_{i_2}\bar{\psi}_{i_3}\gamma^5\gamma^\mu\psi_{i_3} 
-\frac{1}{4}\bar{\psi}_{i_1}\psi_{i_1}\bar{\psi}_{i_2}\gamma_\mu\psi_{i_2}\bar{\psi}_{i_3}\gamma^\mu\psi_{i_3}
 -\frac{1}{8}\bar{\psi}_{i_1}\psi_{i_1}\bar{\psi}_{i_2}\sigma_{\lambda\mu}\psi_{i_2}\bar{\psi}_{i_3}\sigma^{\lambda\mu}\psi_{i_3}
\end{split}
\end{equation}

\begin{equation}
\begin{split}
&\bar{\psi}_{i_1}\psi_{i_2}\bar{\psi}_{i_2}\psi_{i_3}\bar{\psi}_{i_3}\psi_{i_1} =
\frac{1}{16}\bar{\psi}_{i_1}\psi_{i_1}\bar{\psi}_{i_2}\psi_{i_2}\bar{\psi}_{i_3}\psi_{i_3}
+\frac{3}{16}\bar{\psi}_{i_1}\psi_{i_1}\bar{\psi}_{i_2}\gamma^5\psi_{i_2}\bar{\psi}_{i_3}\gamma^5\psi_{i_3} -\\
 &- \frac{i}{32}\epsilon_{\kappa\lambda\mu\nu}\bar{\psi}_{i_1}\gamma^5\psi_{i_1}\bar{\psi}_{i_2}\sigma^{\kappa\lambda}\psi_{i_2}\bar{\psi}_{i_3}\sigma^{\mu\nu}\psi_{i_3}
+\frac{9}{16}\bar{\psi}_{i_1}\psi_{i_1}\bar{\psi}_{i_2}\gamma_\mu\psi_{i_2}\bar{\psi}_{i_3}\gamma^\mu\psi_{i_3} 
+ \frac{3}{16}\bar{\psi}_{i_1}\gamma^5\psi_{i_1}\bar{\psi}_{i_2}\gamma_\mu\psi_{i_2}\bar{\psi}_{i_3}\gamma^5\gamma^\mu\psi_{i_3}  -\\
&-\frac{3}{8}\bar{\psi}_{i_1}\psi_{i_1}\bar{\psi}_{i_2}\gamma^5\gamma_\mu\psi_{i_2}\bar{\psi}_{i_3}\gamma^5\gamma^\mu\psi_{i_3}  
- \frac{3}{32}\epsilon_{\kappa\lambda\mu\nu}\bar{\psi}_{i_1}\gamma^\kappa\psi_{i_1}\bar{\psi}_{i_2}\gamma^5\gamma^\lambda\psi_{i_2}\bar{\psi}_{i_3}\sigma^{\mu\nu}\psi_{i_3}
+\frac{13}{32}\bar{\psi}_{i_1}\psi_{i_1}\bar{\psi}_{i_2}\sigma_{\lambda\mu}\psi_{i_2}\bar{\psi}_{i_3}\sigma^{\lambda\mu}\psi_{i_3}
\end{split}
\end{equation}

\begin{equation}
\begin{split}
&\bar{\psi}_{i_1}\psi_{i_1}\bar{\psi}_{i_2}\gamma_\mu\psi_{i_3}\bar{\psi}_{i_3}\gamma^\mu\psi_{i_2} = 
-\bar{\psi}_{i_1}\psi_{i_1}\bar{\psi}_{i_2}\psi_{i_2}\bar{\psi}_{i_3}\psi_{i_3} +
\bar{\psi}_{i_1}\psi_{i_1}\bar{\psi}_{i_2}\gamma^5\psi_{i_2}\bar{\psi}_{i_3}\gamma^5\psi_{i_3} +\\
 &+\frac{1}{2}\bar{\psi}_{i_1}\psi_{i_1}\bar{\psi}_{i_2}\gamma_\mu\psi_{i_2}\bar{\psi}_{i_3}\gamma^\mu\psi_{i_3}
 + \frac{1}{2}\bar{\psi}_{i_1}\psi_{i_1}\bar{\psi}_{i_2}\gamma^5\gamma_\mu\psi_{i_2}\bar{\psi}_{i_3}\gamma^5\gamma^\mu\psi_{i_3} 
\end{split}
\end{equation}

\begin{equation}
\bar{\psi}_{i_1}\psi_{i_2}\bar{\psi}_{i_2}\gamma_\mu\psi_{i_1}\bar{\psi}_{i_3}\gamma^\mu\psi_{i_3} = 
-\frac{1}{2}\bar{\psi}_{i_1}\psi_{i_1}\bar{\psi}_{i_2}\gamma_\mu\psi_{i_2}\bar{\psi}_{i_3}\gamma^\mu\psi_{i_3} 
+\frac{1}{4}\epsilon_{\kappa\lambda\mu\nu}\bar{\psi}_{i_1}\gamma^\kappa\psi_{i_1}\bar{\psi}_{i_2}\gamma^5\gamma^\lambda\psi_{i_2}\bar{\psi}_{i_3}\sigma^{\mu\nu}\psi_{i_3}
\end{equation}

\begin{equation}
\begin{split}
&\bar{\psi}_{i_1}\psi_{i_2}\bar{\psi}_{i_2}\gamma_\mu\psi_{i_3}\bar{\psi}_{i_3}\gamma^\mu\psi_{i_1} = 
\frac{1}{4}\bar{\psi}_{i_1}\psi_{i_1}\bar{\psi}_{i_2}\psi_{i_2}\bar{\psi}_{i_3}\psi_{i_3}
-\frac{1}{4}\bar{\psi}_{i_1}\psi_{i_1}\bar{\psi}_{i_2}\gamma^5\psi_{i_2}\bar{\psi}_{i_3}\gamma^5\psi_{i_3} - \\
& -\frac{3}{4}\bar{\psi}_{i_1}\psi_{i_1}\bar{\psi}_{i_2}\gamma_\mu\psi_{i_2}\bar{\psi}_{i_3}\gamma^\mu\psi_{i_3}
-\frac{3}{8}\bar{\psi}_{i_1}\gamma^5\psi_{i_1}\bar{\psi}_{i_2}\gamma_\mu\psi_{i_2}\bar{\psi}_{i_3}\gamma^5\gamma^\mu\psi_{i_3}
+ \frac{i}{16}\epsilon_{\kappa\lambda\mu\nu}\bar{\psi}_{i_1}\gamma^5\psi_{i_1}\bar{\psi}_{i_2}\sigma^{\kappa\lambda}\psi_{i_2}\bar{\psi}_{i_3}\sigma^{\mu\nu}\psi_{i_3}  -\\
&-\frac{7}{8}\bar{\psi}_{i_1}\psi_{i_1}\bar{\psi}_{i_2}\gamma^5\gamma_\mu\psi_{i_2}\bar{\psi}_{i_3}\gamma^5\gamma^\mu\psi_{i_3} 
+ \frac{1}{8}\bar{\psi}_{i_1}\psi_{i_1}\bar{\psi}_{i_2}\sigma_{\lambda\mu}\psi_{i_2}\bar{\psi}_{i_3}\sigma^{\lambda\mu}\psi_{i_3} 
\end{split}
\end{equation}

\begin{equation}
\begin{split}
\bar{\psi}_{i_1}\psi_{i_1}\bar{\psi}_{i_2}\sigma_{\lambda\mu}\psi_{i_3}\bar{\psi}_{i_3}\sigma^{\lambda\mu}\psi_{i_2} = &
-3\bar{\psi}_{i_1}\psi_{i_1}\bar{\psi}_{i_2}\psi_{i_2}\bar{\psi}_{i_3}\psi_{i_3} -\\
&-3\bar{\psi}_{i_1}\psi_{i_1}\bar{\psi}_{i_2}\gamma^5\psi_{i_2}\bar{\psi}_{i_3}\gamma^5\psi_{i_3}
 + \frac{1}{2}\bar{\psi}_{i_1}\psi_{i_1}\bar{\psi}_{i_2}\sigma_{\lambda\mu}\psi_{i_2}\bar{\psi}_{i_3}\sigma^{\lambda\mu}\psi_{i_3}
\end{split}
\end{equation}

\begin{equation}
\begin{split}
&\bar{\psi}_{i_1}\psi_{i_2}\bar{\psi}_{i_2}\sigma_{\sigma_{1}\sigma_{2}}\psi_{i_1}\bar{\psi}_{i_3}\sigma^{\sigma_{1}\sigma_{2}}\psi_{i_3} = 
\frac{i}{4}\epsilon_{\kappa\lambda\mu\nu}\bar{\psi}_{i_1}\gamma^5\psi_{i_1}\bar{\psi}_{i_2}\sigma^{\kappa\lambda}\psi_{i_2}\bar{\psi}_{i_3}\sigma^{\mu\nu}\psi_{i_3} - \\
&-\frac{1}{2}\bar{\psi}_{i_1}\psi_{i_1}\bar{\psi}_{i_2}\sigma_{\lambda\mu}\psi_{i_2}\bar{\psi}_{i_3}\sigma^{\lambda\mu}\psi_{i_3}
+\frac{1}{2}\epsilon_{\kappa\lambda\mu\nu}\bar{\psi}_{i_1}\gamma^{\kappa}\psi_{i_1}\bar{\psi}_{i_2}\gamma^5\gamma^{\lambda}\psi_{i_2}\bar{\psi}_{i_3}\sigma^{\mu\nu}\psi_{i_3} 
 - \frac{i}{2}\bar{\psi}_{i_1}{\sigma_\kappa}^\lambda\psi_{i_1}\bar{\psi}_{i_2}{\sigma_\lambda}^\mu\psi_{i_2}\bar{\psi}_{i_3}{\sigma_\mu}^\kappa\psi_{i_3}
\end{split}
\end{equation}

\begin{equation}
\begin{split}
&\bar{\psi}_{i_1}\psi_{i_2}\bar{\psi}_{i_2}\sigma_{\sigma_{1}\sigma_{2}}\psi_{i_3}\bar{\psi}_{i_3}\sigma^{\sigma_{1}\sigma_{2}}\psi_{i_1} =\frac{3}{4}\bar{\psi}_{i_1}\psi_{i_1}\bar{\psi}_{i_2}\psi_{i_2}\bar{\psi}_{i_3}\psi_{i_3} 
+ \frac{9}{4}\bar{\psi}_{i_1}\psi_{i_1}\bar{\psi}_{i_2}\gamma^5\psi_{i_2}\bar{\psi}_{i_3}\gamma^5\psi_{i_3} + \\
& + \frac{3}{4}\bar{\psi}_{i_1}\psi_{i_1}\bar{\psi}_{i_2}\gamma_\mu\psi_{i_2}\bar{\psi}_{i_3}\gamma^\mu\psi_{i_3}  
- \frac{i}{8}\epsilon_{\kappa\lambda\mu\nu}\bar{\psi}_{i_1}\gamma^5\psi_{i_1}\bar{\psi}_{i_2}\sigma^{\kappa\lambda}\psi_{i_2}\bar{\psi}_{i_3}\sigma^{\mu\nu}\psi_{i_3} -\frac{9}{8}\bar{\psi}_{i_1}\psi_{i_1}\bar{\psi}_{i_2}\sigma_{\kappa\lambda}\psi_{i_2}\bar{\psi}_{i_3}\sigma^{\mu\nu}\psi_{i_3} +\\
&+ \frac{1}{8}\epsilon_{\kappa\lambda\mu\nu}\bar{\psi}_{i_1}\gamma^\kappa\psi_{i_1}\bar{\psi}_{i_2}\gamma^5\gamma^\lambda\psi_{i_2}\bar{\psi}_{i_3}\sigma^{\mu\nu}\psi_{i_3}  
-\frac{3}{4}\bar{\psi}_{i_1}\psi_{i_1}\bar{\psi}_{i_2}\gamma^5\gamma_\mu\psi_{i_2}\bar{\psi}_{i_3}\gamma^5\gamma^\mu\psi_{i_3} 
\end{split}
\end{equation}

\begin{equation}
\begin{split}
&\bar{\psi}_{i_1}\psi_{i_1}\bar{\psi}_{i_2}\gamma^5\gamma_\mu\psi_{i_3}\bar{\psi}_{i_3}\gamma^5\gamma^\mu\psi_{i_2} =
\bar{\psi}_{i_1}\psi_{i_1}\bar{\psi}_{i_2}\psi_{i_2}\bar{\psi}_{i_3}\psi_{i_3}
 -\bar{\psi}_{i_1}\psi_{i_1}\bar{\psi}_{i_2}\gamma^5\psi_{i_2}\bar{\psi}_{i_3}\gamma^5\psi_{i_3} + \\
& +\frac{1}{2}\bar{\psi}_{i_1}\psi_{i_1}\bar{\psi}_{i_2}\gamma_\mu\psi_{i_2}\bar{\psi}_{i_3}\gamma^\mu\psi_{i_3}
+ \frac{1}{2}\bar{\psi}_{i_1}\psi_{i_1}\bar{\psi}_{i_2}\gamma^5\gamma_\mu\psi_{i_2}\bar{\psi}_{i_3}\gamma^5\gamma^\mu\psi_{i_3} 
\end{split}
\end{equation}

\begin{equation}
\begin{split}
&\bar{\psi}_{i_1}\psi_{i_2}\bar{\psi}_{i_2}\gamma^5\gamma_\mu\psi_{i_1}\bar{\psi}_{i_3}\gamma^5\gamma^\mu\psi_{i_3} = 
-\frac{1}{2}\bar{\psi}_{i_1}\psi_{i_1}\bar{\psi}_{i_2}\gamma^5\gamma_\mu\psi_{i_2}\bar{\psi}_{i_3}\gamma^5\gamma^\mu\psi_{i_3} 
-\frac{1}{4}\epsilon_{\kappa\lambda\mu\nu}\bar{\psi}_{i_1}\gamma^\kappa\psi_{i_1}\bar{\psi}_{i_2}\gamma^5\gamma^\lambda\psi_{i_2}\bar{\psi}_{i_3}\sigma^{\mu\nu}\psi_{i_3}
\end{split}
\end{equation}

\begin{equation}
\begin{split}
&\bar{\psi}_{i_1}\psi_{i_2}\bar{\psi}_{i_2}\gamma^5\gamma_\mu\psi_{i_3}\bar{\psi}_{i_3}\gamma^5\gamma^\mu\psi_{i_1} = 
-\frac{1}{4}\bar{\psi}_{i_1}\psi_{i_1}\bar{\psi}_{i_2}\psi_{i_2}\bar{\psi}_{i_3}\psi_{i_3} 
+ \frac{1}{4}\bar{\psi}_{i_1}\psi_{i_1}\bar{\psi}_{i_2}\gamma^5\psi_{i_2}\bar{\psi}_{i_3}\gamma^5\psi_{i_3}  -\\
&  -\frac{3}{8}\bar{\psi}_{i_1}\gamma^5\psi_{i_1}\bar{\psi}_{i_2}\gamma_\mu\psi_{i_2}\bar{\psi}_{i_3}\gamma^5\gamma^\mu\psi_{i_3} 
-\frac{5}{4}\bar{\psi}_{i_1}\psi_{i_1}\bar{\psi}_{i_2}\gamma_\mu\psi_{i_2}\bar{\psi}_{i_3}\gamma^\mu\psi_{i_3}  
- \frac{i}{16}\epsilon_{\kappa\lambda\mu\nu}\bar{\psi}_{i_1}\gamma^5\psi_{i_1}\bar{\psi}_{i_2}\sigma^{\kappa\lambda}\psi_{i_2}\bar{\psi}_{i_3}\sigma^{\mu\nu}\psi_{i_3}  -\\
&-\frac{3}{8}\bar{\psi}_{i_1}\psi_{i_1}\bar{\psi}_{i_2}\gamma^5\gamma_\mu\psi_{i_2}\bar{\psi}_{i_3}\gamma^5\gamma^\mu\psi_{i_3} 
-\frac{1}{8}\bar{\psi}_{i_1}\psi_{i_1}\bar{\psi}_{i_2}\sigma_{\mu\nu}\psi_{i_2}\bar{\psi}_{i_3}\sigma^{\mu\nu}\psi_{i_3} 
\end{split}
\end{equation}

\begin{equation}
\begin{split}
&\bar{\psi}_{i_1}\psi_{i_1}\bar{\psi}_{i_2}\gamma^5\psi_{i_3}\bar{\psi}_{i_3}\gamma^5\psi_{i_2} =
-\frac{1}{4}\bar{\psi}_{i_1}\psi_{i_1}\bar{\psi}_{i_2}\psi_{i_2}\bar{\psi}_{i_3}\psi_{i_3}
-\frac{1}{4}\bar{\psi}_{i_1}\psi_{i_1}\bar{\psi}_{i_2}\gamma^5\psi_{i_2}\bar{\psi}_{i_3}\gamma^5\psi_{i_3} -\\
&-\frac{1}{4}\bar{\psi}_{i_1}\psi_{i_1}\bar{\psi}_{i_2}\gamma^5\gamma_\mu\psi_{i_2}\bar{\psi}_{i_3}\gamma^5\gamma^\mu\psi_{i_3}
+\frac{1}{4}\bar{\psi}_{i_1}\psi_{i_1}\bar{\psi}_{i_2}\gamma_\mu\psi_{i_2}\bar{\psi}_{i_3}\gamma^\mu\psi_{i_3} 
-\frac{1}{8}\bar{\psi}_{i_1}\psi_{i_1}\bar{\psi}_{i_2}\sigma_{\mu\nu}\psi_{i_2}\bar{\psi}_{i_3}\sigma^{\mu\nu}\psi_{i_3}
\end{split}
\end{equation}

\begin{equation}
\begin{split}
&\bar{\psi}_{i_1}\psi_{i_2}\bar{\psi}_{i_2}\gamma^5\psi_{i_1}\bar{\psi}_{i_3}\gamma^5\psi_{i_3} = 
-\frac{1}{2}\bar{\psi}_{i_1}\psi_{i_1}\bar{\psi}_{i_2}\gamma^5\psi_{i_2}\bar{\psi}_{i_3}\gamma^5\psi_{i_3}
+\frac{i}{16}\epsilon_{\kappa\lambda\mu\nu}\bar{\psi}_{i_1}\gamma^5\psi_{i_1}\bar{\psi}_{i_2}\sigma^{\kappa\lambda}\psi_{i_2}\bar{\psi}_{i_3}\sigma^{\mu\nu}\psi_{i_3}
\end{split}
\end{equation}

\begin{equation}
\begin{split}
&\bar{\psi}_{i_1}\psi_{i_2}\bar{\psi}_{i_2}\gamma^5\psi_{i_3}\bar{\psi}_{i_3}\gamma^5\psi_{i_1} = 
\frac{1}{16}\bar{\psi}_{i_1}\psi_{i_1}\bar{\psi}_{i_2}\psi_{i_2}\bar{\psi}_{i_3}\psi_{i_3}
+\frac{3}{16}\bar{\psi}_{i_1}\psi_{i_1}\bar{\psi}_{i_2}\gamma^5\psi_{i_2}\bar{\psi}_{i_3}\gamma^5\psi_{i_3} - \\
&- \frac{I}{32}\epsilon_{\kappa\lambda\mu\nu}\bar{\psi}_{i_1}\gamma^5\psi_{i_1}\bar{\psi}_{i_2}\sigma^{\kappa\lambda}\psi_{i_2}\bar{\psi}_{i_3}\sigma^{\mu\nu}\psi_{i_3} 
-\frac{7}{16}\bar{\psi}_{i_1}\psi_{i_1}\bar{\psi}_{i_2}\gamma_\mu\psi_{i_2}\bar{\psi}_{i_3}\gamma^\mu\psi_{i_3} -\\
&-\frac{3}{16}\bar{\psi}_{i_1}\gamma^5\psi_{i_1}\bar{\psi}_{i_2}\gamma_\mu\psi_{i_2}\bar{\psi}_{i_3}\gamma^5\gamma^\mu\psi_{i_3} 
+\frac{1}{4}\bar{\psi}_{i_1}\psi_{i_1}\bar{\psi}_{i_2}\gamma^5\gamma_\mu\psi_{i_2}\bar{\psi}_{i_3}\gamma^5\gamma^\mu\psi_{i_3} +\\
&+ \frac{1}{32}\epsilon_{\kappa\lambda\mu\nu}\bar{\psi}_{i_1}\gamma^\kappa\psi_{i_1}\bar{\psi}_{i_2}\gamma^5\gamma^\lambda\psi_{i_2}\bar{\psi}_{i_3}\sigma^{\mu\nu}\psi_{i_3}
+ \frac{13}{32}\bar{\psi}_{i_1}\psi_{i_1}\bar{\psi}_{i_2}\sigma_{\mu\nu}\psi_{i_2}\bar{\psi}_{i_3}\sigma^{\mu\nu}\psi_{i_3}
\end{split}
\end{equation}

\begin{equation}
\begin{split}
&\bar{\psi}_{i_1}\gamma^5\psi_{i_1}\bar{\psi}_{i_2}\gamma_{\kappa}\psi_{i_3}\bar{\psi}_{i_3}\gamma^5\gamma^{\kappa}\psi_{i_2} = 
\bar{\psi}_{i_1}\gamma^5\psi_{i_1}\bar{\psi}_{i_2}\gamma_\mu\psi_{i_2}\bar{\psi}_{i_3}\gamma^5\gamma^\mu\psi_{i_3}
\end{split}
\end{equation}

\begin{equation}
\begin{split}
&\bar{\psi}_{i_1}\gamma^5\psi_{i_2}\bar{\psi}_{i_2}\gamma_\mu\psi_{i_1}\bar{\psi}_{i_3}\gamma^5\gamma^\mu\psi_{i_3} = 
-\frac{1}{2}\bar{\psi}_{i_1}\gamma^5\psi_{i_1}\bar{\psi}_{i_2}\gamma_\mu\psi_{i_2}\bar{\psi}_{i_3}\gamma^5\gamma^\mu\psi_{i_3} 
\end{split}
\end{equation}

\begin{equation}
\begin{split}
&\bar{\psi}_{i_1}\gamma^5\psi_{i_2}\bar{\psi}_{i_2}\gamma_\mu\psi_{i_3}\bar{\psi}_{i_3}\gamma^5\gamma^\mu\psi_{i_1} = 
-\frac{1}{4}\bar{\psi}_{i_1}\psi_{i_1}\bar{\psi}_{i_2}\psi_{i_2}\bar{\psi}_{i_3}\psi_{i_3}
+ \frac{1}{4}\bar{\psi}_{i_1}\psi_{i_1}\bar{\psi}_{i_2}\gamma^5\psi_{i_2}\bar{\psi}_{i_3}\gamma^5\psi_{i_3}  -\\
& - \frac{i}{16}\epsilon_{\kappa\lambda\mu\nu}\bar{\psi}_{i_1}\gamma^5\psi_{i_1}\bar{\psi}_{i_2}\sigma^{\kappa\lambda}\psi_{i_2}\bar{\psi}_{i_3}\sigma^{\mu\nu}\psi_{i_3}
-\frac{1}{8}\bar{\psi}_{i_1}\psi_{i_1}\bar{\psi}_{i_2}\sigma_{\mu\nu}\psi_{i_2}\bar{\psi}_{i_3}\sigma^{\mu\nu}\psi_{i_3}  -\\
&-\frac{5}{4}\bar{\psi}_{i_1}\gamma^5\psi_{i_1}\bar{\psi}_{i_2}\gamma_\mu\psi_{i_2}\bar{\psi}_{i_3}\gamma^5\gamma^\mu\psi_{i_3}  
-\frac{5}{8}\bar{\psi}_{i_1}\psi_{i_1}\bar{\psi}_{i_2}\gamma^5\gamma_\mu\psi_{i_2}\bar{\psi}_{i_3}\gamma^5\gamma^\mu\psi_{i_3}
-\frac{1}{8}\bar{\psi}_{i_1}\psi_{i_1}\bar{\psi}_{i_2}\gamma_\mu\psi_{i_2}\bar{\psi}_{i_3}\gamma^\mu\psi_{i_3} 
\end{split}
\end{equation}

\begin{equation}
\begin{split}
&\bar{\psi}_{i_1}\gamma_{\kappa}\psi_{i_1}\bar{\psi}_{i_2}\gamma^5\psi_{i_3}\bar{\psi}_{i_3}\gamma^5\gamma^{\kappa}\psi_{i_2} = 
-\frac{1}{2}\bar{\psi}_{i_1}\gamma^5\psi_{i_1}\bar{\psi}_{i_2}\gamma_{\mu}\psi_{i_2}\bar{\psi}_{i_3}\gamma^5\gamma^\mu\psi_{i_3}
\end{split}
\end{equation}

\begin{equation}
\begin{split}
&\bar{\psi}_{i_1}\gamma_{\kappa}\psi_{i_2}\bar{\psi}_{i_2}\gamma^5\psi_{i_3}\bar{\psi}_{i_3}\gamma^5\gamma^{\kappa}\psi_{i_1} = 
\frac{1}{4}\bar{\psi}_{i_1}\psi_{i_1}\bar{\psi}_{i_2}\psi_{i_2}\bar{\psi}_{i_3}\psi_{i_3}
+\frac{1}{4}\bar{\psi}_{i_1}\gamma^5\psi_{i_1}\bar{\psi}_{i_2}\gamma_\mu\psi_{i_2}\bar{\psi}_{i_3}\gamma^5\gamma^\mu\psi_{i_3} -\\
&-\frac{1}{4}\bar{\psi}_{i_1}\psi_{i_1}\bar{\psi}_{i_2}\gamma^5\psi_{i_2}\bar{\psi}_{i_3}\gamma^5\psi_{i_3} 
+\frac{5}{16}\bar{\psi}_{i_1}\psi_{i_1}\bar{\psi}_{i_2}\gamma_\mu\psi_{i_2}\bar{\psi}_{i_3}\gamma^\mu\psi_{i_3} 
+ \frac{1}{16}\epsilon_{\kappa\lambda\mu\nu}\bar{\psi}_{i_1}\gamma^\kappa\psi_{i_1}\bar{\psi}_{i_2}\gamma^5\gamma^\lambda\psi_{i_2}\bar{\psi}_{i_3}\sigma^{\mu\nu}\psi_{i_3} + \\
&+\frac{7}{16}\bar{\psi}_{i_1}\psi_{i_1}\bar{\psi}_{i_2}\gamma^5\gamma_\mu\psi_{i_2}\bar{\psi}_{i_3}\gamma^5\gamma^\mu\psi_{i_3} 
+ \frac{1}{8}\bar{\psi}_{i_1}\psi_{i_1}\bar{\psi}_{i_2}\sigma_{\mu\nu}\psi_{i_2}\bar{\psi}_{i_3}\sigma^{\mu\nu}\psi_{i_3} 
+ \frac{i}{16}\epsilon_{\kappa\lambda\mu\nu}\bar{\psi}_{i_1}\gamma^5\psi_{i_1}\bar{\psi}_{i_2}\sigma^{\kappa\lambda}\psi_{i_2}\bar{\psi}_{i_3}\sigma^{\mu\nu}\psi_{i_3}
\end{split}
\end{equation}

\begin{equation}
\begin{split}
&\bar{\psi}_{i_1}\gamma_{\lambda}\psi_{i_1}\bar{\psi}_{i_2}\gamma_{\mu}\psi_{i_3}\bar{\psi}_{i_3}\sigma^{\lambda\mu}\psi_{i_2} =  - \frac{3i}{2}\bar{\psi}_{i_1}\gamma^5\psi_{i_1}\bar{\psi}_{i_2}\gamma_\mu\psi_{i_2}\bar{\psi}_{i_3}\gamma^5\gamma^\mu\psi_{i_3} 
\end{split}
\end{equation}

\begin{equation}
\begin{split}
&\bar{\psi}_{i_1}\gamma_{\lambda}\psi_{i_2}\bar{\psi}_{i_2}\gamma_{\mu}\psi_{i_1}\bar{\psi}_{i_3}\sigma^{\lambda\mu}\psi_{i_3} = 0
\end{split}
\end{equation}

\begin{equation}
\begin{split}
&\bar{\psi}_{i_1}\gamma_{\lambda}\psi_{i_2}\bar{\psi}_{i_2}\gamma_{\mu}\psi_{i_3}\bar{\psi}_{i_3}\sigma^{\lambda\mu}\psi_{i_1} =  
- \frac{3i}{4}\bar{\psi}_{i_1}\psi_{i_1}\bar{\psi}_{i_2}\psi_{i_2}\bar{\psi}_{i_3}\psi_{i_3} + \\
&+\frac{15i}{8}\bar{\psi}_{i_1}\gamma^5\psi_{i_1}\bar{\psi}_{i_2}\gamma_\mu\psi_{i_2}\bar{\psi}_{i_3}\gamma^5\gamma^\mu\psi_{i_3} 
+ \frac{3i}{4}\bar{\psi}_{i_1}\psi_{i_1}\bar{\psi}_{i_2}\gamma^5\psi_{i_2}\bar{\psi}_{i_3}\gamma^5\psi_{i_3}  -\\
& - \frac{9i}{16}\bar{\psi}_{i_1}\psi_{i_1}\bar{\psi}_{i_2}\gamma_\mu\psi_{i_2}\bar{\psi}_{i_3}\gamma^\mu\psi_{i_3} 
- \frac{i}{8}\epsilon_{\kappa\lambda\mu\nu}\bar{\psi}_{i_1}\gamma^\kappa\psi_{i_1}\bar{\psi}_{i_2}\gamma^5\gamma^\lambda\psi_{i_2}\bar{\psi}_{i_3}\sigma^{\mu\nu}\psi_{i_3} -\\
& - \frac{9i}{16}\bar{\psi}_{i_1}\psi_{i_1}\bar{\psi}_{i_2}\gamma^5\gamma_\mu\psi_{i_2}\bar{\psi}_{i_3}\gamma^5\gamma^\mu\psi_{i_3} 
+ \frac{i}{8}\bar{\psi}_{i_1}\psi_{i_1}\bar{\psi}_{i_2}\sigma_{\mu\nu}\psi_{i_2}\bar{\psi}_{i_3}\sigma^{\mu\nu}\psi_{i_3} 
- \frac{1}{16}\epsilon_{\kappa\lambda\mu\nu}\bar{\psi}_{i_1}\gamma^5\psi_{i_1}\bar{\psi}_{i_2}\sigma^{\kappa\lambda}\psi_{i_2}\bar{\psi}_{i_3}\sigma^{\mu\nu}\psi_{i_3}
\end{split}
\end{equation}

\begin{equation}
\begin{split}
&\bar{\psi}_{i_1}\gamma^5\gamma_{\lambda}\psi_{i_1}\bar{\psi}_{i_2}\gamma^5\gamma_{\mu}\psi_{i_3}\bar{\psi}_{i_3}\sigma^{\lambda\mu}\psi_{i_2} =  - \frac{3i}{2}\bar{\psi}_{i_1}\gamma^5\psi_{i_1}\bar{\psi}_{i_2}\gamma_\mu\psi_{i_2}\bar{\psi}_{i_3}\gamma^5\gamma^\mu\psi_{i_3}
\end{split}
\end{equation}

\begin{equation}
\begin{split}
&\bar{\psi}_{i_1}\gamma^5\gamma_{\lambda}\psi_{i_2}\bar{\psi}_{i_2}\gamma^5\gamma_{\mu}\psi_{i_1}\bar{\psi}_{i_3}\sigma^{\lambda\mu}\psi_{i_3} =  0
\end{split}
\end{equation}

\begin{equation}
\begin{split}
&\bar{\psi}_{i_1}\gamma^5\gamma_{\lambda}\psi_{i_2}\bar{\psi}_{i_2}\gamma^5\gamma_{\mu}\psi_{i_3}\bar{\psi}_{i_3}\sigma^{\lambda\mu}\psi_{i_1} = 
\frac{3i}{4}\bar{\psi}_{i_1}\psi_{i_1}\bar{\psi}_{i_2}\psi_{i_2}\bar{\psi}_{i_3}\psi_{i_3}  -\\
& - \frac{3i}{4}\bar{\psi}_{i_1}\psi_{i_1}\bar{\psi}_{i_2}\gamma^5\psi_{i_2}\bar{\psi}_{i_3}\gamma^5\psi_{i_3}  
+\frac{1}{16}\epsilon_{\kappa\lambda\mu\nu}\bar{\psi}_{i_1}\gamma^5\psi_{i_1}\bar{\psi}_{i_2}\sigma^{\kappa\lambda}\psi_{i_2}\bar{\psi}_{i_3}\sigma^{\mu\nu}\psi_{i_3} + \\
&+\frac{3i}{16}\bar{\psi}_{i_1}\psi_{i_1}\bar{\psi}_{i_2}\gamma_\mu\psi_{i_2}\bar{\psi}_{i_3}\gamma^\mu\psi_{i_3}   
+ \frac{i}{16}\epsilon_{\kappa\lambda\mu\nu}\bar{\psi}_{i_1}\gamma^\kappa\psi_{i_1}\bar{\psi}_{i_2}\gamma^5\gamma^\lambda\psi_{i_2}\bar{\psi}_{i_3}\sigma^{\mu\nu}\psi_{i_3} + \\
&+\frac{9i}{8}\bar{\psi}_{i_1}\gamma^5\psi_{i_1}\bar{\psi}_{i_2}\gamma_\mu\psi_{i_2}\bar{\psi}_{i_3}\gamma^5\gamma^\mu\psi_{i_3} 
- \frac{i}{8}\bar{\psi}_{i_1}\psi_{i_1}\bar{\psi}_{i_2}\sigma_{\mu\nu}\psi_{i_2}\bar{\psi}_{i_3}\sigma^{\mu\nu}\psi_{i_3} 
+ \frac{3i}{16}\bar{\psi}_{i_1}\psi_{i_1}\bar{\psi}_{i_2}\gamma^5\gamma_\mu\psi_{i_2}\bar{\psi}_{i_3}\gamma^5\gamma^\mu\psi_{i_3}
\end{split}
\end{equation}

\begin{equation}
\begin{split}
&\epsilon_{\kappa\lambda\mu\nu}\bar{\psi}_{i_1}\gamma^{\kappa}\psi_{i_1}\bar{\psi}_{i_2}\gamma^5\gamma^{\lambda}\psi_{i_3}\bar{\psi}_{i_3}\sigma^{\mu\nu}\psi_{i_2} = 
3\bar{\psi}_{i_1}\psi_{i_1}\bar{\psi}_{i_2}\gamma_\mu\psi_{i_2}\bar{\psi}_{i_3}\gamma^\mu\psi_{i_3}
-\frac{1}{2}\epsilon_{\kappa\lambda\mu\nu}\bar{\psi}_{i_1}\gamma^{\kappa}\psi_{i_1}\bar{\psi}_{i_2}\gamma^5\gamma^\lambda\psi_{i_2}\bar{\psi}_{i_3}\sigma^{\mu\nu}\psi_{i_3}
\end{split}
\end{equation}

\begin{equation}
\begin{split}
&\epsilon_{\kappa\lambda\mu\nu}\bar{\psi}_{i_1}\gamma^{\kappa}\psi_{i_2}\bar{\psi}_{i_2}\gamma^5\gamma^{\lambda}\psi_{i_1}\bar{\psi}_{i_3}\sigma^{\mu\nu}\psi_{i_3} = 
\frac{i}{2}\epsilon_{\kappa\lambda\mu\nu}\bar{\psi}_{i_1}\gamma^5\psi_{i_1}\bar{\psi}_{i_2}\sigma^{\kappa\lambda}\psi_{i_2}\bar{\psi}_{i_3}\sigma^{\mu\nu}\psi_{i_3}
+ \bar{\psi}_{i_1}\psi_{i_1}\bar{\psi}_{i_2}\sigma_{\mu\nu}\psi_{i_2}\bar{\psi}_{i_3}\sigma^{\mu\nu}\psi_{i_3}
\end{split}
\end{equation}

\begin{equation}
\begin{split}
&\epsilon_{\kappa\lambda\mu\nu}\bar{\psi}_{i_1}\gamma^{\kappa}\psi_{i_2}\bar{\psi}_{i_2}\gamma^5\gamma^{\lambda}\psi_{i_3}\bar{\psi}_{i_3}\sigma^{\mu\nu}\psi_{i_1} = 
-\frac{3}{2}\bar{\psi}_{i_1}\psi_{i_1}\bar{\psi}_{i_2}\psi_{i_2}\bar{\psi}_{i_3}\psi_{i_3} + \\
&+\frac{3}{2}\bar{\psi}_{i_1}\psi_{i_1}\bar{\psi}_{i_2}\gamma^5\psi_{i_2}\bar{\psi}_{i_3}\gamma^5\psi_{i_3} 
+ \frac{9}{4}\bar{\psi}_{i_1}\gamma^5\psi_{i_1}\bar{\psi}_{i_2}\gamma_\mu\psi_{i_2}\bar{\psi}_{i_3}\gamma^5\gamma^\mu\psi_{i_3}  -\\
&-\frac{15}{8}\bar{\psi}_{i_1}\psi_{i_1}\bar{\psi}_{i_2}\gamma_\mu\psi_{i_2}\bar{\psi}_{i_3}\gamma^\mu\psi_{i_3} 
+ \frac{1}{4}\bar{\psi}_{i_1}\psi_{i_1}\bar{\psi}_{i_2}\sigma_{\mu\nu}\psi_{i_2}\bar{\psi}_{i_3}\sigma^{\mu\nu}\psi_{i_3}  
-\frac{15}{8}\bar{\psi}_{i_1}\psi_{i_1}\bar{\psi}_{i_2}\gamma^5\gamma_\mu\psi_{i_2}\bar{\psi}_{i_3}\gamma^5\gamma^\mu\psi_{i_3} + \\
&+\frac{1}{4}\epsilon_{\kappa\lambda\mu\nu}\bar{\psi}_{i_1}\gamma^\kappa\psi_{i_1}\bar{\psi}_{i_2}\gamma^5\gamma^\lambda\psi_{i_2}\bar{\psi}_{i_3}\sigma^{\mu\nu}\psi_{i_3} + 
\frac{i}{8}\epsilon_{\kappa\lambda\mu\nu}\bar{\psi}_{i_1}\gamma^5\psi_{i_1}\bar{\psi}_{i_2}\sigma^{\kappa\lambda}\psi_{i_2}\bar{\psi}_{i_3}\sigma^{\mu\nu}\psi_{i_3}
\end{split}
\end{equation}

\begin{equation}
\epsilon_{\kappa\lambda\mu\nu}\bar{\psi}_{i_1}\gamma^5\gamma^{\kappa}\psi_{i_1}\bar{\psi}_{i_2}\gamma^{\lambda}\psi_{i_3}\bar{\psi}_{i_3}\sigma^{\mu\nu}\psi_{i_2} = 
3\bar{\psi}_{i_1}\psi_{i_1}\bar{\psi}_{i_2}\gamma^5\gamma_\mu\psi_{i_2}\bar{\psi}_{i_3}\gamma^5\gamma^\mu\psi_{i_3}
-\frac{1}{2}\epsilon_{\kappa\lambda\mu\nu}\bar{\psi}_{i_1}\gamma^{\kappa}\psi_{i_1}\bar{\psi}_{i_2}\gamma^5\gamma^\lambda\psi_{i_2}\bar{\psi}_{i_3}\sigma^{\mu\nu}\psi_{i_3}
\end{equation}

\begin{equation}
\begin{split}
&\epsilon_{\kappa\lambda\mu\nu}\bar{\psi}_{i_1}\gamma^5\gamma^{\kappa}\psi_{i_2}\bar{\psi}_{i_2}\gamma^{\lambda}\psi_{i_3}\bar{\psi}_{i_3}\sigma^{\mu\nu}\psi_{i_1} = 
\frac{3}{2}\bar{\psi}_{i_1}\psi_{i_1}\bar{\psi}_{i_2}\psi_{i_2}\bar{\psi}_{i_3}\psi_{i_3}
-\frac{3}{2}\bar{\psi}_{i_1}\psi_{i_1}\bar{\psi}_{i_2}\gamma^5\psi_{i_2}\bar{\psi}_{i_3}\gamma^5\psi_{i_3} -\\ 
&-\frac{3}{4}\bar{\psi}_{i_1}\gamma^5\psi_{i_1}\bar{\psi}_{i_2}\gamma_\mu\psi_{i_2}\bar{\psi}_{i_3}\gamma^5\gamma^\mu\psi_{i_3}
-\frac{9}{8}\bar{\psi}_{i_1}\psi_{i_1}\bar{\psi}_{i_2}\gamma_\mu\psi_{i_2}\bar{\psi}_{i_3}\gamma^\mu\psi_{i_3} 
-\frac{1}{8}\epsilon_{\kappa\lambda\mu\nu}\bar{\psi}_{i_1}\gamma^\kappa\psi_{i_1}\bar{\psi}_{i_2}\gamma^5\gamma^\lambda\psi_{i_2}\bar{\psi}_{i_3}\sigma^{\mu\nu}\psi_{i_3} - \\
& - \frac{i}{8}\epsilon_{\kappa\lambda\mu\nu}\bar{\psi}_{i_1}\gamma^5\psi_{i_1}\bar{\psi}_{i_2}\sigma^{\kappa\lambda}\psi_{i_2}\bar{\psi}_{i_3}\sigma^{\mu\nu}\psi_{i_3} 
-\frac{1}{4}\bar{\psi}_{i_1}\psi_{i_1}\bar{\psi}_{i_2}\sigma^{\mu\nu}\psi_{i_2}\bar{\psi}_{i_3}\sigma^{\mu\nu}\psi_{i_3}  
-\frac{9}{8}\bar{\psi}_{i_1}\psi_{i_1}\bar{\psi}_{i_2}\gamma^5\gamma_\mu\psi_{i_2}\bar{\psi}_{i_3}\gamma^5\gamma^\mu\psi_{i_3} + \\
&-\frac{i}{4}\bar{\psi}_{i_1}{\sigma_\kappa}^\lambda\psi_{i_1}\bar{\psi}_{i_2}{\sigma_\lambda}^\mu\psi_{i_2}\bar{\psi}_{i_3}{\sigma_\mu}^\kappa\psi_{i_3}
\end{split}
\end{equation}

\begin{equation}
\begin{split}
&\epsilon_{\kappa\lambda\mu\nu}\bar{\psi}_{i_1}\gamma^5\psi_{i_1}\bar{\psi}_{i_2}\sigma^{\kappa\lambda}\psi_{i_3}\bar{\psi}_{i_3}\sigma^{\mu\nu}\psi_{i_2} =  - 12i\bar{\psi}_{i_1}\psi_{i_1}\bar{\psi}_{i_2}\gamma^5\psi_{i_2}\bar{\psi}_{i_3}\gamma^5\psi_{i_3} 
+\frac{1}{2}\epsilon_{\kappa\lambda\mu\nu}\bar{\psi}_{i_1}\gamma^5\psi_{i_1}\bar{\psi}_{i_2}\sigma^{\kappa\lambda}\psi_{i_2}\bar{\psi}_{i_3}\sigma^{\mu\nu}\psi_{i_3}
\end{split}
\end{equation}

\begin{equation}
\begin{split}
&\epsilon_{\kappa\lambda\mu\nu}\bar{\psi}_{i_1}\gamma^5\psi_{i_2}\bar{\psi}_{i_2}\sigma^{\kappa\lambda}\psi_{i_1}\bar{\psi}_{i_3}\sigma^{\mu\nu}\psi_{i_3} = 
-\frac{1}{2}\epsilon_{\kappa\lambda\mu\nu}\bar{\psi}_{i_1}\gamma^5\psi_{i_1}\bar{\psi}_{i_2}\sigma^{\kappa\lambda}\psi_{i_2}\bar{\psi}_{i_3}\sigma^{\mu\nu}\psi_{i_3} + \\
&-i\bar{\psi}_{i_1}\psi_{i_1}\bar{\psi}_{i_2}\sigma_{\mu\nu}\psi_{i_2}\bar{\psi}_{i_3}\sigma^{\mu\nu}\psi_{i_3} 
 - i\epsilon_{\kappa\lambda\mu\nu}\bar{\psi}_{i_1}\gamma^\kappa\psi_{i_1}\bar{\psi}_{i_2}\gamma^5\gamma^\lambda\psi_{i_2}\bar{\psi}_{i_3}\sigma^{\mu\nu}\psi_{i_3} 
 + \bar{\psi}_{i_1}{\sigma_\kappa}^\lambda\psi_{i_1}\bar{\psi}_{i_2}{\sigma_\lambda}^\mu\psi_{i_2}\bar{\psi}_{i_3}{\sigma_\mu}^\kappa\psi_{i_3}
\end{split}
\end{equation}

\begin{equation}
\begin{split}
&\epsilon_{\kappa\lambda\mu\nu}\bar{\psi}_{i_1}\gamma^5\psi_{i_2}\bar{\psi}_{i_2}\sigma^{\kappa\lambda}\psi_{i_3}\bar{\psi}_{i_3}\sigma^{\mu\nu}\psi_{i_1} = 
\frac{9i}{2}\bar{\psi}_{i_1}\psi_{i_1}\bar{\psi}_{i_2}\gamma^5\psi_{i_2}\bar{\psi}_{i_3}\gamma^5\psi_{i_3} + \\
&+\frac{3i}{2}\bar{\psi}_{i_1}\psi_{i_1}\bar{\psi}_{i_2}\gamma^5\gamma_\mu\psi_{i_2}\bar{\psi}_{i_3}\gamma^5\gamma^\mu\psi_{i_3} 
+ \frac{3i}{2}\bar{\psi}_{i_1}\psi_{i_1}\bar{\psi}_{i_2}\psi_{i_2}\bar{\psi}_{i_3}\psi_{i_3}
+ i\bar{\psi}_{i_1}\psi_{i_1}\bar{\psi}_{i_2}\sigma_{\mu\nu}\psi_{i_2}\bar{\psi}_{i_3}\sigma^{\mu\nu}\psi_{i_3}  -\\
&-\frac{1}{2}\epsilon_{\kappa\lambda\mu\nu}\bar{\psi}_{i_1}\gamma^5\psi_{i_1}\bar{\psi}_{i_2}\sigma^{\kappa\lambda}\psi_{i_2}\bar{\psi}_{i_3}\sigma^{\mu\nu}\psi_{i_3}   
- \frac{3i}{2}\bar{\psi}_{i_1}\psi_{i_1}\bar{\psi}_{i_2}\gamma_\mu\psi_{i_2}\bar{\psi}_{i_3}\gamma^\mu\psi_{i_3}
\end{split}
\end{equation}

\begin{equation}
\begin{split}
&\bar{\psi}_{i_1}{\sigma_{\kappa}}^{\lambda}\psi_{i_1}\bar{\psi}_{i_2}{\sigma_{\lambda}}^{\mu}\psi_{i_3}\bar{\psi}_{i_3}{\sigma_{\mu}}^{\kappa}\psi_{i_2} = 
\bar{\psi}_{i_1}{\sigma_{\kappa}}^{\lambda}\psi_{i_2}\bar{\psi}_{i_2}{\sigma_{\lambda}}^{\mu}\psi_{i_1}\bar{\psi}_{i_3}{\sigma_{\mu}}^{\kappa}\psi_{i_3} = 
\bar{\psi}_{i_1}{\sigma_{\kappa}}^{\lambda}\psi_{i_3}\bar{\psi}_{i_2}{\sigma_{\lambda}}^{\mu}\psi_{i_3}\bar{\psi}_{i_3}{\sigma_{\mu}}^{\kappa}\psi_{i_1} = 0
\end{split}
\end{equation}

\begin{equation}
\begin{split}
&\bar{\psi}_{i_1}{\sigma_{\kappa}}^{\lambda}\psi_{i_2}\bar{\psi}_{i_2}{\sigma_{\lambda}}^{\mu}\psi_{i_3}\bar{\psi}_{i_3}{\sigma_{\mu}}^{\kappa}\psi_{i_1} = 
\frac{i}{4}\bar{\psi}_{i_1}\psi_{i_1}\bar{\psi}_{i_2}\sigma_{\mu\nu}\psi_{i_2}\bar{\psi}_{i_3}\sigma^{\mu\nu}\psi_{i_3} 
+ \frac{1}{8}\epsilon_{\kappa\lambda\mu\nu}\bar{\psi}_{i_1}\gamma^5\psi_{i_1}\bar{\psi}_{i_2}\sigma^{\kappa\lambda}\psi_{i_2}\bar{\psi}_{i_3}\sigma^{\mu\nu}\psi_{i_3}  -\\
& - \frac{9i}{2}\bar{\psi}_{i_1}\psi_{i_1}\bar{\psi}_{i_2}\gamma^5\psi_{i_2}\bar{\psi}_{i_3}\gamma^5\psi_{i_3}
 - \frac{3i}{2}\bar{\psi}_{i_1}\psi_{i_1}\bar{\psi}_{i_2}\psi_{i_2}\bar{\psi}_{i_3}\psi_{i_3}
\end{split}
\end{equation}

\begin{equation}
\begin{split}
&\bar{\psi}_{i_1}{\sigma_{\lambda}}^{\mu}\psi_{i_2}\bar{\psi}_{i_2}{\sigma_{\kappa}}^{\lambda}\psi_{i_3}\bar{\psi}_{i_3}{\sigma_{\mu}}^{\kappa}\psi_{i_1} = 
\frac{i}{4}\bar{\psi}_{i_1}\psi_{i_1}\bar{\psi}_{i_2}\sigma_{\mu\nu}\psi_{i_2}\bar{\psi}_{i_3}\sigma^{\mu\nu}\psi_{i_3} 
- \frac{1}{8}\epsilon_{\kappa\lambda\mu\nu}\bar{\psi}_{i_1}\gamma^5\psi_{i_1}\bar{\psi}_{i_2}\sigma^{\kappa\lambda}\psi_{i_2}\bar{\psi}_{i_3}\sigma^{\mu\nu}\psi_{i_3}  +\\
& + \frac{9i}{2}\bar{\psi}_{i_1}\psi_{i_1}\bar{\psi}_{i_2}\gamma^5\psi_{i_2}\bar{\psi}_{i_3}\gamma^5\psi_{i_3}
 + \frac{3i}{2}\bar{\psi}_{i_1}\psi_{i_1}\bar{\psi}_{i_2}\psi_{i_2}\bar{\psi}_{i_3}\psi_{i_3}
\end{split}
\end{equation}

\bibliographystyle{ieeetr}
\bibliography{NJL_effective_action} 

\begin{thebibliography}{10}

\bibitem{Svidz}
A.~Svidzinskyi, {\em Spatially inhomogeneous propblems in the theory of
  superconductivity}.
\newblock Moscow: Nauka, 1982.
\newblock \textit{(in Russian)}.

\bibitem{HolWald}
S.~Hollands and R.~M. Wald, ``Axiomatic quantum field theory in curved
  spacetime,'' {\em Communications in Mathematical Physics}, vol.~293,
  pp.~85--125, sep 2009.

\bibitem{NJL}
Y.~Nambu and G.~Jona-Lasinio, ``Dynamical model of elementary particles based
  on an analogy with superconductivity. {I},'' {\em Phys. Rev.}, vol.~122,
  pp.~345--358, Apr 1961.

\bibitem{Kikkawa1974}
K.~Kikkawa, ``Quantum corrections in superconductor models,'' {\em Progress of
  Theoretical Physics}, vol.~56, pp.~947--955, 09 1976.

\bibitem{Ebert1978}
D.~Ebert and H.~Reinhardt, ``Functional approach to nuclear field theory: A
  schematic model with pairing and particle-hole forces,'' {\em Nuclear Physics
  A}, vol.~298, no.~1, pp.~60--76, 1978.

\bibitem{Reinhardt1978}
H.~Reinhardt, ``Nuclear field theory,'' {\em Nuclear Physics A}, vol.~298,
  no.~1, pp.~77--92, 1978.

\bibitem{Volkov1984}
M.~K. Volkov, ``Meson lagrangians in a superconductor quark model,'' {\em
  Annals Phys.}, vol.~157, pp.~282--303, 1984.

\bibitem{Terazawa1977}
H.~Terazawa, Y.~Chikashige, and K.~Akama, ``Unified model of the
  {N}ambu-{J}ona-{L}asinio type for all elementary-particle forces,'' {\em
  Phys. Rev. D}, vol.~15, pp.~480--487, Jan 1977.

\bibitem{Jackiw1974}
R.~Jackiw, ``{Functional evaluation of the effective potential},'' {\em Phys.
  Rev. D}, vol.~9, p.~1686, 1974.

\bibitem{Cornwall1974}
J.~M. Cornwall, R.~Jackiw, and E.~Tomboulis, ``Effective action for composite
  operators,'' {\em Phys. Rev. D}, vol.~10, pp.~2428--2445, Oct 1974.

\bibitem{Gusynin1989a}
V.~Gusynin and V.~Miransky, ``Nonperturbative scale anomaly in gauge
  theories,'' {\em Zh. Eksp. Teor. Fiz}, vol.~95, pp.~410--427, 1989.

\bibitem{Gusynin1989b}
V.~P. Gusynin, V.~A. Kushnir, and V.~A. Miransky, ``Green's functions of
  composite operators and bound states in gauge theories,'' {\em Phys. Rev. D},
  vol.~39, pp.~2355--2367, Apr 1989.

\bibitem{Gusynin1991}
V.~Gusynin and V.~Miransky, ``On the effective action in field theories with
  dynamical symmetry breaking,'' {\em Modern Physics Letters A}, vol.~06, 11
  1991.

\bibitem{Gusynin1992}
V.~Gusynin and V.~Miransky, ``Effective action in the nambu-jona-lasinio gauge
  model,'' {\em Zh. Eksp. Teor. Fiz}, vol.~101, pp.~414--430, 1992.

\bibitem{Calzetta1988}
E.~Calzetta and B.~L. Hu, ``Nonequilibrium quantum fields: Closed-time-path
  effective action, {W}igner function, and {B}oltzmann equation,'' {\em Phys.
  Rev. D}, vol.~37, pp.~2878--2900, May 1988.

\bibitem{Anastopoulos2023}
C.~Anastopoulos, B.-L. Hu, and K.~Savvidou, ``Towards a field-theory based
  relativistic quantum information,'' {\em Journal of Physics: Conference
  Series}, vol.~2533, p.~012004, jun 2023.

\bibitem{Fierz1937}
M.~Fierz, ``Zur fermischen theorie des $\beta$-zerfalls,'' {\em Zeitschrift
  f{\"u}r Physik}, vol.~104, pp.~553--565, 1937.

\bibitem{itzykson2012}
C.~Itzykson and J.~Zuber, {\em Quantum Field Theory}.
\newblock Dover Books on Physics, Dover Publications, 2012.

\bibitem{Weinberg1996}
S.~Weinberg, {\em {The quantum theory of fields. Vol. 2: Modern applications}}.
\newblock Cambridge University Press, 8 2013.

\bibitem{Zavyalov1990}
O.~I. Zavyalov, {\em {Renormalized quantum field theory}}.
\newblock 1990.

\bibitem{StreaterWightman1964}
{R. F. Streater, A. S. Wightman}, {\em PCT, spin and statistics, and all that}.
\newblock W.A. Benjamin, New York, 1964.

\bibitem{Haag1955}
R.~Haag, ``{On quantum field theories},'' {\em Kong. Dan. Vid. Sel. Mat. Fys.
  Med.}, vol.~29N12, pp.~1--37, 1955.

\bibitem{Hepp1966}
K.~Hepp, ``{Proof of the Bogoliubov-Parasiuk theorem on renormalization},''
  {\em Communications in Mathematical Physics}, vol.~2, no.~4, pp.~301 -- 326,
  1966.

\bibitem{Zimmermann1969}
W.~Zimmermann, ``{Convergence of Bogoliubov's method of renormalization in
  momentum space},'' {\em Communications in Mathematical Physics}, vol.~15,
  no.~3, pp.~208 -- 234, 1969.

\bibitem{Weyl1946}
H.~Weyl, {\em The classical groups: their invariants and representations}.
\newblock No.~1, Princeton university press, 1946.

\bibitem{Kraft1996}
H.~Kraft and C.~Procesi, {\em A Primer of Classical Invariant Theory:
  Preliminary Version}.
\newblock 1996.
\newblock preprint.

\bibitem{Haber2021}
H.~Haber, ``Useful relations among the generators in the defining and adjoint
  representations of {SU}(n),'' {\em {SciPost} Physics Lecture Notes}, jan
  2021.

\bibitem{PBPal2007}
P.~B. Pal, ``Representation-independent manipulations with dirac matrices and
  spinors,'' 2007.

\bibitem{Lowenstein1975}
J.~Lowenstein and W.~Zimmermann, ``On the formulation of theories with
  zero-mass propagators,'' {\em Nuclear Physics B}, vol.~86, no.~1,
  pp.~77--103, 1975.

\bibitem{Lowenstein1976}
J.~H. Lowenstein, ``{Convergence Theorems for Renormalized Feynman Integrals
  with Zero-Mass Propagators},'' {\em Commun. Math. Phys.}, vol.~47,
  pp.~53--68, 1976.

\bibitem{Blaschke2013}
D.~N. Blaschke, F.~Gieres, F.~Heindl, M.~Schweda, and M.~Wohlgenannt, ``{BPHZ}
  renormalization and its application to non-commutative field theory,'' {\em
  The European Physical Journal C}, vol.~73, sep 2013.

\bibitem{Kondo2003}
K.~I. Kondo, ``Vacuum condensates, effective gluon mass and color confinement
  in a new reformulation of qcd,'' 2003.

\bibitem{Ioffe2002}
B.~L. Ioffe, ``{Condensates in quantum chromodynamics},'' {\em Phys. Atom.
  Nucl.}, vol.~66, pp.~30--43, 2003.

\bibitem{Kutnii2023}
S.~Kutnii, ``Dirac: a command-line $\gamma$-matrix calculator,'' {\em Computer
  Physics Communications}, vol.~294, to be published in 2024.

\end{thebibliography}

\end{document}